\documentclass[onecolumn]{IEEEtran}

\usepackage{amsfonts}                                                          
\usepackage{epsfig}
\usepackage{amsthm}
\usepackage{graphicx}        
\usepackage{latexsym}
\usepackage{amssymb}
\usepackage{amsmath}
\usepackage{parskip}
\usepackage{multirow}
\usepackage{cite}
\usepackage{mathtools}
\usepackage{epstopdf}
\usepackage{etoolbox}
\usepackage{array}
\usepackage{setspace}
\usepackage{mathptmx}
\usepackage[bottom]{footmisc}
\usepackage{tikz}

\usetikzlibrary{decorations.pathreplacing}

\usepackage{verbatim}
\usepackage{calrsfs}

\DeclareMathAlphabet{\pazocal}{OMS}{zplm}{m}{n}

\newcommand{\mb}{\mathbb}

\let\bbordermatrix\bordermatrix
\patchcmd{\bbordermatrix}{8.75}{4.75}{}{}
\patchcmd{\bbordermatrix}{\left(}{\left[}{}{}
\patchcmd{\bbordermatrix}{\right)}{\right]}{}{}

\newcommand{\sr}{\stackrel}

\newcommand{\rar}{\rightarrow}

\newcommand{\tri}{\sr{\triangle}{=}}

\newcommand{\be}{\begin{equation}}
\newcommand{\ee}{\end{equation}}
\newcommand{\bea}{\begin{eqnarray}}
\newcommand{\eea}{\end{eqnarray}}
\newcommand{\bes}{\begin{eqnarray*}}
\newcommand{\ees}{\end{eqnarray*}}
\newcommand{\bce}{\begin{center}}
\newcommand{\ece}{\end{center}}
\newcommand{\beae}{\begin{IEEEeqnarray}{rCl}}
\newcommand{\eeae}{\end{IEEEeqnarray}}

\def\VR{\kern-\arraycolsep\strut\vrule &\kern-\arraycolsep}
\def\vr{\kern-\arraycolsep & \kern-\arraycolsep}

\newcommand{\ben}{\begin{enumerate}}
\newcommand{\een}{\end{enumerate}}

\newcommand{\hso}{\hspace{.1in}}
\newcommand{\hst}{\hspace{.2in}}

\newcommand{\noi}{\noindent}

\newtheorem{theorem}{Theorem}[section]

\newtheorem{remark}{Remark}[section]
\newtheorem{corollary}{Corollary}[section]

\newtheorem{definition}{Definition}[section]
\newtheorem{lemma}{Lemma}[section]

\onehalfspace

\begin{document}


\title{Information Structures for Feedback Capacity of Channels with Memory and Transmission Cost: Stochastic Optimal Control \& Variational Equalities-Part I}


\author{Christos~K.~Kourtellaris  and Charalambos~D.~Charalambous
\thanks{The authors are with the Department of Electrical and Computer Engineering, University of Cyprus, 75 Kallipoleos Avenue, P.O. Box 20537, Nicosia, 1678, Cyprus, e-mail: $\{kourtellaris.christos,chadcha@ucy.ac.cy\}$.
This work was financially supported by a medium size University of Cyprus grant entitled ``DIMITRIS", and was presented in part at the 2016 IEEE International Symposium on Information Theory.}}


\maketitle

\begin{abstract}
The Finite Transmission Feedback Information (FTFI) capacity is characterized for any class of  channel conditional distributions $\big\{{\bf P}_{B_i|B^{i-1}, A_i} :i=0, 1, \ldots, n\big\}$ and $\big\{  {\bf P}_{B_i|B_{i-M}^{i-1}, A_i} :i=0, 1, \ldots, n\big\}$, where $M$  is the memory of the channel, $B^n \tri \{B_j: j=\ldots, 0,1, \ldots, n\}$ are the channel outputs and  $A^n\tri \{A_j: j=\ldots, 0,1, \ldots, n\}$ are the channel inputs. The characterizations of FTFI capacity, are obtained by first identifying the information structures  of the optimal channel input  conditional distributions ${\cal  P}_{[0,n]} \tri \big\{  {\bf P}_{A_i|A^{i-1}, B^{i-1}}:  i=0, \ldots, n\big\}$,  which maximize directed information 
\begin{align}
C_{A^n \rar B^n}^{FB} \tri \sup_{   {\cal P}_{[0,n]} }    I(A^n\rar B^n),  \hst I(A^n \rar B^n) \tri \sum_{i=0}^n I(A^i;B_i|B^{i-1}) . \nonumber 
\end{align}
The main theorem states,  for any channel with memory $M$, the optimal channel input conditional distributions occur in the subset satisfying conditional independence $\sr{\circ}{\cal P}_{[0,n]}\tri \big\{ {\bf P}_{A_i|A^{i-1}, B^{i-1}}= {\bf P}_{A_i|B_{i-M}^{i-1}}:  i=1, \ldots, n\big\}$, and the characterization of FTFI capacity is given by 
\begin{align}
C_{A^n \rar B^n}^{FB, M} \tri \sup_{   \sr{\circ}{\cal P}_{[0,n]} }  \sum_{i=0}^n  I(A_i; B_i|B_{i-M}^{i-1}) .  \nonumber %
\end{align}
 Similar conclusions are derived for problems with transmission cost constraints  of the form 
 $\frac{1}{n+1} {\bf E}\Big\{\sum_{i=0}^n \gamma_i(A_i, T^iB^{n-1})\Big\} \leq \kappa,$  $ \kappa >0$, 
where  $\big\{\gamma_i(A_i, T^i B^{n-1}):i=0,1, \ldots, n\big\}$ is any class of multi-letter functions such that   $T^iB^{n-1} = \{B_{i-1}, B_{i-2}, \ldots,  B_{i-K}\}$ or $T^iB^{n-1} = \{B^{i-1} \}$, for $i=0, \ldots, n$ and $K$ a nonnegative integer.

The methodology utilizes stochastic optimal control theory, to identify the control process, the controlled process, and a variational equality of directed information, to derive upper bounds on $I(A^n \rar B^n)$, which are achievable over specific subsets of channel input conditional distributions ${\cal P}_{[0,n]}$, which are characterized by conditional independence. For channels with limited memory, this implies the transition probabilities of the channel output process are also  of limited memory.

For any of the above classes of channel distributions and transmission cost functions, a direct analogy, in terms of conditional independence,  of the characterizations of FTFI capacity and Shannon's  capacity formulae of Memoryless Channels is identified.    
\end{abstract}

\section{Introduction}
\label{intro}
Feedback  capacity of channel conditional distributions, $\big\{ {\bf P}_{B_i|B^{i-1}, A^{i}}:i=0, 1, \ldots, n\big\}$, where $a^n \tri \{\ldots, a_{-1}, a_0, a_1, \ldots, a_n\} \in {\mathbb A}^n$,  $b^n \tri \{\ldots, b_{i-1}, b_0, b_1, \ldots, b_n\} \in {\mathbb B}^n$, are the channel input and output sequences, respectively,  is often defined by  maximizing  directed information \cite{marko1973,massey1990},  $I(A^n\rar  B^n)$, from channel input sequences 
$a^n  \in {\mathbb A}^n$ to channel  output sequences  $b^n \in {\mathbb B}^n$,  over an admissible set of channel input conditional distributions,    ${\cal P}_{[0,n]} \tri \big\{ {\bf P}_{A_i|A^{i-1}, B^{i-1}}:i=0, 1, \ldots, n\big\}$ (with feedback),      as follows. 
\begin{align}
C_{A^\infty\rar  B^\infty}^{FB} \tri & \liminf_{n \longrightarrow \infty} \frac{1}{n+1} C_{A^n \rar  B^n}^{FB}, \hst C_{A^n \rar B^n}^{FB} \tri   \sup_{    {\cal P}_{[0,n]} }I(A^n\rar B^n), \label{cap_fb_1}\\
 I(A^n\rar B^n) \tri & \sum_{i=0}^n I(A^i;B_i|B^{i-1})= \sum_{i=0}^n {\bf E}_\mu \Big\{ \log \Big( \frac{ d{\bf P}_{B_i|B^{i-1}, A^i}(\cdot|B^{i-1}, A^i)}{d{\bf P}_{B_i|B^{i-1}}(\cdot|B^{i-1})}(B_i)\Big)\Big\} \label{intro_fbc1a}
  \end{align}
 where for each $i$,   ${\bf P}_{B_i|B^{i-1}}$, is the conditional probability distribution of the channel output process, and ${\bf E}_\mu\{\cdot\}$ denotes expectation with respect to the joint distribution  ${\bf P}_{A^i, B^i}$, for $i=0, \ldots, n$, for a fixed initial distribution\footnote{The subscript notation on probability distributions and expectation, i.e., ${\bf P}_\mu$ and ${\bf E}_\mu\{\cdot\}$ is often omitted because it is clear from the context.} ${\bf P}_{A^{-1}, B^{-1}}\equiv \mu(da^{-1}, db^{-1})$ of the initial data $\{(A^{-1}, B^{-1})=(a^{-1}, b^{-1})\}$.  In (\ref{cap_fb_1}),  $\liminf$ can be replaced by $\lim$ if it  exists and it is finite.   For finite alphabet spaces sufficient conditions are identified in \cite{permuter-weissman-goldsmith2009}. Moreover, for finite alphabet spaces  $\sup$ can be replaced by maximum, since probability mass functions can be viewed as closed and bounded subsets of finite-dimensional spaces. However, for countable or abstract alphabet spaces it is more difficult, and often requires an analysis using the topology of weak convergence of probability distributions, because information theoretic measures are not necessarily continuous with respect to pointwise convergence of probability distributions, and showing compactness of the set of distributions is quite involved.   \\
It is shown in  \cite{kramer2003,tatikonda-mitter2009,permuter-weissman-goldsmith2009}, using tools from \cite{dobrushin1959,pinsker1964,gallager1968,blahut1987,cover-thomas2006,ihara1993,verdu-han1994,han2003},  under appropriate conditions which include abstract alphabets,  that the quantity $C_{A^\infty \rar B^\infty}^{FB}$ is the supremum of all achievable rates of a sequence of feedback codes  $\{(n, { M}_n, \epsilon_n):n=0, 1, \dots\}$, defined as follows.   \\
(a)  A set of uniformly distributed source messages ${\cal M}_n \tri \{ 1,  \ldots, M_n\}$ and a set of encoding strategies,  mapping source messages  into channel inputs of block length $(n+1)$, defined by
\begin{align}
{\cal E}_{[0,n]}^{FB} \triangleq & \Big\{g_i: {\cal M}_n \times {\mb A}^{i-1} \times {\mb B}^{i-1}  \longmapsto {\mb A}_i: i=0, \ldots, n:  \;  a_0=g_0(w), a_1=g_1(w,a_0,b_0),\ldots, a_n=g_n(w, a^{n-1}, b^{n-1}), \hso  w\in {\cal M}_n  \Big\}, \; n=0, 1, \ldots. \label{block-code-nf-non}
\end{align}
The codeword for any $w \in {\cal M}_n$  is $u_w\in{\mb A}^n$, $u_w=(g_0(w), g_1(w, a_0, b_0)
,\dots,g_n(w, a^{n-1}, b^{n-1}))$, and ${\cal C}_n=( u_1,u_2,\dots,u_{{M}_n})$ is  the code for the message set ${\cal M}_n$, and $\{A^{-1}, B^{-1}\}=\{\emptyset\}$.  In general, the code  may depend on the initial data, depending on the convention, i.e.,  $(A^{-1}, B^{-1})=(a^{-1}, b^{-1})$, which are often known to the encoder and decoder. 
\\
(b)  Decoder measurable mappings $d_{0,n}:{\mb B}^n\longmapsto {\cal M}_n$,  such that the average
probability of decoding error satisfies\footnote{The superscript on expectation, i.e., ${\bf P}^g$ indicates the dependence of the distribution on the encoding strategies.} 
\begin{align}
{\bf P}_e^{(n)} \triangleq \frac{1}{M_n} \sum_{w \in {\cal M}_n} {\bf  P}^g \Big\{d_{0,n}(B^{n}) \neq w |  W=w\Big\}\equiv {\bf  P}^g\Big\{d_{0,n}(B^n) \neq W \Big\} \leq \epsilon_n, \hst w \in {\cal M}_n \nonumber
\end{align}
and the decoder may also assume knowledge of the initial data.\\
The coding rate or transmission rate is defined by  $r_n\triangleq \frac{1}{n+1} \log M_n$.
A rate $R$ is said to be an achievable rate, if there exists  a  code sequence satisfying
$\lim_{n\longrightarrow\infty} {\epsilon}_n=0$ and $\liminf_{n \longrightarrow\infty}\frac{1}{n+1}\log{{M}_n}\geq R$. The feedback capacity is supremum of all achievable rates, i.e., defined by $C\triangleq \sup \{R: R \: \mbox{is achievable}\}$.

The underlying  assumption for $C_{A^\infty\rar  B^\infty}^{FB}$ to correspond to feedback capacity is that the source process $\big\{X_i: i=0, \ldots, \big\}$   to be encoded and transmitted over the channel has finite entropy  rate, and  satisfies  the following conditional independence \cite{massey1990}. 
\begin{align}
{\bf P}_{B_i|B^{i-1}, A^i, X^k}={\bf P}_{B_i|B^{i-1}, A^i} \hso   \forall k \in \{0,1, \ldots, n\},\hso i=0, \ldots, n \label{CI_Massey_N} 
\end{align}

Coding theorems for  channels with memory with and without feedback are developed extensively over the years, in an anthology of papers, such as,   \cite{dobrushin1959,pinsker1964,gallager1968,blahut1987,cover-thomas2006,ihara1993,verdu-han1994,kramer1998,han2003,kramer2003,tatikonda-mitter2009,permuter-weissman-goldsmith2009,gamal-kim2011}, in three direction.  For  jointly stationary ergodic  processes, for information stable processes, and for arbitrary nonstationary and nonergodic processes. Since many of the coding theorems presented in the above references are either directly applicable or applicable subject to the assumptions imposed in these references,  the main emphasis of the current investigation is on the characterizations of FTFI capacity, for different channels with transmission cost.

Similarly, feedback capacity with transmission cost is defined by 
\begin{align}
C_{A^\infty\rar  B^\infty}^{FB}(\kappa) \tri & \liminf_{n \longrightarrow \infty} \frac{1}{n+1} C_{A^n \rar  B^n}^{FB}(\kappa), \hst C_{A^n \rar B^n}^{FB}(\kappa) \tri   \sup_{    {\cal P}_{[0,n]}(\kappa) }I(A^n\rar B^n) \label{cap_fb_1_TC}\\
{\cal P}_{[0,n]}(\kappa)\tri & \Big\{ {\bf P}_{A_i|A^{i-1}, B^{i-1}},  i=1, \ldots, n:  \frac{1}{n+1} {\bf E}\Big(\sum_{i=0}^n \gamma_i(T^iA^n, T^iB^{n-1})\Big) \leq \kappa\Big\} \label{cap_fb_3}
\end{align}
where for      each  $i$, $T^ia^n \subseteq \{\ldots, a_{-1}, a_0, a_1, \ldots, a_i\},  T^ib^{n-1} \subseteq \{\ldots, b_{-1}, b_0, b_1, \ldots, b_{i-1}\}$, for $i=0, \ldots, n$, and these are either fixed or nondecreasing with $i$, for        $i=0,1, \ldots, n$. 

The hardness of such extremum problems of capacity, and in general, of other similar problems of information theory, is attributed to the form of the directed information density or sample  pay-off functional, defined by 
\bea
\iota_{A^n \rar B^n}(a^n,b^n)\tri \sum_{i=0}^n \log \Big( \frac{ d{\bf P}_{B_i|B^{i-1}, A^i}(\cdot|b^{i-1}, a^i)}{d{\bf P}_{B_i|B^{i-1}}(\cdot|b^{i-1})}(b_i)\Big) \label{DI_Den_1}
\eea
which in not fixed. Rather, the pay-off $\iota_{A^n \rar B^n}(a^n,b^n)$  depends on  the channel output conditional probabilities $\big\{{\bf P}_{B_i|B^{i-1}}(db_i|b^{i-1}): i=0, \ldots, n\big\}$, which in turn  depends on the the channel input conditional distributions $\big\{{\bf P}_{A_i|A^{i-1}, B^{i-1}}(da_i|a^{i-1}, b^{i-1}): i=0, \ldots, n\big\} \in {\cal P}_{[0,n]}(\kappa)$, chosen to maximize the expectation ${\bf E} \big\{\iota_{A^n \rar B^n}(A^n,B^n)\big\}$. This means, given a specific channel conditional distribution and a transmission cost function, the information structure of the  channel input conditional distribution denoted by $ {\cal I}_i^{P}\subseteq \{a^{i-1}, b^{i-1}\}, i=0, \ldots, n$, which maximizes directed information (i.e., the dependence of the optimal channel input conditional distribution on past information),  needs to be identified, and then used to obtain the characterizations of the Finite Transmission Feedback Information (FTFI) capacity, $C_{A^n\rar  B^n}^{FB}(\kappa)$,  and feedback    capacity  $C_{A^\infty\rar  B^\infty}^{FB}(\kappa)$.

For memoryless stationary channels (such as, Discrete Memoryless Channels (DMCs)),  described by ${\bf P}_{B_i|B^{i-1}, A^i}={\bf P}_{B_i|A_i}\equiv {\bf P}_{B|A}, i=0,1, \ldots, n $,  without feedback  (with transmission cost constraints if necessary), Shannon \cite{shannon1948}    characterized channel capacity by the well-known two-letter  formulae
\begin{align}
C \tri \max_{ {\bf P}_A } I(A; B) =\max_{ {\bf P}_A } {\bf E} \Big\{ \log \Big( \frac{ d{\bf P}_{A , B}(\cdot, \cdot)}{ d({\bf P}_{A}(\cdot)\times {\bf P}_{B}(\cdot))}(A, B)\Big)\Big\}  \label{DMC_Shannon_1}
\end{align}
where ${\bf P}_{A , B}(da,db)={\bf P}_{ B|A}(db|a)\otimes {\bf P}_A(da)$ is the joint distribution,  ${\bf P}_{ B|A}(db|a)$ is the channel conditional distribution, ${\bf P}_{A}(da)$ is the channel input distribution, ${\bf P}_{B}(db) =\int {\bf P}_{ B|A}(db|a)\otimes {\bf P}_A(da)$ is the channel output distribution, and ${\bf E}\big\{\cdot \big\}$ denotes expectation with respect to ${\bf P}_{A , B}$. \\
This characterization is often obtained by identifying the  information structures of optimal channel input distributions, via  the upper bound  
\bea
C_{A^n ; B^n}\tri \max_{{\bf P}_{A^n}} I(A^n;B^n) \leq \; \max_{{\bf P}_{A_i}, i=0, \ldots, n} \sum_{i=0}^n I(A_i;B_i) \leq (n+1) C \label{cap_nf_c}
\eea
since this bound is achievable, when the  channel input distribution  satisfies conditional independence $
{\bf P}_{A_i|A^{i-1}}(da_i|a^{i-1})={\bf P}_{A_i}(da_i),         i=0, 1, \ldots, n$, and moreover $C$ is obtained, when  $\{A_i:i =0, 1, \ldots, \}$ is identically distributed, which then implies  the  joint process $\{(A_i, B_i): i=0,1, \ldots, \}$ is independent and identically distributed, and $I(A^n;B^n)=(n+1)I(A; B)$. \\
For memoryless stationary channels  with feedback, the characterization of feedback capacity, denoted by $C^{FB}$, is shown by Shannon and subsequently Dobrushin\cite{dobrushin1958} to correspond to the  capacity without feedback, i.e., $C^{FB}=C$. This fundamental formulae  is often shown  by first  applying  the converse to the coding theorem,  to show that feedback does not increase capacity (see \cite{cover-pombra1989} for discussion on this subtle issue), which then implies
\begin{align}
 {\bf P}_{A_i|A^{i-1}, B^{i-1}}(da_i|a^{i-1}, b^{i-1})={\bf P}_{A_i}(da_i),        \hso  i=0, 1, \ldots, n \label{CI_DMC}
\end{align} 
   and $C^{FB}=C$ is obtained if  $\{A_i:i =0, 1, \ldots, \}$ is identically distributed. That is, since feedback does not increase capacity, then mutual information and directed information are identical, in view of (\ref{CI_DMC}). However, as pointed out elegantly by Massey \cite{massey1990}, for channels with feedback it will be a mistake to use the arguments in (\ref{cap_nf_c}) to derive $C^{FB}$. 
    The  conditional independence condition (\ref{CI_DMC}) implies that the {\it Information Structure} of the maximizing channel input distributions is  the {\it Null Set}. \\  

\noi The methodology developed in this paper, establishes a direct analogy between the conditional independence properties (\ref{CI_DMC}) of capacity achieving channel input distributions of memoryless channels and corresponding properties for channels with memory and feedback.  
To this date,  no such systematic methodology is developed in the literature, to determine the information structure of  optimal channel input distributions, which maximize directed information $I(A^n\rar B^n)$, via achievable upper bounds over subsets of channel input conditional distributions  $\overline{\cal P}_{[0,n]} \subseteq  {\cal P}_{[0,n]}(\kappa)$,  which satisfy conditional independence,  and to  characterize the corresponding FTFI capacity and feedback capacity.
\begin{figure}
  \centering
    \includegraphics[width=0.75\textwidth]{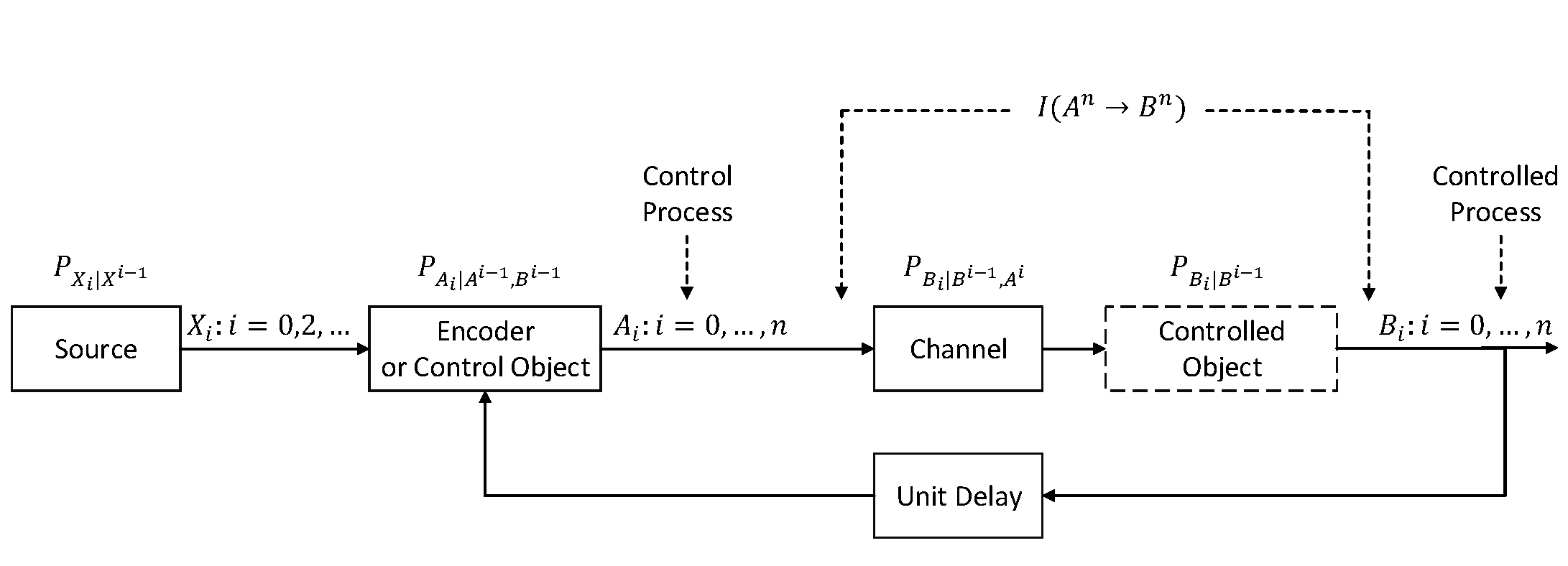}
    \label{figurefig3}
      \caption{Communication block diagram and its analogy to stochastic optimal control.}
\end{figure}

In this first part,  of a two-part investigation, the main objective is to  
\begin{description}
\item develop a methodology to identify  information structures of optimal channel input conditional distributions,  for channels with memory, with and without transmission cost, of extremum problems defined by (\ref{cap_fb_1}) and (\ref{cap_fb_1_TC}), and to characterize the corresponding FTFI capacity and feedback capacity. 
\end{description}
This is addressed  by  utilizing connections between stochastic optimal control and information theoretic concepts, as follows. \\
The theory of stochastic optimal control is linked to the identification of information structures of optimal channel input conditional distributions and to the characterization of FTFI capacity, by first establishing the  the following analogy.
\begin{description}
\item The information measure $I(A^n \rar B^n)$ is the pay-off; 

\item the channel output process $\{B_i: i=0, 1, \ldots, n\}$ is the controlled process;

\item the channel input process $\{A_i: i=0,1, \ldots, n\}$ is the control process.
\end{description}
Indeed, as depicted in Fig.\ref{figurefig3}, the channel output process $\{B_i: i=0, 1, \ldots, n\}$   is controlled, by controlling its conditional probability distribution  $\big\{ {\bf P}_{B_i|B^{i-1}}(db_i|b^{i-1}): i=0, \ldots, n\big\}$ via the choice of the   transition probability distribution $\big\{ {\bf P}_{A_i|A^{i-1}, B^{i-1}}(da_i| a^{i-1}, b^{i-1}): i=0, \ldots, n\big\}\in {\cal P}_{[0,n]}$ called the control object.\\
As in any stochastic optimal control problem, given a channel distribution, the distribution of the initial data, and a transmission cost function, the main objective is to determine the controlled object conditional distribution, the control object conditional distribution, and  the functional dependence of the pay-off on these  objects. 

However,  unlike classical stochastic optimal control theory, the directed information density pay-off (i.e., (\ref{DI_Den_1})), depends nonlinearly on the channel output conditional  distributions $\big\{ {\bf P}_{B_i|B^{i-1}}(db_i|b^{i-1}): i=0, \ldots, n\big\}$, induced by the control objects,   a variational equality   is linked to tight upper bounds on directed information $I(A^n\rar B^n)$,  which together with the stochastic optimal control analogy,  are shown to be  achievable over specific subsets of the control objects. These achievable bounds  depend   on the structural properties of the channel conditional distributions and the transmission cost functions.

 The methodology   is based on a two-step procedure, as follows. Given a class of channel conditional distributions, the distribution of the initial data, and a class of transmission cost functions, any candidate of   the optimal channel input conditional distribution or control object,  which maximizes $I(A^n \rar B^n)$  is shown to satisfy the following conditional independence.
\begin{align}
& {\bf P}_{A_i|A^{i-1}, B^{i-1}}(da_i|a^{i-1}, b^{i-1})={\bf P}(da_i| {\cal  I}_i^{{\bf P}^*}) \equiv     {\mb P}\Big\{A_i \in da_i| {\cal  I}_i^{{\bf P}^*}\Big\},  \hso  {\cal  I}_i^{{\bf P}^*} \subseteq \big\{a^{i-1}, b^{i-1}\big\}, \hso i=0,1, \ldots, n, \label{CI_COND} \\
&{\cal  I}_i^{{\bf P}^*} \tri \:  \mbox{Information Structure of optimal channel input  distributions which maximizes $I(A^n \rar B^n)$ for $i=0,1,\ldots,n$}.
\end{align}
Moreover,  the information structure ${\cal G}_i^{{\bf P}^*}, i=0,1, \ldots, n$,  is specified by the memory of the channel conditional distribution, and the dependence of the transmission cost function on the channel input and output symbols. Consequently, the dependence of the joint distribution of $\{(A_i, B_i): i=0, \ldots, n\}$, the conditional distribution of the channel output process $\{B_i: i=0, \ldots, n\}$, i.e, $\big\{ {\bf P}_{B_i|B^{i-1}}(db_i|b^{i-1}): i=0, \ldots, n\big\}$ and the directed information density $\iota_{A^n \rar B^n}(A^n,B^n)$, on the control object, is determined,  and the  characterization of FTFI capacity is obtained. \\
The characterization of feedback capacity is obtained from  the per unit time limiting version of the characterization of the FTFI capacity. 

These structural properties of channel input distribution, which  maximize directed information settle various open problems in Shannon's information theory, which include the role of feedback signals to control, via the control process (channel input), the controlled process (channel output process), and the design of encoders which achieve the characterizations of FTFI capacity and capacity. 

Indeed, in the second part of this two-part  investigation \cite{charalambous-kourtellarisIT2015_Part_2}, and based on these structural properties, a  methodology is developed  to realize  optimal channel input conditional distributions,  by information lossless randomized strategies (driven by uniform Random Variables on $[0,1]$, which can generate any distribution), and to construct encoders, which achieve the characterizations of FTFI capacity and feedback capacity. Applications of the results of this first part,  to various channel models, which include  Multiple-Input Multiple Output (MIMO) Gaussian Channel Models with memory,    are found in the second part of this investigation. In this part, we give an illustrative simple example to clarify the importance of information structures of optimal channel input distributions, in reduction of computation complexity, and to indicate the analogy to Shannon's  two-letter capacity formulae of DMCs.

\subsection{Literature Review of Feedback Capacity of Channels with Memory}
Although, in this paper we do not treat channels with memory dependence on past channel input symbols, for completeness we review such literature, and we discuss possibly extensions of our methodology to such channels at the end of the paper.

\noi  Cover and Pombra \cite{cover-pombra1989}  (see also Ihara \cite{ihara1993}) characterized the feedback capacity of  non-stationary  Additive Gaussian Noise (AGN) channels  with memory, defined by 
\begin{align}
 B_i=A_i+Z_i, \hst i=0,1, \ldots, n, \hst \frac{1}{n+1} \sum_{i=0}^{n} {\bf E} \big\{ |A_i|^2\big\} \leq \kappa, \hso \kappa \in [0,\infty) \label{c-p1989}    
\end{align}
where $\{Z_i : i=0,1, \ldots, n\}$ is a real-valued (scalar) jointly non-stationary Gaussian process, denoted by  $N(\mu_{Z^n}, K_{Z^n})$, and ``$A^n$ is causally related to $Z^n$'' defined by\footnote{\cite{cover-pombra1989}, page 39, above Lemma~5.}   ${\bf P}_{A^n, Z^n}(da^n, dz^n)= \otimes_{i=0}^n  {\bf P}_{A_i|A^{i-1}, Z^{i-1}}(da_i|a^{i-1}, z^{i-1}) \otimes   {\bf P}_{Z^n}(dz^n)$.
The authors in  \cite{cover-pombra1989}  characterized the capacity of this non-stationary AGN channel,
 by first characterizing the FTFI   capacity   formulae \footnote{The methodology in \cite{cover-pombra1989} utilizes the converse coding theorem to obtain an upper bound on the entropy $H(B^n)$, by restricting  $\{A_i: i=0, \ldots, n\}$ to a Gaussian process.}  via the expression
\begin{align}
 C_{0,n}^{FB, CP}(\kappa) \tri \max_{  \Big\{ (\Gamma, K_{V^n}): \frac{1}{n+1} {\bf E} \big\{tr\big(A^n (A^n)^T\big)\big\}\leq \kappa, \hso A^n= \Gamma Z^n + V^n \Big\} }  H(B^n) - H(Z^n ) \label{cp1989}
  \end{align}
 where $V^n \tri \{V_i: i=0, 1, \ldots,n\}$ is a Gaussian process $N( { 0}, K_{{V}^n})$,  orthogonal to $Z^{n}\tri \{Z_i: i=0, \ldots, n\}$, and  $\Gamma$ is lower diagonal time-varying matrix with  deterministic entries.  The feedback capacity is given by \cite{cover-pombra1989} $C^{FB, CP}(\kappa) \tri \lim_{n \longrightarrow \infty} \frac{1}{n+1} C_{0,n}^{FB, CP}(\kappa)$.\\ 
Kim \cite{kim2010} revisited the stationary  version of feedback capacity characterization of the Cover and Pombra AGN channel, and utilized frequency domain methods, and  their relations to scalar Riccati equations, and showed that if the noise power spectral density corresponds to a stationary Gaussian autoregressive moving-average model of order $K$, then a
 $K-$dimensional generalization of the Schalkwijk-Kailath \cite{schalkwijk-kailath1966} coding scheme achieves feedback capacity. Yang, Kavcic, and Tatikonda \cite{yang-kavcic-tatikonda2007} analyzed the feedback capacity of stationary AGN channels,  re-visited the Cover and Pombra AGN channel, and proposed solution methods based on dynamic programming, to perform the optimization in (\ref{cp1989}). 
Butman \cite{butman1969,butman1976} evaluated the performance of linear feedback schemes for AGN channels, when the noise is described by an autoregressive moving average model. A  historical account regarding Gaussian channels with memory and feedback, related to the the Cover and Pombra \cite{cover-pombra1989} AGN channel, is found in \cite{kim2010}.   \\  
Recently, for finite alphabet channels with memory and feedback,  expressions of feedback capacity are derived for the trapdoor channel by Permuter, Cuff, Van Roy and Tsachy \cite{permuter-cuff-roy-weissman2010}, for the Ising Channel by Elishco and Permuter \cite{elishco-permuter2014}, for the Post$(a,b)$ channel by Permuter, Asnani and Tsachy \cite{permuter-asnani-weissman2013},  all without  transmission cost constraints, and in  \cite{kourtellaris-charalambous2015}   for the BSSC($\alpha,\beta)$   with and without feedback and transmission cost.  Tatikonda, Yang and Kavcic \cite{yang-kavcic-tatikonda2005} showed  that if the input to the channel and the channel state are related by  a one-to-one mapping,  and the channel assumes a specific structure, specifically,  $\big\{{\bf P}_{B_i|A_{i}, A_{i-1}}: i=0, \ldots, n\big\}$,  then dynamic programming can be used to compute the  feedback capacity expression given in \cite{yang-kavcic-tatikonda2005}.  Chen and Berger  \cite{chen-berger2005} analyzed  the Unit Memory Channel Output (UMCO) channel  $\big\{{\bf P}_{B_i|B_{i-1}, A_i}: i=0, \ldots, n\}$, under the assumption that the optimal channel input distribution is $\big\{{\bf P}_{A_i|B_{i-1}}: i=0, \ldots, n\}$. The authors in \cite{chen-berger2005} showed that the UMCO channel can be transformed to one with state information, and that under certain conditions on the channel and channel input distributions, dynamic programming can be used to compute feedback capacity.  \\

\subsection{Discussion of Main  Results and Methodology}
\label{mo-dr} 
In this paper,  the emphasis is on  any combination of the following classes of   channel distributions and  transmission cost functions\footnote{The methodology developed  in the paper can be extended  to channels and transmission cost functions with past dependence on channel input symbols; however,  such generalizations  are beyond the scope of this paper.}.
\begin{align}
&\mbox{\bf Channel Distributions}\nonumber \\
&\mbox{ Class A.}  \hso {\bf P}_{B_i|B^{i-1}, A^{i}}(db_i|b^{i-1}, a^{i}) = {\bf P}_{B_i|B^{i-1}, A_i}(db_i|b^{i-1}, a_i), \hso  i=0, \ldots, n.  \label{CD_C1} \\
&\mbox{ Class B.} \hso {\bf P}_{B_i|B^{i-1}, A^{i}}(db_i|b^{i-1}, a^{i})     = {\bf P}_{B_i|B_{i-M}^{i-1}, A_i}(db_i|b_{i-M}^{i-1}, a_i), \hso
    i=0, \ldots,n. \label{CD_C4}\\
&\mbox{\bf Transmission Costs}\nonumber \\    
&\mbox{ Class A.} \hso  \gamma_i(T^ia^n, T^ib^{n-1}) ={\gamma}_i^{A}(a_i, b^{i-1}), \hst i=0, \ldots, n, \\
&\mbox{ Class B.} \hso \gamma_i(T^ia^n, T^ib^{n-1}) ={\gamma}_i^{B}(a_i, b_{i-K}^{i-1}), \hst i=0, \ldots, n. \label{TC_2}
\end{align}
Here, $ \{K, M\}$  are nonnegative finite integers and the following convention is used. 
\begin{align}
&\mbox{If $M=0$ then}\hso   {\bf P}_{B_i|B_{i-M}^{i-1}, {A}_i}(db_i|b_{i-M}^{i-1}, {a}_i)\Big|_{M=0} \equiv  {\bf P}_{B_i| {A}_i}(db_i| {a}_i), \hso \mbox{for}  \hso i=0,1, \ldots, n. \nonumber \\
&\mbox{If $K=0$ then} \hso  {\gamma}_i^{B}(a_i, b_{i-K}^{i-1})\Big|_{K=0} \equiv {\gamma}_i^{B}(a_i), \hso i=0, \ldots, n. \nonumber 
\end{align}
Thus, for $M=0$  the above convention implies the channel degenerates to the memoryless channel  ${\bf P}_{B_i| A_i}(db_i| a_i), i=0,1, \ldots, n$.

The above classes of  channel conditional distributions may be induced by various nonlinear channel models (NCM), such as, nonlinear and linear time-varying  Autoregressive models, and nonlinear and linear  channel models expressed in state space form \cite{caines1988}. Such classes are investigated in \cite{charalambous-kourtellarisIT2015_Part_2}.\\
An  over view of  the methodology and results obtained, is  discussed below,     to illustrate  analogies to Shannon's  two-letter capacity formulae (\ref{DMC_Shannon_1}) and conditional independence conditions (\ref{CI_DMC}) in relation to (\ref{CI_COND}).

\subsubsection{\bf Channels of Class A and Transmission Cost of Class A or B}
In Theorem~\ref{thm-ISR}, Step 1 of a two-step procedure, based on stochastic optimal control,  is applied to channel distributions of Class A,  $\big\{ {\bf P}_{B_i|B^{i-1}, A_i}(db_i|b^{i-1}, a_i):i=0, 1, \ldots, n\big\}$, to show  the  optimal channel input conditional distribution, which  maximizes  $I(A^n \rar B^n)$ satisfies conditional independence ${\bf P}_{A_i|A^{i-1}, B^{i-1}}={\bf P}_{A_i|B^{i-1}}, i=0, \ldots, n$, and hence it    occurs in the subset   
\begin{align}
\overline{\cal P}_{[0,n]}^{A} \tri\big\{ {\bf P}_{A_i|B^{i-1}}(da_i|b^{i-1}):i=0, \ldots, n\big\} \subset  {\cal P}_{[0,n]} .
\end{align}
This means that  for each $i$, the information structures of the maximizing channel input distribution is ${\cal I}_i^{\bf P} \tri \{b^{i-1}\} \subset \{a^{i-1}, b^{i-1}\}$, for $i=0,1, \ldots, n$. \\
The  characterization of the FTFI capacity is 
\begin{align}
C_{A^n \rar B^n}^{FB, A} 
= \sup_{ \overline{\cal P}_{[0,n]}^{A}} \sum_{i=0}^n I(A_i; B_i|B^{i-1}).
\end{align}
 If a transmission cost ${\cal P}_{[0,n]}(\kappa)$ is imposed corresponding to any of the functions  $\gamma_i^{A}(a_i, b^{i-1})$, $\gamma_i^{B}(a_i, b_{i-K}^{i-1}), i=0,1, \ldots, n$,  the characterization of the FTFI capacity is
\bea 
 C_{A^n \rar B^n}^{FB, A}(\kappa) \tri \sup_{ \overline{\cal P}_{[0,n]}^{A}\bigcap  {\cal P}_{[0,n]}(\kappa) } \sum_{i=0}^n I(A_i; B_i|B^{i-1}). \label{TR_A-B}
 \eea

\subsubsection{\bf Channels of Class B Transmission Cost of Class A or B}
In Theorem~\ref{cor-ISR_C4}, Step 2 of the two-step procedure, a variational equality of directed information,  is  applied to channel distributions of Class B,  $\big\{ {\bf P}_{B_i|B_{i-M}^{i-1}, A_i}(db_i|b_{i-M}^{i-1}, a_i): i=0,1, \ldots, n\big\}$, to show the optimal channel input conditional distribution, which maximizes $I(A^n \rar B^n)$ satisfies conditional independence ${\bf P}_{A_i|A^{i-1}, B^{i-1}}={\bf P}_{A_i|B_{i-M}^{i-1}}, i=0, \ldots, n$, and hence it   occurs in the subset 
\begin{align}
\sr{\circ}{ {\cal P}}_{[0,n]}^{B.M} \tri \hso &   \Big\{ {\bf P}_{A_i|B_{i-M}^{i-1}}(da_i|b_{i-M}^{i-1}):i=0, 1, \ldots, n\Big\}. 
\end{align}
The characterization of the FTFI capacity is then given by the following expression.
\begin{align}
C_{A^n \rar B^n}^{FB, B.M} = \sup_{ \sr{\circ}{\cal P}_{[0,n]}^{B.M}  } \sum_{i=0}^n I(A_i; B_i|B_{i-M}^{i-1}).
\end{align}
If a transmission cost ${\cal P}_{[0,n]}(\kappa)$ is imposed corresponding to cost functions of Class B, 
$\big\{\gamma_i^B(a_i, b_{i-K}^{i-1}): i=0, \ldots, n\big\}$, the optimal channel input conditional distribution occurs in the subset  $\sr{\circ}{ {\cal P}}_{[0,n]}^{B.J} \bigcap {\cal P}_{[0,n]}(\kappa)$, where $J \tri \max\{M, K\}$. \\
The characterization of the FTFI capacity is given by the following expression.
\begin{align}
C_{A^n \rar B^n}^{FB, B.J}(\kappa) = \sup_{ \sr{\circ}{\cal P}_{[0,n]}^{B.J}  \bigcap {\cal P}_{[0,n]}(\kappa)  } \sum_{i=0}^n \int \log \Big(   \frac{d {\bf P}_{B_i|B_{i-M}^{i-1}, A_i}(\cdot|b_{i-M}^{i-1}, a_i)}{ d{\bf P}_{B_i|B_{i-J}^{i-1}}(\cdot|b_{i-J}^{i-1})}(b_i)\Big) {\bf P}_{B_{i-J}^{i}, A_i}(db_{i-J}^{i}, da_i), \hso J \tri \max\{M, K\}. \label{Intr_B1a}
\end{align}
where 
\begin{align}
{\bf P}_{B_{i-J}^i, A_i}(db_{i-J}^i, da_i) =& {\bf P}_{B_i|B_{i-M}^{i-1}, A_i}(db_i|b_{i-M}^{i-1}, a_i)\otimes {\bf P}_{A_i|B_{i-J}^{i-1}}(da_i|b_{i-J}^{i-1}) \otimes {\bf P}_{B_{i-J}^{i-1}}(db_{i-J}^{i-1}),  \hso i=0,1, \ldots, n, \\
{\bf P}_{B_i|B_{i-J}^{i-1}}(db_i|b_{i-J}^{i-1}) =&\int {\bf P}_{B_i|B_{i-M}^{i-1}, A_i}(db_i|b_{i-M}^{i-1}, a_i)\otimes {\bf P}_{A_i|B_{i-J}^{i-1}}(da_i|b_{i-J}^{i-1}),  \hso i=0,1, \ldots, n.  \label{Intr_B1}
\end{align}
The above expressions imply the channel output process or controlled process $\{B_i: i=0, \ldots, n\}$ is a $J-$order  Markov process.\\
On the other hand,  if a transmission cost ${\cal P}_{[0,n]}(\kappa)$ is imposed corresponding to  $ \gamma_i^{A}(a_i, b^{i-1}), i=0,1, \ldots, n$,  the optimal channel input distribution occurs in the set   $\overline{\cal P}_{[0,n]}^{A}\bigcap {\cal P}_{[0,n]}(\kappa) 
$.

The above characterizations of FTFI capacity (and by extension of feedback capacity characterizations)  state that the information structure of the optimal channel input conditional distribution  is determined by $\max\{M, K\}$, where $M$ specifies the order of the memory  of the channel conditional distribution, and $K$ specifies the dependence of the  transmission cost function, on  past channel output symbols. \\
These structural properties of optimal channel input conditional distributions are analogous to those of memoryless channels, and  they hold for finite, countable and abstract alphabet spaces (i.e., continuous), and channels defined by nonlinear models, state space models, autoregressive models, etc.   \\
The following special cases illustrate the explicit analogy to Shannon's two-letter capacity formulae of memoryless channels.

{\bf Special Case-$M=2, K=1$.} For any channel   $\big\{ {\bf P}_{B_i|B_{i-1},B_{i-2}, A_i}(db_i|b_{i-1},b_{i-2}, a_i):i=0, 1, \ldots, n\big\}$, and  transmission cost function  $\big\{ \gamma_i^{B.1}(a_i, b_{i-1}), i=1, \ldots, n\big\}$, from (\ref{Intr_B1a})-(\ref{Intr_B1}),  the  optimal channel input conditional distribution occurs in the subset 
\bea
\sr{\circ}{ {\cal P}}_{[0,n]}^{B.2}(\kappa) \tri     \Big\{ {\bf P}_{A_i|B_{i-1}, B_{i-2}}(da_i|b_{i-1}, b_{i-2}),i=0, 1, \ldots, n:   \frac{1}{n+1} {\bf E}\big\{\sum_{i=0}^n \gamma_i^{B.1}(A_i, B_{i-1})\big\} \leq \kappa\Big\} . \label{CID_A.1}
\eea
The information structure of the optimal channel input conditional distribution implies  the joint distribution of $(A_i,B_i)$, conditioned on $(A^{i-1}, B^{i-1})$, is given by  
\begin{align} 
 {\bf P}_{A_i. B_i|A^{i-1}, B^{i-1}}={\bf P}_{A_i,B_i| A_{i-1}, B_{i-1}, B_{i-2}}\equiv {\bf P}_{B_i| A_{i}, B_{i-1}, B_{i-2}}\otimes {\bf P}_{A_i|B_{i-1}, B_{i-2}},   \hso  i=0,1, \ldots, n. \label{SC_1}
      \end{align}
the channel output process  $\{B_i: i=0, \ldots, n\}$ is a second-order Markov process, i.e.,  
\begin{align}      
      {\bf P}_{B_i|B^{i-1}}={\bf P}_{B_i| B_{i-1}, B_{i-2}} = \int_{ {\mb A}_i} {\bf P}_{B_i|B_{i-1},B_{i-2}, A_i}(db_i|b_{i-1}, b_{i-2}, a_i)\otimes {\bf P}_{A_i|B_{i-1}, B_{i-2}}(da_i|b_{i-1}, b_{i-2}),   \hso  i=0,1, \ldots, n. \label{SC_2}
      \end{align}
and that the characterization of the FTFI capacity is given by the following 4-letter expression at each time $i=0, \ldots, n$.
\begin{align}
C_{A^n \rar B^n}^{FB,B.2}(\kappa) \tri&  \sup_{ \sr{\circ}{ {\cal P}}_{[0,n]}^{B.2}(\kappa)    }\sum_{i=0}^n I(A_i; B_i|B_{i-1}, B_{i-2}). \label{intro_fbumc3}\\
=& \sup_{ \sr{\circ}{ {\cal P}}_{[0,n]}^{B.2}(\kappa)    }\sum_{i=0}^{n} {\bf E}\bigg\{  \ell_i\Big(A_i, S_i\Big)\bigg\}, \hst S_j\tri (B_{j-1}, B_{j-2}) \hso j=0, \ldots, n, \label{CIS_6f_new}
\end{align}
where the pay-off $\ell_j(\cdot, \cdot)$ is given by
\begin{align}
  \ell_j\Big(a_j, s_j \Big)  \tri  \int_{ {\mb B}_j } \log \Big(\frac{d{\bf P}_{B_j|S_j, A_j}(\cdot| s_j, a_j)}
  {d{\bf P}_{B_j|S_j}(\cdot| s_j)}(b_j)\Big) {\bf P}_{B_j|S_j, A_j}(db_j| s_j, a_j), \hso j=0, \ldots, n \label{BAM31_new1}
\end{align}
Moreover, if the channel input distribution is restricted to a time-invariant distribution, i.e.,   ${\bf P}_{A_i|S_i}(da|s)\equiv {\bf P}^\infty(da|s), i=0, \ldots,$ and the channel distribution is time-invariant, then the transition probability distribution of $\{S_i: i=0, \ldots, \}$ is time-invariant, and ${\bf P}_{B_i|S_{i-1}}\equiv {\bf P}^\infty(db|s), i=0, \ldots, $ is also time-invariant. Consequently, the per unit limiting version of (\ref{intro_fbumc3}), specifically, $C_{A^\infty \rar B^\infty}^{FB,B.2}(\kappa)$, under conditions which ensure ergodicity,  is characterized by time-invariant and the ergodic distribution of $\{S_i: i=0, \ldots, n\}$ \cite{hernandezlerma-lasserre1996}. 
  \\
{\bf Special Case-$M=2, K=0$.}
This means  no transmission cost is imposed, and hence  the supremum in (\ref{intro_fbumc3}) is over $\sr{\circ}{ {\cal P}}_{[0,n]}^{B.2} \tri     \big\{ {\bf P}_{A_i|B_{i-1}, B_{i-2}}(da_i|b_{i-1}, b_{i-2}),i=0, 1, \ldots, n\big\}$. Let $C_t^{B.2}: {\mb B}_{t-1} \times {\mathbb B}_{t-2} \longmapsto {\mb R}$ denote the cost-to-go corresponding to (\ref{CIS_6f_new}), with $K=0$,  from time ``$t$'' to the terminal time   ``$n$'' given the values of the output  $S_t=(B_{t-1}, B_{t-2})=(b_{t-1}, b_{t-2})$. \\
Then the  cost-to-go satisfies the following dynamic programming  recursions. 
\begin{align}
C_n^{B.2}(s_n) =& \sup_{ {\bf P}_{A_n|S_n}} \Big\{\int_{{\mb A}_n \times {\mb B}_n }    \log\Big(\frac{d{\bf P}_{B_n|S_n, A_n}(\cdot| s_n, a_n)}{d{\bf P}_{B_n|S_n}(\cdot|s_n)}(b_n)\Big)   {\bf P}_{B_n|S_n, A_n}(db_n|s_n, a_n) \otimes {\bf P}_{A_n|S_n}(da_n|s_n)  
  \Big\},  \label{NCM-B.2-DP2_IT} \\
C_t^{B.2}(s_t) =& \sup_{  {\bf P}_{A_t|S_t}} \Big\{\int_{{\mb A}_t \times {\mb B}_t }    \log\Big(\frac{d{\bf P}_{B_t|S_t, A_t}(\cdot|s_t, a_t)}{d{\bf P}_{B_t|S_t}(\cdot|s_t)}(b_t)\Big)   {\bf P}_{B_t|S_t, A_t}(db_t|s_t, a_t)  \otimes {\bf P}_{A_t|S_t}(da_t|s_t)  \nonumber  \\
  & + \int_{{\mb A}_t \times {\mb B}_t }   C_{t+1}^{B.2}(s_t) {\bf P}_{B_t|S_t, A_t}(db_t|s_t, a_t)  \otimes {\bf P}_{A_t|S_t}(da_t|s_t)    \Big\}, \hst t=n-1, n-2, \ldots, 0. \label{NCM-B.2-DP1_IT}
  \end{align}
The characterization of the FTFI capacity and feedback capacity are expressed via the $C_0^{B.2}(s_0)$ and the fixed distribution $\mu_{B_{-1}, B_{-2}}(db_{-1}, db_{-2})$ by 
\begin{align}
C_{A^n \rar B^n}^{FB, B.2}=\int_{{\mb B}_{-1} \times {\mb B}_{-2}} C_0^{B.2}(s_0)\mu_{B_{-1}, B_{-2}}(ds_0), \hst C_{A^\infty \rar B^\infty} ^{FB. B.2}\tri \liminf_{n \longrightarrow \infty}\frac{1}{n+1} C_{A^n \rar B^n}^{FB, B.2}.
\end{align}
Obviously, even for finite ``$n$'',  from the above recursions, we deduce that the information structure, $\{S_t=B_{t-1}, B_{t-2}: t=0,\ldots, n\}$, of the control object, namely, $\big\{{\bf P}_{A_t|S_t}: t=0, \ldots, n\big\}$, induces conditional  probabilities  $\big\{{\bf P}_{B_t|S^t}={\bf P}_{B_t|S_t}: t=0, \ldots, n\big\}$  which are $2$nd order Markov, i.e., $\big\{{\bf P}_{S_{t+1}|S^{t}}={\bf P}_{S_{t+1}|S_{t}}: t=0, \ldots, n-1\big\}$, resulting in a significant reduction in computational complexity of the above dynamic programming recursions. Clearly, for any fixed $S_0=s_0$, then   $C_{A^\infty \rar B^\infty}$ depends on the initial state $S_0=s_0$. However, if the channel is time-invariant and the the distributions $\big\{{\bf P}_{A_t|S_t}: t=0, \ldots, \big\}$ are either restricted or converge to time-invariant distributions, and the corresponding transition probabilities $\big\{{\bf P}_{S_{t+1}|S^{t}}={\bf P}_{S_{t+1}|S_{t}}: t=0, \ldots, n-1\big\}$ are irreducible and aperiod, then there is unique invariant distribution for $\{S_i: i=0, \ldots, \}$ and $C_{A^\infty \rar B^\infty}^{B.2}$ is independent of the initial distribution $\mu_{B_{-1}, B_{-2}}(ds_0)$. Such questions are addressed in \cite{chen-berger2005} for the channel $\big\{ {\bf P}_{B_i|B_{i-1}, A_i}(db_i|b_{i-1}, a_i):i=0, 1, \ldots, n\big\}$.  They can be addressed from the general theory of per unit time-infinite horizon Markov decision theory \cite{kumar-varayia1986}, and more generally by solving explicitly the above dynamic programming recursions and investigating their per unit-time limits (see \cite{stavrou-charalambous-kourtellaris:C2016}).

{\bf Special Case-$M=K=1$.} If the channel is the so-called Unit Memory Channel Output (UMCO) defined by  $\big\{ {\bf P}_{B_i|B_{i-1}, A_i}(db_i|b_{i-1}, a_i):i=0, 1, \ldots, n\big\}$, and  the transmission cost function is $\big\{ \gamma_i^{B.1}(a_i, b_{i-1}), i=1, \ldots, n\big\}$, the  optimal channel input conditional distribution occurs in the subset 
\bea
\sr{\circ}{ {\cal P}}_{[0,n]}^{B.1}(\kappa) \tri     \Big\{ {\bf P}_{A_i|B_{i-1}}(da_i|b_{i-1}),i=0, 1, \ldots, n:   \frac{1}{n+1} {\bf E}\big\{\sum_{i=0}^n \gamma_i^{B.1}(A_i, B_{i-1})\big\} \leq \kappa\Big\} . \label{CID_A.1_a}
\eea
and the characterization of the FTFI capacity degenerates to the following sums of a 3-letter expressions.
\begin{align}
C_{A^n \rar B^n}^{FB,B.1}(\kappa) \tri  \sup_{ \sr{\circ}{ {\cal P}}_{[0,n]}^{B.1}(\kappa)    }\sum_{i=0}^n I(A_i; B_i|B_{i-1}).  \label{intro_fbumc3_A}
\end{align}
The importance  of   variational equalities to identify  information structures of capacity achieving channel input conditional distributions is first applied  in \cite{ckthesis}. For the  BSSC$(\alpha, \beta)$ (which is a special case of the UMCO) with transmission cost, it is shown in \cite{kourtellaris-charalambous2015}, that the characterizations of feedback capacity and capacity without feedback, admit  closed form expressions.  Moreover,  this channel  is matched to the Binary Symmetric Markov Source through the use of nonanticipative Rate Distortion Function (RDF) in \cite{kourtellaris-charalambous-boutros:2015} 
That is, there is a perfect duality between the  BSSC$(\alpha, \beta)$ with transmission cost and the Binary Symmetric Markov Source with a single letter distortion function.   \\
 Recently,  the results of this paper are applied in  \cite{stavrou-charalambous-kourtellaris:C2016} (see also \cite{stavrou-charalambous-kourtellaris:J2016}) to derive   sequential necessary and sufficient conditions to optimize the characterizations of FTFI capacity. Moreover, using the necessary and sufficient conditions  closed form expressions for the optimal channel input distributions and 
 feedback capacity, are obtained for various applications examples defined on finite alphabet spaces. This paper includes in Section~\ref{example}, an illustrative example, which reveals many silent properties of capacity achieving distributions, with and without feedback, for a simple first-order Gaussian Linear Channel Model. \\
 A detailed investigation of the  characterization  of FTFI capacity, and feedback capacity,  of Multiple-Input Multiple Output (MIMO) Gaussian Linear Channel Models   with memory   is  included  
in the second part of this two-part investigation \cite{charalambous-kourtellarisIT2015_Part_2}.

\section{Directed Information and Definitions of Extremum Problems of Capacity}
In this section, the notation adopted  in the rest of the paper is introduced, and   a variational equality of directed information is recalled from \cite{charalambous-stavrou2013aa}. \\
The following notation is used throughout the paper.
\begin{align}
& {\mathbb Z}: \hso  \mbox{set of  integer};  \nonumber \\
& {\mathbb N}_0: \hso  \mbox{set of nonnegative integers} \hso \{0, 1,2,\dots\}; \nonumber \\
&(\Omega, {\cal F}, {\mathbb P}): \mbox{probability space, where ${\cal F}$ is the $\sigma-$algebra generated by subsets of $\Omega$}; \nonumber \\ 
& {\cal  B}({\mathbb  W}): \hso \mbox{Borel $\sigma-$algebra of a given topological space  ${\mathbb W}$};  \nonumber \\
&{\cal M}({\mathbb W}): \hso \mbox{set of all probability measures on ${\cal  B}({\mathbb W})$ of a Borel space ${\mathbb W}$}; \nonumber\\
&{\cal K}({\mathbb V}|{\mathbb W}): \hso \mbox{set of all stochastic kernels on  $({\mathbb V}, {\cal  B}({\mathbb V}))$ given $({\mathbb W}, {\cal  B}({\mathbb W}))$ of Borel spaces ${\mathbb W}, {\mathbb V}$}.\nonumber
\end{align}
All spaces (unless stated otherwise) are complete separable metric spaces, also called  Polish spaces, i.e., Borel spaces. This generalization is judged necessary to treat simultaneously discrete, finite alphabet,  real-valued ${\mathbb R}^k$ or complex-valued ${\mathbb C}^k$ random processes for any positive integer $k$,  etc. 

\subsection{Basic Notions of Probability}
\label{subsec-prob}
 The product measurable space of the two measurable spaces  $({\mb X}, {\cal  B}({\mb X}))$ and $({\mb Y}, {\cal  B}({\mb Y}))$ is denoted by $({\mb X} \times {\mb Y}, {\cal  B}({\mb X})\odot  {\cal  B}({\mb Y}))$, where  ${\cal  B}({\mb X})\odot  {\cal  B}({\mb Y})$ is the product $\sigma-$algebra generated by $\{A \times B:  A \in {\cal  B}({\mb X}), B\in  {\cal  B}({\mb Y})\}$. \\
A Random Variable (RV)  defined on a probability space $(\Omega, {\cal F}, {\mathbb P})$ by the mapping $X: (\Omega, {\cal F}) \longmapsto ({\mb X}, {\cal  B}({\mb X}))$  induces a probability measure $ {\bf P}(\cdot) \equiv {\bf P}_X(\cdot)$ on  $({\mb X}, {\cal  B}({\mb X}))$ as follows\footnote{The subscript on $X$ is often omitted.}.
\begin{align}
{\bf P}(A) \equiv  {\bf P}_X(A)  \tri {\mathbb P}\big\{ \omega \in \Omega: X(\omega)  \in A\big\},  \hso  \forall A \in {\cal  B}({\mb X}).
 \end{align}
A RV is called discrete if there exists a countable set ${\cal S}_X\tri \{x_i: i \in {\mathbb N}\}$ such that $\sum_{x_i \in {\cal S}_X} {\mathbb  P} \{ \omega \in \Omega : X(\omega)=x_i\}=1$. The probability measure ${\bf P}_X(\cdot)$  is then concentrated on  points in ${\cal S}_X$, and it is defined by 
\bea
 {\bf P}_X(A)  \tri \sum_{x_i \in {\cal S}_X \bigcap A} {\mathbb P} \{ \omega \in \Omega : X(\omega)=x_i\}, \hso \forall A \in {\cal  B}({\mb X}). 
\eea 

If the cardinality of ${\cal S}_X$ is finite then the RV is finite-valued  and it is called a finite alphabet RV. \\
Given another RV $Y: (\Omega, {\cal F}) \longmapsto ({\mb Y}, {\cal  B}({\mb Y}))$, for each Borel subset $B$ of ${\mb Y}$ and any sub-sigma-field ${\cal G} \in {\cal F}$ (collection of events)   the conditional probability of  event $\{Y \in B\}$ given ${\cal G}$ is defined by ${\mathbb P}\{ Y \in B| {\cal G}\}(\omega),$ and this  is an ${\cal G}-$measurable function $\forall \omega \in \Omega$. This conditional probability  induces a conditional probability measure on  $({\mb Y}, {\cal  B}({\mb Y}))$ defined by ${\bf P}(B| {\cal G})(\omega)$, which is a version of ${\mathbb P}\{ Y \in B| {\cal G}\}(\omega).$ For example,  if ${\cal G}$ is the $\sigma-$algebra generated by RV $X$, and  $B=dy$, then  ${\bf P}_{Y|X}(dy| X)(\omega)$ is called the conditional distribution of RV $Y$ given RV $X$. The conditional distribution of RV $Y$ given $X=x$ is denoted by ${\bf P}_{Y|X}(dy| X=x)  \equiv {\bf P}_{Y|X}(dy|x)$. Such conditional distributions are  equivalently  described   by  stochastic kernels or transition functions ${\bf K}(\cdot|\cdot)$ on $ {\cal  B}({\mb Y}) \times {\mathbb X}$, mapping ${\mb X}$ into ${\cal M}({\mathbb Y})$ (the space of probability measures on $({\mathbb Y}, ({\cal B}({\mathbb Y})))$, i.e., $x \in {\mathbb X}\longmapsto {\bf K}(\cdot|x)\in{\cal M}({\mathbb Y})$, and hence the distributions  are parametrized by $x \in {\mb X}$.\\
The family of probability measures on $({\mb Y}, {\cal B}({\mb Y})$ parametrized by $x \in {\mb X}$, is defined by $${\cal K}({\mb Y}| {\mb X})\tri \big\{{\bf K}(\cdot| x) \in {\cal M}({\mathbb Y}): \hso x \in {\mathbb X}\hso \mbox{and  $\forall F \in {\cal  B}({\mb Y})$, the function ${\bf K}(F|\cdot)$ is ${\cal  B}({\mb X})$-measurable.}\big\}.$$

 \subsection{FTFI Capacity and Variational Equality}
\label{def-sub2}
The channel input and  channel output alphabets are  sequences of  measurable spaces $\{({\mb A}_i,{\cal  B}({\mb A }_i)):i\in\mathbb{Z}\}$ and  $\{({\mb  B}_i,{\cal  B}({\mb  B}_i)):i\in\mathbb{Z}\}$, respectively, and 
their history spaces  are the product spaces ${\mb A}^{\mathbb{Z}}\tri {{\times}_{i\in\mathbb{Z}}}{\mb A}_i,$ ${\mb  B}^{\mathbb{Z}}\tri {\times_{i\in\mathbb{Z}}}{\mb  B}_i$.  These spaces are endowed with their respective product topologies, and  ${\cal  B}({\Sigma}^{\mathbb{Z}})\tri \odot_{i\in\mathbb{Z}}{\cal  B}({\Sigma }_i)$  denotes the $\sigma-$algebra on ${\Sigma }^{\mathbb{Z}}$, where ${\Sigma}_i \in  \Big\{{\mb A}_i, {\mb  B}_i\Big\}$,  ${\Sigma}^{\mathbb{Z}} \in  \Big\{{\mb A}^{\mathbb N}, {\mb  B}^{\mathbb Z}\Big\}$,  generated by cylinder sets. Thus, for any $n \in {\mb Z}$,  ${\cal  B}({\Sigma }^n)$ denote the $\sigma-$algebras of cylinder sets in ${\Sigma }^{\mathbb{Z}}$, with bases over  $C_i\in{\cal  B}({\Sigma}_i)$, $i=\ldots, -1, 0,1,\ldots,n$, respectively. Points in ${\mb A}^{n}$, ${\mb  B}^{n}$  are denoted by  $a^n\tri \{\ldots, a_{-1}, a_0,a_1,\ldots,a_n\}\in{\mb A}^n$, $b^n\tri \{\ldots, b_{-1}, b_0,b_1,\ldots,b_n\}\in{\mb  B}^n$. Similarly, points in ${\mb Z}_k^m \tri \times_{j=k}^m {\mb Z}_j$ are denoted by $z_{k}^m \tri \{z_k, z_{k+1}, \ldots, z_m\} \in {\mb Z}_k^m$,   $(k, m)\in   {\mathbb Z} \times {\mathbb Z}$. We often restrict ${\mathbb Z}$ to ${\mathbb N}_0$. \\

\noi{\bf Channel Distribution with Memory.}  A sequence of stochastic kernels or distributions defined by 
\begin{align}
{\cal C}_{[0,n]} \tri \Big\{Q_i(db_i|b^{i-1},a^{i})= {\bf P}_{B_i|B^{i-1}, A^i}  \in {\cal K}({\mb  B}_i| {\mb  B}^{i-1} \times {\mb A}^i) :  i=0,1, \ldots, n \Big\}. \label{channel1}
\end{align}
 At each time instant $i$ the conditional distribution of channel output $B_i$  is affected causally by previous channel output symbols $b^{i-1} \in {\mb B}^{i-1}$ and current and previous channel input symbols $a^{i} \in {\mb A}^i, i=0,1, \ldots, n$.

\noi{\bf Channel Input Distribution with Feedback.}  A  sequence of stochastic kernels defined by 
\bea
{\cal P}_{[0,n]} \tri  \Big\{  P_i(da_i|a^{i-1},b^{i-1})={\bf P}_{A_i|A^{i-1}, B^{i-1}}  \in  {\cal K}({\mb A}_i| {\mb A}^{i-1} \times {\mb  B}^{i-1}):   i=0,1, \ldots, n \Big\}. \label{rancodedF}
\eea
At each time instant $i$ the conditional  distribution of channel input $A_i$  is affected causally by past  channel inputs and  output symbols  $\{a^{i-1}, b^{i-1}\} \in {\mb A}^{i-1} \times {\mb B}^{i-1}, i=0,1, \ldots, n$. Hence, the information structure of the channel input distribution at time instant $i$ is ${\cal I}_i^{P} \tri \{a^{i-1},b^{i-1}\}\in {\mathbb A}^{i-1} \times {\mathbb B}^{i-1}, i=0,1, \ldots, n$. 

{\bf Admissible Histories.} For each $i=-1, 0, \ldots, n$, we introduce the space ${\mb G}^{i}$ of admissible histories of channel input and output symbols,   as follows. Define
\begin{IEEEeqnarray}{rCl}
 {\mb  G}^i\triangleq  {\mb B}^{-1} \times  \mathbb{A}_0\times \mathbb{B}_0\times \hdots \times \mathbb{A}_{i-1}\times\mathbb{B}_{i-1}\times  \mathbb{A}_i\times {\mb B}_i,
 \hso  i=0, \ldots,n,  \hso {\mb G}^{-1}= {\mb B}^{-1}.
\end{IEEEeqnarray}
A typical element of ${\mb G}^i$ is a sequence of the form $( b^{-1},a_0,b_0,\hdots, a_{i},b_i)$.  We equip the space 
  ${\mb G}^i$  with the natural $\sigma$-algebra 
  ${\cal B}({\mb G}^i)$, for $i=-1,0,\ldots, n$. Hence, for each $i$, the information structure of the channel input distribution  is
\begin{align}
{\cal I}_i^P\tri \Big\{B^{-1}, A_0, B_0, \ldots, A_{i-1}, B_{i-1}\Big\},  \; i=0,1, \ldots, n, \hst {\cal I}_0^P \tri  \big\{ B^{-1}\big\}
\end{align}
This implies at time $i=0$, the initial distribution is $P_0(da_0|a^{-1}, b^{-1})=P_0(da_0| {\cal I}_0^P)=P_0(da_0| b^{-1})$. However, we can modify  ${\cal I}_0^P$ to  consider an alternative convention such as ${\cal I}_0^P =\{\emptyset\}$ or ${\cal I}_0^P=\{a^{-1}, b^{-1}\}$.

{\bf Joint and Marginal Distributions.}  Given any channel input conditional  distribution $\big\{{ P}_i(da_i|a^{i-1}, b^{i-1}): i=0,1, \ldots, n\big\} \in {\cal P}_{[0,n]}$,  any channel distribution $\big\{Q(db_i| b^{i-1}, a^{i}): i=0,1, \ldots, n\big\}\in {\cal C}_{[0,n]}$,   and the initial probability distribution ${\bf P}(db^{-1})\equiv \mu(db^{-1})\in {\cal M}({\mb G}^{-1})$,  the  induced joint distribution  ${\bf P}^{P}(da^n, db^n)$ on  the  canonical space  $\Big({\mb  G}^n, {\cal  B}({\mb G}^n))\Big)$ is defined uniquely, and a probability space $\Big(\Omega, {\cal F}, {\mathbb P}\Big)$ carrying the sequence of RVs $\{(A_i, B_i): i=0, \ldots, n\}$ and $B^{-1}$ can be  constructed, as follows\footnote{The superscript notation, i.e.,  ${\bf P}^P$, ${\bf E}^P$ is used to track  the dependence of the  distribution and expectation  on the channel input distribution $P_i(da_i|a^{i-1}, b^{i-1}), i=0, \ldots,n$.}.
\begin{align}
 {\mathbb P}\big\{A^n \in d{a}^n, B^n \in d{b}^n\big\}  \tri  & {\bf P}^P( db^{-1},da_0,db_0,\hdots,da_{n},db_n),  \hso n \in {\mathbb N}_0 \nonumber \\
=&\mu(db^{-1})\otimes P_0(da_0|b^{-1})  \otimes Q_0(db_0|b^{-1},a_0) \otimes 
 P_1(da_1|b^{-1},b_0,a_0) \nonumber \\
 &\otimes\hdots\otimes Q_{n-1}(db_{n-1}|b^{n-2},a^{n-1}) \otimes  P_{n}(da_{n}|b^{n-1},a^{n-1})\otimes Q_n(db_n|b^{n-1},a^{n}) \label{CIS_2gde2new} \\
\equiv & \mu(db^{-1})   \otimes_{j=0}^n \Big(Q_j(db_j|b^{j-1}, a^j)\otimes P_j(da_j|a^{j-1}, b^{j-1})\Big). \label{CIS_2gg_new} 
\end{align}
The joint distribution of $\big\{B_i: i=0, \ldots, n\big\}$ and its conditional distribution are defined by
\begin{align}
{\mathbb  P}\big\{B^n \in db^n\big\} \tri \; & {\bf P}^{P}(db^n) =  \int_{{\mb A}^n}  {\bf P}^{ P}(da^n, db^n) ,   \hso  n \in {\mathbb N}_0,  \label{CIS_3g}\\
\equiv \; &   \Pi_{0,n}^{P}(db^n) = \mu(db^{-1}) \otimes_{i=0}^n \Pi_i^{ P}(db_i|b^{i-1}) \label{MARGINAL} \\
\Pi_i^{ P}(db_i|b^{i-1})= \; &  \int_{{\mb A}^i} Q_i(db_i|b^{i-1}, a^i)\otimes P_i(da_i|a^{i-1}, b^{i-1}) \otimes {\bf P}^{P}(da^{i-1}|b^{i-1}), \hso  i=0, \ldots, n . \label{CIS_3a}
\end{align}
The above distributions are parametrized by  either a fixed $B^{-1}=b^{-1} \in {\mathbb B}^{-1}$ or a fixed distribution ${\bf P}(db^{-1})=\mu(db^{-1})$. 

{\bf  FTFI Capacity.} Directed  information (pay-off)   $I(A^n \rar B^n)$  is defined   by 
\begin{align}
I(A^n\rar B^n) \tri  &\sum_{i=0}^n {\bf E}^{{ P}} \Big\{  \log \Big( \frac{dQ_i(\cdot|B^{i-1},A^i) }{d\Pi_i^{{ P} }(\cdot|B^{i-1})}(B_i)\Big)\Big\}   \\
=& \sum_{i=0}^n \int_{{\mb A}^{i} \times {\cal  B}^{i}   }^{}   \log \Big( \frac{ dQ_i(\cdot|b^{i-1}, a^i) }{d\Pi_i^{{ P}}(\cdot|b^{i-1})}(b_i)\Big) {\bf P}^{ P}( da^i, db^i)  \label{CIS_6}\\
 \equiv &  {\mathbb I}_{A^n\rightarrow B^n}({P}_i,{ Q}_i,  :~i=0,1,\ldots,n)\label{DI_4}
\end{align}
where the notation (\ref{DI_4}) illustrates that  $I(A^n \rar B^n)$ is  a functional of the two sequences of conditional  distributions, $\big\{{P}_i(da_i|a^{i-1}, b^{i-1}), { Q}_i(db_i|b^{i-1}, a^i): i=0,1, \ldots, n\big\}$, and the initial distribution,  which uniquely define the joint distribution,  the marginal and conditional distributions $\big\{{\bf P}(da^i, db^i), \Pi_{0,i}^P(db^i), \Pi_i^P(db_i|b^{i-1}): i=0,1, \ldots, n\big\}$.\\
Clearly, (\ref{CIS_6}) includes formulations with respect to probability density functions and probability mass functions.

\noi{\bf Transmission Cost.}  The cost of transmitting and receiving  symbols $a^n\in {\mb A}^n, b^n \in {\mb B}^n$ over the  channel    is a  measurable function $c_{0,n}:{\mb A}^n\times{\mb  B}^{n-1} \longmapsto [0,\infty)$. The set of channel  input distributions with transmission cost is defined by 
\begin{align}
{\cal P}_{[0,n]}(\kappa) \tri  & \Big\{  P_i(da_i|a^{i-1}, b^{i-1}) \in {\cal K}({\mb A}_i| {\mb A}^{i-1}\times {\mb  B}^{i-1}),  i=0, \ldots, n: \nonumber \\
&\frac{1}{n+1} {\bf E}^{ P} \Big( c_{0,n}(A^n, B^{n-1}) \Big)\leq  \kappa\Big\}  \subset {\cal P}_{[0,n]}, \hst  c_{0,n}(a^n, b^{n-1}) \tri \sum_{i=0}^n \gamma_i(T^ia^n, T^ib^{n-1}), \;      \kappa \in [0,\infty) \label{rc1}
\end{align} 
where ${\bf E}^{ P }\{\cdot\}$ denotes expectation with respect to the the joint distribution, and superscript ``P'' indicates its dependence  on the choice of  conditional distribution $\{P_i(da_i|a^{i-1}, b^{i-1}) : i=0, \ldots, n\} \in {\cal P}_{[0,n]}$.

The characterization of feedback capacity $C_{A^\infty \rar B^\infty}^{FB}(\kappa)$, is investigated as a consequence of the following definition of FTFI capacity characterization.  \\

\begin{definition} (Extremum problem with feedback)\\
\label{def-gsub2sc}
Given  any  channel  distribution from the class ${\cal C}_{[0,n]}$, 
find the {\it Information Structure } of the optimal channel input distribution $\big\{P_i(da_i|a^{i-1}, b^{i-1}): i=0, \ldots, n\big\}  \in  {\cal P}_{[0,n]}(\kappa) $ (assuming it exists) of the extremum problem defined by
\begin{align}
 C_{A^n \rar B^n}^{FB}(\kappa) \tri \sup_{\big\{P_i(da_i|a^{i-1}, b^{i-1}): i=0,\ldots, n\big\} \in {\cal P}_{[0,n]}(\kappa) } I(A^n\rar B^n),  \hst   I(A^n \rar B^n) =(\ref{CIS_6}).  \label{prob2tc}
 \end{align}
If no transmission cost is imposed the optimization in (\ref{prob2tc}) is carried out  over  ${\cal P}_{[0,n]}$,  and  $C_{A^n \rar B^n}^{FB}(\kappa)$ is replaced by $C_{A^n \rar B^n}^{FB}$.  
\end{definition}

Clearly, for each time $i$ the largest information structure of the channel input conditional distribution of extremum  problem   $C_{A^n \rar B^n}^{FB}(\kappa)$   is ${\cal I}_i^P \tri \{a^{i-1}, b^{i-1}\}, i=1, \ldots, n, {\cal I}_0^P \tri \{b^{-1}\}$. 

{\bf Variational Equality of Directed Information.} Often, in extremum problems of information theory, upper or lower bounds are introduced and then shown to be achievable over specific sets of distributions, such as, in entropy maximization with and without constraints, etc. In any extremum problem of capacity with feedback (resp. without feedback), identifying achievable upper bounds on  directed information $I(A^n \rar B^n)$ (resp. mutual information $I(A^n; B^n)$) is not an easy task. However,  by invoking  a variational equality  of directed information \cite{charalambous-stavrou2013aa} (resp. mutual information \cite{blahut1972}), such achievable upper bounds can be identified. 


Indeed, Step 2 of the proposed  Two Step procedure (discussed in Section~\ref{intro}) is based on utilizing the variation equality of directed information, 
 given in the next theorem.\\

\begin{theorem}(Variational Equality-Theorem~IV.1 in \cite{charalambous-stavrou2013aa}.)\\
\label{thm-var} 
Given a channel input conditional distribution $\big\{P_i(da_i|a^{i-1},b^{i-1}): i=0,1, \ldots, n\big\} \in {\cal P}_{[0,n]}$, a channel distribution  $\big\{Q_i(db_i|b^{i-1},a^i) : i=0,1, \ldots, n\big\} \in {\cal C}_{[0,n]} $, and the initial distribution $\mu(db^{i-1})$,  define the  corresponding  joint and  marginal distributions    by  (\ref{CIS_2gde2new})-(\ref{CIS_3a}).\\
Let    $\big\{V_i(db_i|b^{i-1}) \in {\cal M}({\mathbb B}_i): i=0, \ldots, n\big\}$ be any arbitrary distribution. \\
 Then the following variational equality holds.
\begin{align}
I(A^n\rightarrow B^n) =\inf_{  \big\{ V_i(db_i| b^{i-1} )\in {\cal M}({\mb  B}_{i}): i=0,1, \ldots, n\big\}}\sum_{i=0}^n  \int_{{\mb A}^i\times {\mb  B}^i}   \log \Big( \frac{dQ_i(\cdot|b^{i-1},a^i) }{ dV_i(\cdot|b^{i-1})}(b_i)\Big)   {\bf P}^P(da^i,db^i)     \label{BAM52a}
\end{align}
and the infimum in  (\ref{BAM52a}) is achieved at $V_i(db_i|b^{i-1})= {\Pi}_i^P(db_i|b^{i-1}), i=0, \ldots, n$  given by (\ref{CIS_3g})-(\ref{CIS_3a}).

\end{theorem}

The implications  of  variational equality (\ref{BAM52a}) are illustrated via  the following  identity. \\
For any arbitrary distribution $\big\{V_i(db_i|b^{i-1})\in {\cal M}({\mathbb B}_i): i=0, \ldots, n\big\}$, the following identities hold.
\begin{align}
\sum_{i=0}^n {\bf E}^{P} \Big\{  \log \Big( \frac{d Q_i(\cdot|B^{i-1},A^i) }{ d V_i(\cdot|B^{i-1})}(B_i)\Big)\Big\} =& \sum_{i=0}^n {\bf E}^P \Big\{  \log \Big( \frac{dQ_i(\cdot|B^{i-1},A^i) }{d{\Pi}_i^P(\cdot|B^{i-1})}(B_i)\Big)\Big\} +  \sum_{i=0}^n \int_{{\mb  B}^i}\log\left(\frac{ d{\Pi}_i^P(\cdot|b^{i-1})}{ V_i(\cdot|b^{i-1})}(b_i)\right) {\Pi}_{0,i}^P(db^i)\\
 =& I(A^n \rar B^n)\: + \: \sum_{i=0}^n \int_{{\mb  B}^i}\log\left(\frac{ d{\Pi}_i^P(\cdot|b^{i-1})}{ V_i(\cdot|b^{i-1})}(b_i)\right) {\Pi}_{0,i}^P(db^i).\label{com-ve}
\end{align}
Note that  the second right hand side term in (\ref{com-ve}) is  the sum of relative entropy terms between the marginal distribution  ${\Pi}_i^P(db_i|b^{i-1})$ defined by  the joint distribution ${\bf P}^P(da^i, db^i)$ (i.e., the correct conditional channel output distribution)   and any arbitrary distribution  $V_i(db_i|b^{i-1})$ (i.e.,  incorrect channel output conditional distribution)  for $i=0,1, \ldots, n$.  \\
Identity  (\ref{com-ve})  implies  the minimization of its left hand side over any arbitrary channel output distribution $\big\{V_i(db_i|b^{i-1})\in{\cal M}({\mb  B}_i): i=0, \ldots, n\big\}$ occurs at $V_i(db_i|b^{i-1})= {\Pi}_i^P(db_i|b^{i-1}), i=0, \ldots, n$, i.e., when the relative entropy terms are zero, giving  (\ref{BAM52a}). 

The point to be made regarding the above variational equality is that the characterization of the FTFI capacity can be transformed, for any arbitrary distribution $\big\{V_i(db_i|b^{i-1})\in {\cal M}({\mathbb B}_i): i=0, \ldots, n\big\}$, to the sequential equivalent $\sup \inf\{\cdot\}$ problem 
\begin{align}
C_{A^n \rar B^n}^{FB}(\kappa) = \sup_{    {\cal P}_{[0,n]}(\kappa) } \inf_{  \big\{ V_i(db_i| b^{i-1} )\in {\cal M}({\mb  B}_{i}): i=0,1, \ldots, n\big\}}\sum_{i=0}^n {\bf E}^P \Big\{ \log \Big( \frac{d Q_i(\cdot|B^{i-1},A^i) }{ dV_i(\cdot|B^{i-1})}(B_i)\Big) \Big\}  \label{VE_UP}
\end{align} 
Then  by removing the infimum in (\ref{VE_UP}) an upper bound is identified, which together with stochastic optimal control techniques,  is shown to be achievable over specific subsets of the set of all channel input conditional distributions satisfying conditional independence $\big\{ {\bf P}(da_i|{\cal I}_i^P), {\cal I}_i^P \subseteq \{a^{i-1}, b^{i-1}\}: i=0, \ldots, n\big\} \subseteq \big\{P_i(da_i|a^{i-1}, b^{i-1}): i=0, \ldots, n\}$ and  the average transmission cost constraint. In fact, the characterizations of the FTFI capacity formulas for various channels and transmission cost functions discussed in this paper utilize this observation.

\section{Characterization of FTFI   Capacity }
 \label{randomized}
{\bf The Two-Step Procedure.} The identification of the information structures of the optimal channel input conditional distributions and the corresponding characterizations of the FTFI capacity  $C_{A^n \rar  B^n}^{FB}$ and feedback capacity $C_{A^n \rar  B^n}^{FB}(\kappa)$, are determined by applying the following steps. 
\begin{description}
\item[\bf Step 1.] Apply stochastic optimal control techniques with relaxed or randomized  strategies (conditional distributions) \cite{hernandezlerma-lasserre1996,ahmed-charalambous2012a,charalambous-elliott1998}, to show a certain joint process which generates the information structure of the channel input conditional distribution  is an extended Markov process. This step  implies the  optimal channel input distribution occurs in specific subsets $\overline{\cal P}_{[0,n]} \subset {\cal P}_{[0,n]}$ or  $\overline{\cal P}_{[0,n]}(\kappa) \subset {\cal P}_{[0,n]}(\kappa)$, which satisfy conditional independence.

\item[\bf Step 2.] Apply  variational equality of directed information given in Theorem \ref{thm-var} (\cite{charalambous-stavrou2013aa},  Theorem~I.V.1),  to  pay-off  $I(A^n \rar B^n)$, together with stochastic optimal control techniques,  to identify  upper bounds which are achievable over  specific subsets $\sr{\circ}{ {\cal P}}_{[0,n]} \subset  \overline{\cal P}_{[0,n]}$ or $\sr{\circ}{ {\cal P}}_{[0,n]}(\kappa) \subset  \overline{\cal P}_{[0,n]}(\kappa)$, which satisfy a further conditional independence.
\end{description}
For certain channel distributions and instantaneous transmission cost functions,  Step 1 is sufficient to identify the information structures of  channel input distributions (i.e., Class A channels and transmission cost functions),  and to characterize the FTFI capacity, while for others, Step 1 may serve as  an intermediate step prior to applying   Step 2. For example, if the channel distribution is of  limited memory with respect to channel outputs, i.e.,  of the Class B, by applying Step 2   an upper bound on the FTFI capacity is obtained, which together with stochastic optimal control techniques, it is shown to be achievable over channel input distributions with limited memory on channel outputs. \\
It is also possible to apply Steps 1 and 2 jointly; this will be illustrated in specific applications.  

Step 1 is  a generalization of equivalent methods often applied  in stochastic optimal Markov decision or control problems to show  that optimizing a   pay-off \cite{vanschuppen2010,kumar-varayia1986} over all possible non-Markov policies or strategies,  occurs in the smaller set of Markov policies. However, certain issues should be treated with caution, when stochastic optimal control techniques are applied in extremum problems of information theory. These are summarized in the next remark.\\

{\bf Comments on stochastic optimal control in relation to extremum problems of information theory.} 
In fully observable stochastic optimal control theory \cite{hernandezlerma-lasserre1996}, one is given a controlled process $\{X_i: i=0, \ldots, n\}$, often called the state process taking values in $\big\{{\mathbb X}_i: i=0, \ldots, n\big\}$,  affected by  a control process $\{U_i: i=0, \ldots,n \}$ taking values in $\big\{{\mathbb U}_i: i=0, \ldots, n\big\}$, and the corresponding  control object ${\cal P}_{[0,n]}^{CO}\tri   \big\{{\bf P}_{U_i|U^{i-1}, X^{i}}: i=0, \ldots, n\big\}$ and the controlled object ${\cal C}_{[0,n]}^{CO}\tri \big\{{\bf P}_{X_i|X^{i-1}, U^{i-1}}: i=0, \ldots, n\big\}$. \\
Often, the controlled object is Markov  conditional on the past control values, that is,  ${\bf P}_{X_i|X^{i-1}, U^{i-1}}={\bf P}_{X_i|X_{i-1}, U_{i-1}}-a.a. (x^{i-1}, u^{i-1}), i=0, \ldots, n$. Such Markov controlled objects are often  induced by  discrete recursions 
\bea
X_{i+1}=f_i(X_{i}, U_i, V_{i}), \hso X_0=x_0, \hso i=0, \ldots, n, \label{SOC_3}
\eea
 where $\{V_i: i=0, \ldots, n\}$ is an independent noise process taking values in $\big\{{\mathbb V}_i: i=0, \ldots, n\big\}$, independent of the initial state $X_0$. Denote the set of such Markov distributions  or controlled objects by  ${\cal C}_{[0,n]}^{CO-M}\tri \big\{{\bf P}_{X_i|X_{i-1}, U_{i-1}}: i=0, \ldots, n\big\}$.\\
In stochastic optimal control theory,  one is given a sample pay-off function, often of additive form, defined by
\bea 
l: {\mathbb X}^n \times {\mathbb U}^n \longmapsto (-\infty, \infty], \hst  l(x^n, u^n)\tri  \sum_{i=0}^n\ell(u_i,x_{i}) \label{SOC_4}
\eea 
  where the functions $\big\{\ell_i(\cdot, \cdot): i=0, \ldots, n\}$ are fixed and independent of the control object $\big\{{\bf P}_{U_i|U^{i-1}, X^{i}}: i=0, \ldots, n\big\}$.\\
Given the Markov distribution ${\bf P}_{X_i|X_{i-1}, U_{i-1}}: i=0, \ldots, n$, the objective is to optimize the average of the sample path pay-off over all non-Markov strategies in ${\cal P}_{[0,n]}^{CO}$, i.e.,
\bea
J_{0,n}^F( {\bf P}_{U_i|U^{i-1}, X^{i}}^*, i=0, \ldots, n) \tri  \inf_{{\cal P}_{0,n]}^{CO}} {\bf E}\Big\{  \sum_{i=0}^n\ell(U_i,X_{i})\Big\}. \label{SOC_1}
\eea
Two features of stochastic optimal control which are distinct from  any extremum problem of directed information are discussed below.\\
{\bf Feature 1.} Stochastic optimal control formulations pre-suppose, that additional state variables are introduced, prior to arriving at the Markov controlled object $ \big\{{\bf P}_{X_i|X_{i-1}, U_{i-1}}: i=0, \ldots, n\big\}$ or the  discrete recursion (\ref{SOC_3}), and   the pay-off (\ref{SOC_4}). Specifically, the final formulation   (\ref{SOC_1}), pre-supposes the Markov controlled object  $ \big\{{\bf P}_{X_i|X_{i-1}, U_{i-1}}: i=0, \ldots, n\big\}$ is obtained as follows. The state variables which constitute the complete state process $\{X_i: i=0, \ldots, n\}$   may be due to a  noise process which was not independent and converted into an independent  noise process via state augmentation, and/or any dependence on past information, and converted to a Markov controlled object  via state augmentation, and  due to a non-single letter dependence, for each $i$, of the  the sample pay-off functions $\ell_i(\cdot, \cdot)$, which was converted into single letter dependence, i.e,  $(x_i,u_i)$, by additional  state augmentation, so that the controlled object is Markov. Such examples are given in \cite{bertsekas1995} for deterministic or non-randomized strategies, defined by 
\bea
{\cal E}_{[0,n]}^{CO} \tri \big\{e_i: {\mathbb U}^{i-1}\times {\mathbb X}^i \longmapsto {\mathbb U}_i, \hso i=0, \ldots, n: \hso   u_i=e_i(u^{i-1}, x^i), i=0, \ldots, n\big\}.
\eea
In view of the Markovian property of the controlled object, i.e., given by ${\bf P}_{X_i|X_{i-1}, U_{i-1}}, i=0, \ldots, n$, then  the optimization in (\ref{SOC_1}) reduces to the following optimization problem over Markov strategies. 
\bea
J_{0,n}^F({\bf P}_{U_i|U^{i-1}, X^{i}}^*, i=0, \ldots, n)=J_{0,n}^M({\bf P}_{U_i|X_{i}}^*, i=0, \ldots, n) \tri  \inf_{ {\bf P}_{U_i|X_{i}}, i=0, \ldots, n} {\bf E}\Big\{  \sum_{i=0}^n\ell(U_i,X_{i})\Big\}. \label{SOC_2}
\eea
This further implies that the control process $\{X_i: i=0, \ldots, n\}$ is Markov, i.e., it satisfies ${\bf P}_{X_i|X^{i-1}}= {\bf P}_{X_i|X_{i-1}}, i=0, \ldots, n$. On the other hand, if ${\bf P}_{X_i|X^{i-1}}= {\bf P}_{X_i|X_{i-1}}, i=0, \ldots, n$ then (\ref{SOC_2}) holds.

{\bf Feature 2.} Given a general controlled object $\big\{{\bf P}_{X_i|X^{i-1}, U^{i-1}}: i=0, \ldots, n\big\}$ non necessarily Markov, one of the fundamental results of classical stochastic optimal control is that optimizing the pay-off ${\bf E}\Big\{  \sum_{i=0}^n\ell(U_i,X_{i})\Big\}$ over randomized strategies ${\cal P}_{[0,n]}^{CO}$ does not incur a better performance than optimizing it over non-randomized strategies ${\cal E}_{[0,n]}^{CO}$, i.e., 
\begin{align}
J_{0,n}^F( {\bf P}_{U_i|U^{i-1}, X^{i}}^*, i=0, \ldots, n)= &  \inf_{{\cal E}_{[0,n]}^{CO} } {\bf E}\Big\{  \sum_{i=0}^n\ell(U_i,X_{i})\Big\} \label{SOC_5}\\
=& \inf_{g_i(X_i):\; i=0, \ldots, n } {\bf E}^g\Big\{  \sum_{i=0}^n\ell(U_i,X_{i})\Big\} \hst \mbox{if} \hso  {\bf P}_{X_i|X^{i-1}, U^{i-1}}={\bf P}_{X_i|X_{i-1}, U_{i-1}}-a.a., i=0, \ldots, n. \label{SOC_6}
\end{align}


\noi Step 2, i.e., the application of variational equality,  discussed in Section~\ref{intro},  is specific to  information theoretic pay-off functionals and does not have a counterpart to any of the common pay-off functionals of stochastic optimal control problems \cite{vanschuppen2010,kumar-varayia1986}. This is due to the fact that, unlike stochastic optimal control problems (discussed above, feature (1)), any extremum problem of feedback capacity involves directed information density $\iota_{A^n \rar B^n}(a^n,b^n)\equiv \iota_{A^n \rar B^n}^P(a^n,b^n)$ defined by (\ref{DI_Den_1}), which is not a fixed functional, but depends on the channel output transition probability distribution $\{{\bf P}_{B_i|B^{i-1}}\equiv \Pi_i^P(db_i|b^{i-1}): i=0, \ldots, n\}$ defined by (\ref{CIS_3a}), which depends on the channel distribution, and the channel input distribution chosen to maximize directed information $I(A^n \rar B^n)$. The nonlinear dependence of the directed  information density makes extremum problems of directed information, distinct from extremum problems of classical  stochastic optimal control.\\
This implies step 2 or more specifically,  the  variational equalities of directed information and mutual information, are key features of information theoretic pay-off functionals. Often, these variational equalities   need to be incorporated into any extremum problems  of  deriving  achievable  bounds, such as, in extremum problems of  feedback capacity and capacity without feedback, much as, it is often done when deriving achievable bounds, based on the entropy maximizing properties of distributions (i.e., Gaussian distributions). \\
Feature (2), i.e,  (\ref{SOC_5}) and (\ref{SOC_6}), do not have  counters part in any extremum problem of directed information or mutual information, that is, optimizing directed information over channel input distributions is not equivalent to optimizing directed information over deterministic non-randomized strategies. In fact, by definition the  sequence of codes defined by (\ref{block-code-nf-non}) are randomized strategies, and if these are specialized to non-randomized strategies, i.e., by removing their dependence on the randomly generated messages, $W\in {\cal M}_n$, then directed information is zero, i.e, $I(A^n \rar B^n)=0$ if $a_i=e_i(a^{i-1}, b^{i-1}), i=0, \ldots, n$.

\subsection{Channels Class A and Transmission Costs Class A or B}
\label{class_A}
First, the preliminary steps of the derivation of the characterization of FTFI capacity for any  channel distributions of Class A, (\ref{CD_C1}), without transmission cost are introduced, to gain insight into the derivations of information structures, without the need to introduce excessive notation. The analogies and differences between  stochastic optimal control theory and extremum problems of directed information, as discussed above, are  made explicit, throughout the derivations, when tools from stochastic optimal control are applied to the directed information density pay-off.  

Introduce the following definition of channel input distributions satisfying conditional independence.\\

\begin{definition}(Channel input distributions for Class A channels and transmission cost functions)\\
\label{def_CLASS_A}
Define the restricted class of channel input distributions $\overline{\cal P}^{A}_{[0,n]} \subset {\cal P}_{[0,n]}$ satisfying conditional independence  by
\begin{align}
\overline{\cal P}_{[0,n]}^{A}\tri \Big\{ \big\{  &P_i(da_i| a^{i-1}, b^{i-1}): i=0,1, \ldots, n\big\} \in  {\cal P}_{[0,n]}    :  \nonumber \\
& P_i(d{a}_i|a^{i-1}, b^{i-1})={\bf P}_{A_i|B^{i-1}}(da_i|b^{i-1}) \equiv \pi_i(d{a}_i|b^{i-1})-a.a.
  (a^{i-1}, b^{i-1}),i=0,1,\ldots,n\Big\} . \label{rest_set_1}
\end{align}
Similarly, for transmission cost functions ${\gamma}_i(T^ia^n, T^ib^{n-1})=\gamma_i^{A}(a_i, b^{i-1})$ or ${\gamma}_i(T^ia^n, T^ib^{n-1})=\gamma_i^{B.K}(a_i, b_{i-K}^{i-1}),  i=0, \ldots, n$,   define 
\bea
\overline{\cal P}_{[0,n]}^{A}(\kappa)\tri \overline{\cal P}_{[0,n]}^{A} \bigcap  {\cal P}_{[0,n]}(\kappa).
\eea
\end{definition}

\ \

From the  definition of directed information $I(A^n \rar B^n)$ given by (\ref{CIS_6}), and utilizing the channel distribution  (\ref{CD_C1}), the FTFI capacity is defined by 
\begin{align}
{C}_{A^n \rar B^n}^{FB,A} \tri   & \sup_{\big\{P_i(da_i| a^{i-1 }, b^{i-1})  :i=0,\ldots,n\big\} \in {\cal P}_{[0,n]} }
  \sum_{i=0}^n {\bf E}^{{ P}} \Big\{  \log \Big( \frac{dQ_i(\cdot|B^{i-1},A_i) }{d\Pi_i^{{ P} }(\cdot|B^{i-1})}(B_i)\Big)\Big\} \label{CIS_17}
\end{align}
where the channel output transition probability defined by   (\ref{CIS_3a}),  is given by the following expressions.
 \begin{align}
\Pi_i^{ P} (db_i| b^{i-1})=& \int_{{\mb A}^i} Q_i(db_i|b^{i-1}, a^i)\otimes {\bf P}^{ P}(da^{i}|b^{i-1}), \hso  i=0, \ldots, n \\
 =&\int_{{\mb A}^i} Q_i(db_i|b^{i-1}, a_i)\otimes P_i(da_i|a^{i-1}, b^{i-1})\otimes {\bf P}^{ P}(da^{i-1}|b^{i-1}) \label{CIS_9}\\
\sr{(\alpha)}{=}&  \Pi_i^{\pi_i}(db_i| b^{i-1})  \tri  \int_{  {\mb A}_i } Q_i(db_i|b^{i-1}, a_i) \otimes         \pi_i(da_i |b ^{i-1})  \label{CIS_10a} \\
& \hst  \hso \mbox{if}  \hso P_i(da_i|a^{i-1}, b^{i-1}) ={\bf P}(da_i| b^{i-1})\equiv  {\pi}_i(da_i | b^{i-1})-a.a.(a^{i-1}, b^{i-1})\hso  i=0, \ldots, n  \label{CIS_10}  
\end{align}
Note that identity {\it $(\alpha)$ holds if it can be shown that  conditional independence (\ref{CIS_10}) holds} for any candidate $\big\{P_i(da_i| a^{i-1}, b^{i-1}): i=0,1, \ldots, n\big\} \in {\cal P}_{[0,n]}$ maximizing $I(A^n \rar B^n)$; the superscript notation, $\Pi_i^{\pi_i}(db_i| b^{i-1})$, indicates the dependence of $\Pi_i^P(db_i| b^{i-1})$  on conditional distribution ${\bf P}(da_i|b^{i-1})\equiv \pi_i(da_i|b^{i-1})$ (instead on $\{P_j(da_j|a^{j-1}, b^{j-1}): j=0,1,\ldots, i\}$), for $i=1,\ldots, n$.\\
It is important to note that one cannot  assume $\Pi_i^{ P} (db_i| b^{i-1})=  \int_{  {\mb A}_i } Q_i(db_i|b^{i-1}, a_i) \otimes         \pi_i(da_i |b ^{i-1})\equiv \Pi_i^{\pi_i}(db_i|b^{i-1})$, that is, (\ref{CIS_9}) is given by (\ref{CIS_10a})  without proving that such restriction of conditional distributions {\it is a property of the  channel input distribution which maximizes directed information} $I(A^n \rar B^n)$, because the marginal distribution $\Pi^P(db^n)$ is uniquely defined from the joint distribution ${\bf P}^P(da^n , db^n)=\otimes_{i=0}^n \Big(Q_i(db_i|b^{i-1},a_i)\otimes P_i(da_i|a^{i-1}, b^{i-1})\Big)$. The derivation of  Theorem~1 in \cite{yang-kavcic-tatikonda2005} and Theorem~1 in \cite{yang-kavcic-tatikonda2007} for the problems considered by the authors, should  be read with caution, to account for  the above feature, in order  to show the supremum over all channel input conditional distributions occurs in the smaller set, satisfying a conditional independence condition, which is  analogous to  (\ref{CIS_10}).   

{\bf Suppose (\ref{CIS_10}) holds} (its validity is shown in Theorem~\ref{thm-ISR}). Then the expectation ${\bf E}^{ P}\{\cdot\}$ in (\ref{CIS_17}) with respect to the joint distribution simplifies as follows.   
\begin{align}
{\bf P}^{ P}(da_i, db^i) =&{\bf P}^{P}(da_i, db^i) =  Q_i(db_i|b^{i-1},a_i) \otimes {\bf P}(da_i| b^{i-1}) \otimes {\bf P}^{P}(db^{i-1}), \hso  i=0, \ldots, n  \label{CIS_13_G}  \\
=&{\bf P}^{ \pi}(da_i, db^i)   \hso \mbox{if (\ref{CIS_10}) holds}, \hso  i=0,1, \ldots, n,    \label{CIS_13_GG} \\
=& Q_i(db_i|b^{i-1},a_i) \otimes \pi_i(da_i| b^{i-1}) \otimes \Pi_{0,i-1}^{\pi_0, \ldots, \pi_{i-1}}(db^{i-1}), \hso i=0, \ldots, n, \label{CIS_13} \\
\Pi_{0,i-1}^{\pi_0, \ldots, \pi_{i-1}}(db^{i-1})=& \Pi_{i-1}^{\pi_{i-1}}(db_{i-1}| b^{i-2}) \otimes  \Pi_{0, i-2}^{\pi_0, \ldots, \pi_{i-2}}(db^{i-2}), \hso  i=0, \ldots, n \label{CIS_14}
\end{align}
where the superscript notation, ${\bf P}^{\pi}(da_i, b^i)$, indicates the dependence of joint distribution ${\bf P}^P(da_i, b^i)$ on $\{\pi_j(da_j|b^{j-1}): j=0, 1, \ldots, i\}$, for $i=0,\ldots, n$. Clearly, if   (\ref{CIS_10a}) holds, for each $i$, the controlled  conditional distribution-controlled object,    $\Pi_i^{ P}(db_i| b^{i-1})=\Pi_i^{\pi_i}(db_i|b^{i-1})$, depends on the  channel distribution   $Q_i(db_i|b^{i-1},a_i)$ and the control conditional  distribution-control object,  
${\pi}_i(da_i | b^{i-1})$, for $i=0,1,\ldots, n$.\\
Thus, the following holds.
\begin{itemize}
\item {\bf Channel  Class A.1, (\ref{CD_C1}):} If the maximizing channel input conditional distribution satisfies conditional independence   $P_i(da_i|a^{i-1}, b^{i-1})={\pi}_i(da_i | b^{i-1})-a.a.(a^{i-1}, b^{i-1}), i=0, \ldots, n$, then  directed information (\ref{CIS_6}) is a functional of \\ $\{Q_i(db_i| b^{i-1}, a_i), \pi_i(da_i|b^{i-1}): i=0, \ldots, n\}$ and it is given by
  \begin{align}
I(A^n\rar B^n) &= \sum_{i=0}^n I(A_i;B_i|B^{i-1}) =\sum_{i=0}^n {\bf E}^{\pi} \Big\{  \log \Big( \frac{dQ_i(\cdot|B^{i-1},A_i) }{d\Pi_i^{\pi_i}(\cdot|B^{i-1})}(B_i)\Big)\Big\} \label{CIS_15}\\ 
&=\sum_{i=0}^n \int \log \Big( \frac{dQ_i(\cdot|b^{i-1},a_i) }{d\Pi_i^{\pi_i}(\cdot|b^{i-1})}(b_i)\Big) Q_i(db_i|b^{i-1},a_i) \otimes \pi_i(da_i| b^{i-1}) \otimes \Pi_{0,i-1}^{\pi_0, \ldots, \pi_{i-1}}(db^{i-1})  \label{CIS_15_new}\\ 
&\equiv  {\mathbb I}_{A^n\rar B^n}(\pi_i, Q_i: i=0, \ldots, n) \label{CIS_16}
\end{align}
where  $ {\bf E}^{\pi}\{ \cdot\}$ indicates that the joint distribution over which expectation is taken depends on  the sequence of conditional distributions  $\{\pi_j(da_j|b^{j-1}): j=0,1, \ldots,i\}$, for  $i=0, \ldots, n$. 
\end{itemize}

\noi By (\ref{CIS_15}), since the expectation is taken with respect to joint distribution (\ref{CIS_13}), the distribution  $\{\pi_i(da_i| b^{i-1}):i=0,\ldots,n\}$ is  indeed the
control object (conditional distribution), chosen to control the conditional distribution  of the channel output process,  $\{\Pi_i^{\pi_i}(db_i| b^{i-1}) :i=0,1,\ldots,n\}$. By analogy with stochastic optimal control with randomized strategies, for each $i$, the conditional distribution  $\Pi_i^{\pi_i}(db_i|b^{i-1})$ is  affected by the control object $\pi_i(da_i | b^{i-1})$, for $i=0,\ldots,n$, and this  is chosen  to influence  the pay-off  (\ref{CIS_16}),  which is a  functional of $\{\pi_i(da_i| b^{i-1}):i=0,\ldots,n\}$ (since the channel is fixed).\\
Next, it is shown that (\ref{CIS_10a}) is indeed valid, i.e.,  the maximization of $I(A^n \rar B^n)$ over $\{  P_i(da_i|a^{i-1}, b^{i-1}):  i=0, \ldots, n\}$ occurs in the smaller set  
$\overline{\cal P}_{[0,n]}^A \subset {\cal P}_{[0,n]}$.
\\

 \begin{theorem}(Characterization of FTFI capacity for  channels of class A)\\
  \label{thm-ISR}
Suppose the channel distribution is of Class A defined by   (\ref{CD_C1}). Then the  following hold.\\
{\bf  Part A.} The maximization of $I(A^n \rar B^n)$ over  ${\cal P}_{[0,n]}$ occurs in $\overline{\cal P}_{[0,n]}^{A} \subset {\cal P}_{[0,n]} $ and  the characterization of FTFI capacity is given by the following expression.
\begin{align}
{C}_{A^n \rar B^n}^{FB,A} 
  = \sup_{\big\{\pi_i(da_i | b^{i-1}) \in {\cal M}({\mb A}_i) : i=0,\ldots,n\big\}}
  \sum_{i=0}^{n}{\bf E}^{\pi}\Big\{\log\Big(\frac{dQ_i(\cdot| B^{i-1}, A_i)}
  {d\Pi_i^{\pi_i}(\cdot| B^{i-1})}(B_i)\Big)\Big\} \label{CIS_18}
  \end{align}
  where  $\Pi_i^{\pi_i}(db_i| b^{i-1}) = \int_{ {\mb A}_i} Q_i(db_i |b^{i-1}, a_i) \otimes   {\pi}_i(da_i | b^{i-1})$ and the joint distribution over which ${\bf E}^{\pi}\{\cdot\}$ is taken is $\big\{  {\bf P}^{\pi}(da_i, db^i): i=0, \ldots, n\big\}$ defined by  (\ref{CIS_13}).\\
{\bf Part B.} Suppose the following two conditions hold.
\begin{align}
(a)& \hso {\gamma}_i(T^ia^n, T^ib^{n-1})=\gamma_i^{A}(a_i, b^{i-1}) \hso \mbox{or} \hso {\gamma}_i(T^ia^n, T^ib^{n-1})=\gamma_i^{B.K}(a_i, b_{i-K}^{i-1}), \hso i=0, \ldots, n,    \label{CIS_19} \\
(b)& \hso  C_{A^n \rar B^n}^{FB, A}(\kappa) \tri \sup_{\big\{P_i(da_i| a^{i-1 }, b^{i-1}) :i=0,\ldots,n\big\}\in  {\cal P}_{[0,n]}(\kappa) } I(A^{n}
\rar {B}^{n})   \label{ISDS_6b}   \\
=& \inf_{s \geq  0} \sup_{ \big\{ P_i(da_i|a^{i-1}, b^{i-1}): i=0,\ldots, n\big\}\in {\cal P}_{[0,n]}   } \Big\{ I(A^{n}
\rar {B}^{n})  - s \Big\{ {\bf E}^P\Big(c_{0,n}(A^n,B^{n-1})\Big)-\kappa(n+1) \Big\}\Big\}. \label{ISDS_6cc}
\end{align}   
The maximization of $I(A^{n}\rar {B}^{n})$ over channel input distributions  with transmission cost $\big\{P_i(da_i| a^{i-1 }, b^{i-1}) :i=1,\ldots,n\big\}\in  {\cal P}_{[0,n]}(\kappa)  $ occurs in $\overline{\cal P}_{[0,n]}^{A}(\kappa)$, 
and the FTFI capacity is given by the following expression.
\begin{align}
{C}_{A^n \rar B^n}^{FB,A}(\kappa)   
  = \sup_{\pi_i(da_i | b^{i-1}) \in {\cal M}({\mb A}_i), i=0,\ldots,n:  \frac{1}{n+1} {\bf E}^{\pi}\Big\{c_{0,n}(A^n, B^{n-1})\Big\} \leq \kappa  }
  \sum_{i=0}^{n}{\bf E}^{\pi}\Big\{\log\Big(\frac{dQ_i(\cdot| B^{i-1}, A_i)}
  {d\Pi_i^{\pi_i}(\cdot| B^{i-1})}(B_i)\Big)\Big\}. \label{CIS_18_TC}
  \end{align}

\end{theorem}
\begin{proof} The derivation is based on expressing directed information as a functional of the  channel input conditional distribution,  identifying the explicit dependence of the sample path pay-off on appropriate state variable, and then showing the controlled object, which is defined using the state variable is Markov, i.e., as discussed earlier. \\{ Part A.}  By the channel distribution assumption (\ref{CD_C1}), the following equalities are obtained. 
 \begin{align}
I(A^n\rar B^n) =& \sum_{i=0}^n {\bf E}^{{ P}} \Big\{  \log \Big( \frac{dQ_i(\cdot|B^{i-1},A^i) }{d\Pi_i^{{  P} }(\cdot|B^{i-1})}(B_i)\Big)\Big\} \label{CIS_5a}\\ 
 \sr{(\alpha)}{=}& \sum_{i=0}^n \int_{{\mb A}_{i} \times {\mb  B}^{i}}^{}   \log \Big( \frac{ dQ_i(\cdot|B^{i-1}, A_i) }{d\Pi_i^{{  P}}(\cdot|B^{i-1})}(B_i)\Big) {\bf P}^{ P}(dA_i, dB^i)  \label{CIS_6b} \\
 \sr{(\beta)}{=} &\sum_{i=0}^{n}  {\bf E}^{ P} \bigg\{    \log\bigg(\frac{dQ_i(\cdot| B^{i-1}, A_i)}
  {d\Pi_i^{ P}(\cdot | B^{i-1})}(B_i)\bigg)\bigg\}\label{CIS_6c}\\
 \sr{(\gamma)}{=}& \sum_{i=1}^{n}  {\bf E}^{ P} \bigg\{  {\bf E}^{ P}\Big\{\log \Big(\frac{dQ_i(\cdot| B^{i-1},A_i)}
  {d\Pi_i^{ P}(\cdot| B^{i-1})}(B_i)\Big)\bigg{|}A^i, B^{i-1}\Big\} \bigg\}\label{CIS_6d}   \\
  \sr{(\delta)}{=}& \sum_{i=0}^{n}  {\bf E}^{ P}\bigg\{   {\bf E}^{ P}\Big\{\log \Big(\frac{dQ_i(\cdot| B^{i-1}, A_i)}
  {\Pi_i^{ P}(\cdot| B^{i-1})}(B_i)\Big)\bigg{|}A_i, B^{i-1}\Big\}\bigg\}\label{CIS_6e}\\
\sr{(\epsilon)}{=}& \sum_{i=0}^{n} {\bf E}^{ P}\bigg\{  \ell_i^{ P}\Big(A_i, S_i\Big)\bigg\}, \hst S_j\tri B^{j-1}, \hso j=0, \ldots, n, \label{CIS_6f}\\
  \ell_j^{ P}\Big(a_j, s_j \Big)  \tri & \int_{ {\mb B}_j } \log \Big(\frac{dQ_j(\cdot| s_j, a_j)}
  {d\Pi_j^{ P}(\cdot| s_j)}(b_j)\Big) Q_j(db_j| s_j, a_j), \hso j=0, \ldots, n \label{BAM31}
  \end{align}
where \\
$(\alpha)$ is due to the  channel distribution assumption (\ref{CD_C1}); \\
$({\beta})$ is by definition;\\
$(\gamma)$ is due to a property of  expectation;\\
$(\delta)$ is due to the  channel distribution  (\ref{CD_C1});\\
$({\epsilon})$ is by definition of conditional expectation for the measurable function $\ell_i^{ P}(\cdot, \cdot)$  defined by (\ref{BAM31}). \\
 The validity of the claim that the optimal channel input conditional distribution belongs to the class $\overline{\cal P}_{[0,n]}^{A}$, establishing validity of the claimed identity (\ref{CIS_10a}), and consequently validity of (\ref{CIS_15})-(\ref{CIS_16}), is shown as follows.    
Since  for each $i$, the channel $Q_i(db_i|b^{i-1}, a_i)$ and the  pay-off function $\ell_i^{ P}(a_i, s_i)\equiv \ell_i^{ P}(a_i, b^{i-1})$ in (\ref{CIS_6f}) depend on $s_i \tri b^{i-1}$ for $i=0, 1, \ldots,n$, then  $\{S_i\tri B^{i-1}:  i=0, \ldots, n\}$ is the  controlled process,   control by the control process $\{A_i: i=0, \ldots, n\}$ (see discussion on stochastic optimal control). That is,  $\big\{{\bf P}(ds_{i+1}|s^i, a^i): i=0, \ldots, n-1\big\}$ is the controlled object. 
Next, we show  the controlled process $\{S_i: i=0, \ldots, n\}$ is Markov, i.e., the transition probabilities $\big\{{\bf P}(ds_{i+1}|s^{i}):i=0, \ldots, n-1\big\}$ are Markov, and   the maximization of directed information occurs in the set $\overline{\cal P}_{[0,n]}^{A}$, defined by (\ref{rest_set_1}).   By applying   Bayes' theorem and using the definition of the channel distribution,  the  following conditional independence  are easily shown. 
\begin{align}
&{\bf P}(ds_{i+1}|s^i, a^i)= {\bf P}(ds_{i+1}|s_i, a_i), \hso i=0, \ldots, n-1, \label{CO_M}  \\
&{\bf P}(ds_{i+1}|s^{i})=  {\bf P}(ds_{i+1}|s_i)=\int_{{\mb A}_i}  {\bf P}(ds_{i+1}| s_i, a_i) \otimes{\bf P}(da_i|s_i). \label{MT_Case_1_a}
\end{align}
In fact, since the controlled object is Markov, i.e., (\ref{CO_M}) holds, the statement of the theorem follows directly from  stochastic optimal control \cite{vanschuppen2010,hernandezlerma-lasserre1996}, (see discussion on stochastic optimal control). Nevertheless, we give the complete derivation.\\
In view of the above identities the Markov process $\{S_i: i=0, \ldots, n\}$ satisfies the following identity.
\begin{align}
{\bf P}^P(ds_{i+1})= \int_{{\mathbb B}^{i-1}\times {\mathbb A}_i} {\bf P}(ds_{i+1}|s_i, a_i)\otimes {\bf P}(da_i|s_i)\otimes {\bf P}^P(ds_i)\hso  \Longrightarrow \hso  {\bf P}^\pi(ds_{i+1})= \int_{{\mathbb B}^{i-1}\times {\mathbb A}_i}  {\bf P}(ds_{i+1}|s_i, a_i)\otimes \pi_i(da_i|s_i)\otimes {\bf P}^\pi(ds_i) \label{MV_1}
\end{align}
where ${\bf P}^P(ds_{i+1})= {\bf P}^{\pi_0, \ldots, \pi_i}(ds_{i+1})\equiv {\bf P}^{\pi}(ds_{i+1}) $ follows  by iterating the first equation in (\ref{MV_1}). 
Since, the process $\{S_i: i=0,1, \ldots, n\}$ is Markov with transition probability given by the right hand side of (\ref{MT_Case_1_a}), then for $i=0,\ldots, n-1$, the distribution  ${\bf P}(ds_{i+1}|s_i)$ is controlled  by the control object ${\bf P}(da_i|s_i)\equiv {\bf P}(da_i|b^{i-1})$.   Clearly, (\ref{MT_Case_1_a})  implies that any measurable function say, $\overline{\xi}(s_i)$ of $s_i=b^{i-1}$ is affected by the control object ${\bf P}(da_i|s_i)$, and hence by   (\ref{CIS_9}), then  $\{ \xi_{i}(s_i)\tri  \Pi_i^{P}(db_i|b^{i-1})\equiv     \Pi_i^{\pi_i}(db_i|b^{i-1}): i=0, \ldots, n\}$, and this  transition distribution  is controlled by  the control object  $\{{\pi}_i(da_i| b^{i-1}): i=0, \ldots, n\}$. Utilizing this in (\ref{BAM31}) and (\ref{CIS_6f}), the following is obtained.
\begin{align}
I(A^n\rar B^n) =  \sum_{i=0}^n \int_{{\mb A}_{i} \times {\mb  B}^{i}}^{}   \log \Big( \frac{ dQ_i(\cdot|b^{i-1}, a_i) }{d\Pi_i^{{  \pi}}(\cdot|b^{i-1})}(b_i)\Big)Q_i(db_i|b^{i-1},a_i)\otimes \pi(da_i|b^{i-1})\otimes {\bf P}^\pi(db^{i-1}) . \label{CIS_6b_FF} 
 \end{align}
Thus,  the maximization in (\ref{CIS_6b_FF}) over all channel input distributions  is done by choosing the control object $\{\pi_i(da_i | b^{i-1}):i=0,\ldots,n\}$  to control the conditional distribution $\{\Pi^{\pi_i}_i(db_i| b^{i-1}):i=1,\ldots,n\}$,
which for each $i$, depends on $\big\{b^{i-1},\pi_i(da_i| b^{i-1})\big\}$, for $i=0,\ldots, n$. Hence,  the maximizing object in  (\ref{CIS_17}) (if it exists),  is of the 
 the form $P^*_i(da_i | a^{i-1},b^{i-1})={\pi}^*_i(da_i| b^{i-1})-a.a. (a^{i-1}, b^{i-1}), i=1,\ldots,n\}$, and  (\ref{CIS_18}) is obtained. \\
{  Part B.}    Since  for each $i$, the channel conditional distribution $Q_i(db_i|\cdot, \cdot)$ is measurable with respect to ${\cal I}_{i}^{Q}=\{b^{i-1}, a_i\}$ and the  transmission cost   $\gamma_i(\cdot, \cdot)\equiv \gamma_i^A(\cdot, \cdot)$ or $\gamma_i(\cdot, \cdot)\equiv \gamma_i^{B.K}(\cdot, \cdot)$  is  measurable with respect to  ${\cal I}_{i}^{\gamma}=\{a_i, b^{i-1}\}$ or ${\cal I}_{i}^{\gamma}=\{a_i, b_{i-K}^{i-1}\}$ for $i=0, 1, \ldots, n$, the results follow directly from Part A, as follows. Consider the cost function $\{\gamma_i^A(a_i, b^{i-1}), i=0, \ldots, n\}$, and  note that the average cost constraint can be expressed as follows. 
\begin{align}
 \sum_{i=0}^n {\bf E}^{{ P}} \Big\{ \gamma_i^{A}(A_i, B^{i-1}) \Big\} 
 =& \sum_{i=0}^{n}  {\bf E}^{ P} \bigg\{  {\bf E}^{ P}\Big\{ \gamma_i^{A}(A_i, B^{i-1})     \bigg{|}A^i, B^{i-1}\Big\} \bigg\}\label{COST_1}   \\
 =& \sum_{i=0}^{n}  {\bf E}^{\pi}\bigg\{    \overline{\gamma}_i^{A, \pi}(S_i)     \bigg\}, \hso 
 \overline{\gamma}_j^{A, \pi}(s_j)  \tri  \int_{ {\mb A}_j } \gamma_j^A(a_j, s_j) \pi_j(da_j| s_j), \hso j=0, \ldots, n. \label{COST_2}
  \end{align}
where the expectation ${\bf E}^\pi\{\cdot\cdot\}$ is taken with respect to ${\bf P}^\pi(ds_i)$ and this follows from Part A. Since the transmission cost constraint is expressed via (\ref{COST_2}) and depends on the distribution $\{\pi_i(da_i|s_i): i=0, \ldots, n\}$ then the claim holds.\\
 Alternatively, if condition (b) holds then  the Lagrangian of the unconstraint problem (omitting the term $\kappa (n+1)$) is
\bea
I(A^{n}
\rar {B}^{n})  - s  {\bf E}^P\Big(c_{0,n}(A^n,B^{n-1})\Big)=\sum_{i=0}^{n} {\bf E}^{ P}\bigg\{  \ell_i^{ P}\Big(A_i, S_i\Big)  - s\gamma_i^A(A_i,B^{n-1})  \bigg\}
\eea
and the rest of the derivation follows from Part A. 
 This completes the prove. For the cost function $\{\gamma_i^{B.K}(a_i, b_{i-K}^{i-1}), i=0, \ldots, n\}$, since $b_{i-K}^{i-1}$ is the restriction of $s_i=b^{i-1}$ to a subsequence, the above conclusion holds as well.
\end{proof}

Next, we give some comments regarding the previous theorem and discuss possible generalizations.\\

\begin{remark}(Some generalizations)\\
(1) Suppose the channel is $\big\{Q_i(db_i|b_{i-1}, a_i): i=0, \ldots, n\}$ and the transmission cost function is ${\gamma}_i(T^ia^n, T^ib^{n-1})=\gamma_i^{A}(a_i, b^{i-1}), i=0, \ldots, n$. Then  the statements of  Theorem~\ref{thm-ISR}, Part B, remain valid with $Q_i(db_i|b^{i-1}, a_i)$ replaced by $Q_i(db_i|b_{i-1}, a_i)$ for $i=0, \ldots, n$, because the state is $s_i=b^{i-1}$,  and this is determined from the dependence of the  cost function $\gamma_i^{A}(a_i, b^{i-1})$ on $s_i$, for  $i=0, \ldots, n$. \\
(2) Suppose in Theorem~\ref{thm-ISR}, Part B, ${\gamma}_i(T^ia^n, T^ib^{n})=\gamma_i^{A}(a_i, b^{i}), i=0, \ldots, n$, then  from  (\ref{COST_1}), (\ref{COST_2}) we have 
\begin{align}
 \sum_{i=0}^n {\bf E}^{{ P}} \Big\{ \gamma_i^{A}(A_i, B^{i}) \Big\}
 = \sum_{i=0}^{n}  {\bf E}^{\pi}\bigg\{    \overline{\gamma}_i^{A, \pi}(S_i)     \bigg\}, \hso 
 \overline{\gamma}_j^{A, \pi}(s_j)  \tri  \int_{ {\mb A}_j \times {\mathbb B}_j } \gamma_j^A(a_j, s_j)Q_j(db_j|s_j, a_j)\otimes \pi_j(da_j| s_j), \hso j=0, \ldots, n. \label{COST_3}
  \end{align}
Hence, the statements of  Theorem~\ref{thm-ISR}, Part B remain valid. 
\end{remark}

\ \

\begin{remark}(Equivalence of constraint and unconstraint problems)\\
The equivalence of constraint and unconstraint problems in Theorem~\ref{thm-ISR}, follows from  Lagrange's duality theory of optimizing convex functionals over convex sets  \cite{dluenberger1969}. Specifically, from \cite{charalambous-stavrou2013aa}, it follows that the set of distributions ${\bf P}^{C1}(da^n|b^{n-1})\tri \otimes_{i=0}^n P_i(da_i|a^{i-1}, b^{i-1}) \in {\cal M}({\mb A}^n)$ is convex, and this uniquely defines ${\cal P}_{[0,n]}$ and vice-versa,  directed information as a functional of  ${\bf P}^{C1}(da^n|b^{n-1}) \in {\cal M}({\mb A}^n)$ is convex, and  by the linearity the constraint set ${\cal P}_{[0,n]}(\kappa)$ expressed in ${\bf P}^{C1}(da^n|b^{n-1})$, is convex. Hence, if their exists a maximizing distribution and the so-called Slater condition holds (i.e., a sufficient condition is the existence of an interior point to the constraint set), then  the constraint and unconstraint problems are equivalent. For finite alphabet spaces all conditions are  easily checked.
\end{remark}


 Next,  the variational equality of Theorem~\ref{thm-var} is applied, together with stochastic optimal control techniques, to provide an alternative  derivation of Theorem~\ref{thm-ISR}, through upper bounds which are achievable, over the smaller  set of channel input conditional distributions $\overline{\cal P}_{[0,n]}^{A}$. The next theorem is simply introduced to illustrate the importance of the variational equality of directed information in extremum problems of feedback capacity, and to illustrate  its importance, when considering  channels with limited memory on past channel output symbols.   \\

\begin{theorem}
\label{Rem_Case_1}
(Derivation of Theorem~\ref{thm-ISR} via variational equality)\\
Consider the extremum problem of Class A channels defined by (\ref{CIS_17}) (which is investigated in Theorem~\ref{thm-ISR}, Part A). \\Let $\big\{ V_i(\cdot|b^{i-1})\in {\cal M}({\mathbb B}_i): i=0, \ldots, n \big\}$ be a sequence of conditions distributions on $\big\{{\mathbb B}_i: i=0, \ldots, n\big\}$, not necessarily the one generated by the channel and channel input conditional distributions.\\
 Then 
\begin{align}
{C}_{A^n \rar B^n}^{FB,A}   = & \sup_{\big\{\pi_i(da_i|  b^{i-1})  :i=0,\ldots,n\big\} }     \inf_{\big\{ V_i(db_i|b^{i-1})\in {\cal M}({\mathbb B}_i): i=0, \ldots, n \big\}}   
  \sum_{i=0}^n {\bf E}^{{\pi }} \Big\{  \log \Big( \frac{dQ_i(\cdot|B^{i-1},A_i) }{dV_i(\cdot|B^{i-1})}(B_i)\Big)\Big\} \label{CIS_17_VI_A_1} \\
  = & \sup_{\big\{\pi_i(da_i | b^{i-1}) \in {\cal M}({\mb A}_i) : i=0,\ldots,n\big\}}
  \sum_{i=0}^{n}{\bf E}^{\pi}\Big\{\log\Big(\frac{dQ_i(\cdot| B^{i-1}, A_i)}
  {dV_i^*(\cdot| B^{i-1})}(B_i)\Big)\Big\}, i=0, \ldots, n \label{CIS_18_VI_A1}
  \end{align}
where  $\big\{ V_i^*(\cdot|b^{i-1})\in {\cal M}({\mathbb B}_i): i=0, \ldots, n \big\}$ is given by
\begin{align}
V_i^*(db_i|b^{i-1}) \tri \Pi_i^\pi(db_i|b^{i-1}) = \int_{{\mb A}_i}Q_{i}
(d{b}_i|b^{i-1},a_i) \otimes {\pi}_i(d{a}_i|b^{i-1}), \ \ 
i=0,1,\ldots,n .\label{Case_1_in_6}
\end{align}
\end{theorem}
\begin{proof}
By the variational equality of Theorem~\ref{thm-var},  for any arbitrary conditional distribution,  $V_{i}(\cdot|b^{i-1})\in{\cal M}({\mb  B}_i), i=0,1, \ldots, n$, it can be shown that the following identity holds.
\begin{align}
&\sup_{     \big\{   P_i(da_i| a^{i-1 }, b^{i-1})  :i=0,\ldots,n\big\} \in {\cal P}_{[0,n]} }    \sum_{i=0}^{n}
\int \log \left( \frac{dQ_i(\cdot| b^{i-1}, a_i)}{d\Pi^{P}_i(\cdot| b^{i-1})}(b_i)\right){\bf P}^{P}(d{b}_i,db^{i-1},da_i)   \label{Case_1_in_3} \\
=& \sup_{ \big\{ P_i(da_i| a^{i-1 }, b^{i-1}): i=0, \ldots, n \big\} \in {\cal P}_{[0,n]} }         \inf_{ \big\{  V_{i}(db_i|b^{i-1}) \in {\cal M}({\mb B}_i)\big\}_{i=0}^n} \sum_{i=0}^n
\int \log\left(\frac{dQ_i(\cdot|b^{i-1}, a^i)}{dV_i(\cdot|b^{i-1})}(b_i)\right){\bf P}^P(b_i, b^{i-1},a_i) \label{Case_1_in_4} \\
\leq & \sup_{ \big\{ P_i(da_i| a^{i-1 }, b^{i-1}): i=0, \ldots, n \big\} \in {\cal P}_{[0,n]} }        \sum_{i=0}^n
\int \log\left(\frac{dQ_i(\cdot|b^{i-1}, a^i)}{dV_i(\cdot|b^{i-1})}(b_i)\right){\bf P}^P(b_i, b^{i-1},a_i), \hso \forall    V_{i}(db_i|b^{i-1}),  i=0, \ldots, n \label{Case_1_in_5}
\end{align}
where the last inequality holds for any arbitrary distribution $V_{i}(\cdot|b^{i-1})\in{\cal M}({\mb B}_i), i=0,1, \ldots, n$, not necessarily the one generated by $\big\{ P_i(da_i| a^{i-1 }, b^{i-1}):i=0,1,\ldots, n \big\} \in {\cal P}_{[0,n]}$ and the channel distribution.  \\
Next, define the pay-off function
\bea
\ell(a_i, b^{i-1}) \tri \int_{{\mathbb B}_i} \log\left(\frac{d{Q}_i(\cdot|b_{i-1}, a_i)}{dV_i(\cdot|b^{i-1})}(b_i)\right) Q_i(db_i|b^{i-1}, a_i) \equiv \ell(a_i,s_i), \hst s_i=b^{i-1}, \hst i=0, \ldots, n
\eea
Since ${\bf P}^{P}(d{b}_i,db^{i-1},da_i)=Q_i(db_i|db^{i-1},a_i)\otimes {\bf P}(da_i|b^{i-1})\otimes {\bf P}^{P}(db^{i-1})$, for any arbitrary $V_{i}(\cdot|b^{i-1})\in{\cal M}({\mb B}_i), i=0,1, \ldots, n$, by maximizing the right hand side of  (\ref{Case_1_in_5}) the following  upper bound is obtained.
\begin{align}
 (\ref{Case_1_in_3})       \leq &   \sup_{     \big\{   P_i(da_i| a^{i-1 }, b^{i-1})  :i=0,\ldots,n\big\} \in {\cal P}_{[0,n]} }    \sum_{i=0}^{n}
\int \log \left( \frac{dQ_i(\cdot| b^{i-1}, a_i)}{dV_i(\cdot| b^{i-1})}(b_i)\right)Q_i(db_i|db^{i-1},a_i)\otimes {\bf P}(da_i|b^{i-1})\otimes {\bf P}^{P}(db^{i-1})   \label{Case_1_in_7} \\
=& \sup_{     \big\{   P_i(da_i| a^{i-1 }, b^{i-1})  :i=0,\ldots,n\big\} \in {\cal P}_{[0,n]} }    \sum_{i=0}^{n} {\bf E}^P \Big\{\ell(A_i, B^{i-1})\Big\} \label{VE_thm_1} \\
\sr{(\alpha)}{=}&\sup_{\big\{ \pi_i(da_i|b^{i-1})\in {\cal M}({\mb A}_i|): i=0, \ldots, n\big\} } \sum_{i=0}^n \int\log\left(\frac{d{Q}_i(\cdot|b_{i-1}, a_i)}{dV_i(\cdot|b^{i-1})}(b_i)\right){\bf P}^{{\pi}}(d{b}_i,d{b}^{i-1},d{a}_i), \hso \forall V_i(db_i|b^{i-1}), i=0, \ldots, n \label{Case_1_in_8}
\end{align}
where ${\bf P}^{\pi}(d{b}_i,db^{i-1},da_i)=Q_i(db_i|db^{i-1},a_i)\otimes \pi_i(da_i|b^{i-1})\otimes {\bf P}^{\pi}(db^{i-1}), i=0, \ldots, n$, 
 and the equality in $(\alpha)$  is obtained as follows.
Since for each $i$, the pay-off over which the expectation is taken in (\ref{VE_thm_1}) is $\ell(a_i, b^{i-1})\equiv \ell(a_i,s_i)$ and $\{S_i: i=0, \ldots, n\}$ is Markov, as shown  in the proof of Theorem~\ref{thm-ISR}, then the maximization occurs in the subset satisfying conditional independence $P_i(da_i|a^{i-1}, s_i)={\bf P}(da_i|s_i)\equiv \pi_i(da_i|s_i)-a.a.(a^{i-1}, s_i), i=0, \ldots, n$, hence ${\bf P}^{P}(d{b}_i,db^{i-1},da_i)={\bf P}^{\pi}(d{b}_i,db^{i-1},da_i), i=0, \ldots, n$, and  (\ref{Case_1_in_8}) is obtained. \\
Since the distribution 
$\{V_i(db_i|b^{i-1}): i=0, \ldots, n\}$ is arbitrary, by letting this to be the one defined by the channel and $\big\{ \pi_i(da_i|b^{i-1})\in {\cal M}({\mb A}_i): i=0, \ldots, n\big\}$, given by  (\ref{Case_1_in_6}), i.e.,  $V_i(db_i|b^{i-1})\tri \Pi_i^\pi(db_i|b^{i-1}), i=0, \ldots, n$, then the following  upper bound holds.
\begin{align}
 (\ref{Case_1_in_3})       \leq   
\sup_{\big\{ \pi_i(da_i|b^{i-1})\in {\cal M}({\mb A}_i): i=0,1, \ldots, n\big\} } \sum_{i=0}^n \int\log\left(\frac{d{Q}_i(\cdot|b_{i-1}, a_i)}{d\Pi_i^\pi(\cdot|b^{i-1})}(b_i)\right){\bf P}^{{\pi}}(d{b}_i,d{b}^{i-1},d{a}_i) \label{Case_1_in_8_old}
\end{align}
Next, it is shown,   that  the reverse inequality in (\ref{Case_1_in_8_old}) holds, thus establishing the claim. Recall definition (\ref{rest_set_1})  of $\overline{\cal P}_{[0,n]}^{A}$. Since $\overline{\cal P}_{[0,n]}^{A}\subset {\cal P}_{[0,n]}$, it can be shown that the following inequality holds. 
\begin{align}
&\sup_{\big\{\pi_i(da_i | b^{i-1}) \in {\cal M}({\mb A}_i) : i=0,\ldots,n\big\}}      \sum_{i=0}^{n}
\int\log\left(\frac{dQ_i(\cdot| b^{i-1}, a_i)}{d\Pi^{{\pi_i}}_i(\cdot| b^{i-1})}(b_i)\right){\bf P}^{\pi}(d{b}_i,db^{i-1},da_i)    \label{Case_1_in_1}\\
\leq & 
\sup_{\big\{P_i(da_i| a^{i-1 }, b^{i-1})  :i=0,\ldots,n\big\} \in {\cal P}_{[0,n]} }    \sum_{i=0}^{n}
\int\log\left(\frac{dQ_i(\cdot| b^{i-1}, a_i)}{d\Pi^{P}_i(\cdot| b^{i-1})}(b_i)\right){\bf P}^{P}(d{b}_i,db^{i-1},da_i)= (\ref{Case_1_in_4})  \label{Case_1_in_2}
\end{align}
where $\{\Pi_i^P(d{b}_i| b^{i-1}), \Pi_i^{\pi_i}(d{b}_i| b^{i-1}):i=0,\ldots,n\}$ are defined by (\ref{CIS_9}), (\ref{CIS_10a}), and $\{ {\bf P}^P(d{b}_i,d{b}^{i-1},d{a}_i): i=0, \ldots, n\}$, $\{{\bf P}^{\pi}(d{b}_i,d{b}^{i-1},d{a}_i):i=0,\ldots,n\}$ are induced by the channel and $\{P_i(d{a}_i|a^{i-1}, {b}^{i-1}): i=0, \ldots, n\}$, $\{{\pi}_i(d{b}_i|{b}_{i-1}):i=0,\ldots,n\}$, respectively.\\
 Combining inequalities   (\ref{Case_1_in_2}) and  (\ref{Case_1_in_8_old})  establishes the equality in  (\ref{CIS_18_VI_A1}), under (\ref{Case_1_in_6}). 
\end{proof}

Note that Theorem~\ref{Rem_Case_1} can be used to derive  Theorem~\ref{thm-ISR}, {Part B},    by repeating the above derivation,  with  the supremum over the set ${\cal P}_{[0,n]}$ replaced by the set ${\cal P}_{[0,n]}(\kappa)$ in all equations.

\subsection{Channels Class B and Transmission Costs Class A or B} 
\label{class_B}
 In this section, the information structure of  channel input distributions, which  maximize $I(A^n \rar B^n)$ is derived for channel  distributions of Class B and transmission cost functions Class A or B.
The derivation is based on applying  the results of Section~\ref{class_A},
%
 and   the variational equality of directed information, to show the supremum over all channel input conditional distributions  occurs in a smaller set $\sr{\circ}{\cal P}_{[0,n]}\subseteq  \overline{\cal P}_{[0,n]}^A \subset {\cal P}_{[0,n]}$ (for Class B transmission costs the subset is strictly smaller, i.e.,  $\sr{\circ}{\cal P}_{[0,n]}\subset  \overline{\cal P}_{[0,n]}^A$).

The derivation is first  presented for any  channel distribution of Class B and transmission cost of Class B, with $M=2, K=1$, to illustrate the procedure, as the derivation of the general cases are  similar.

\subsubsection{\bf Channel Class B  and Transmission Cost Class B, $M=2, K=1$}
\label{case_3}
First, consider any channel distribution of Class B, with $M=2$, i.e.,  $Q_i(db_i|b_{i-1},b_{i-2}, a_i), i=0,1, \ldots, n$, without transmission cost. \\
Then the FTFI capacity is defined by  
\begin{align}
  {C}_{A^n \rar B^n}^{FB,B.2}  \tri & \sup_{\big\{P_i(da_i | a^{i-1}, b^{i-1}) \ : i=0,\ldots,n\big\} \in {\cal P}_{[0,n]}  }
  \sum_{i=0}^{n}{\bf E}^{P}\Big\{\log\Big(\frac{dQ_i(\cdot| B_{i-1}, B_{i-2}, A_i)}
  {d\Pi_i^{P}(\cdot| B^{i-1})}(B_i)\Big)\Big\} \label{UMC_rd1a}\\
\sr{(\alpha)}{=}& \sup_{\big\{\pi_i(da_i | b^{i-1}) \in {\cal M}({\mb A}_i) : i=0,\ldots,n\big\}}
  \sum_{i=0}^{n}{\bf E}^{\pi}\Big\{\log\Big(\frac{dQ_i(\cdot| B_{i-1},B_{i-2}, A_i)}
  {d\Pi_i^{\pi_i}(\cdot| B^{i-1})}(B_i)\Big)\Big\} \label{UMC_rd1}\\
 \Pi_i^{\pi_i}(db_i| b^{i-1}) =& \int_{ {\mb A}_i} Q_i(db_i |b_{i-1},b_{i-2}, a_i) \otimes   {\pi}_i(da_i | b^{i-1}), \hso i=0,1, \ldots, n \label{UMC_rdd2}
\end{align}
where $(\alpha)$ is due to Theorem~\ref{thm-ISR},  because the set of channel distributions of  Class B  is a subset of the set of channel distributions of Class A, and 
  the joint distribution over which ${\bf E}^{\pi}\{\cdot\}$ is taken is ${\bf P}^{\pi}(da_i, db^i)\equiv  {\bf P}^{\pi_0, \pi_1, \ldots, \pi_i}(da_i, db^i)$, $i=0,1,\ldots, n$.\\
  The main challenge is to show  the optimal channel input distribution induces the following conditional independence on the transition probability of the channel output process: ${\bf P}(db_i|b^{i-1})={\bf P}(db_i|b_{i-1}, b_{i-2})-a.a.b^{i-1}, i=0, \ldots, n$.

This is shown by invoking,   {\it Step 2}, of the two-step procedure (i.e., the variational equality of directed information), to  deduce   that the maximization  in (\ref{UMC_rd1}) occurs  in $\sr{\circ}{\cal P}_{[0,n]}^{B.2}\tri  \{\pi_i^2(da_i|b_{i-1}, b_{i-2}): i=0, \ldots, n\}  \subset \overline{\cal P}_{[0,n]}^{A} \tri  \{\pi_i(da_i|b^{i-1}): i=0, \ldots, n\}$, and that $\Pi_i^{\pi_i}(db_i|b^{i-1})= \nu_i^{\pi^2}(db_i|b_{i-1}, b_{i-2})-a.a. b^{i-1}, i=0, \ldots, n$, that is, the optimal channel input distribution satisfies conditional independence property, $\pi_i(da_i|b^{i-1})=\pi_i^2(da_i|b_{i-1}, b_{i-2})-a.a. b^{i-1}, i=0, \ldots, n$. \\

\begin{lemma}(Characterization of FTFI capacity for channels of class B  and transmission costs of class B, $M=2, K=1$)\\
\label{thm_BAM_2}
Suppose the channel distribution is of Class B with $M=2$. \\
Define the restricted class of policies $\sr{\circ}{\cal P}_{[0,n]}^{B.2} \subset \overline{\cal P}_{[0,n]}^{A}$ by
\begin{align}
\sr{\circ}{\cal P}_{[0,n]}^{B.2}\tri \Big\{ \big\{  \pi_i(d{a}_i| b^{i-1}):i=0,1, \ldots, n\big\}\in  \overline{\cal P}_{[0,n]}^{A}   :   {\pi}_i(d{a}_i| b^{i-1})={\pi}_i^2(d{a}_i|b_{i-1}, b_{i-2})-a.a.
  b^{i-1},i=0,1,\ldots,n\Big\} . \nonumber 
\end{align}
Then the following hold.\\
{\bf Part A.} The maximization in (\ref{UMC_rd1a}) over  $\{P_i(d{a}_i| a^{i-1}, b^{i-1}):i=0, \ldots, n\} \in {\cal P}_{[0,n]}$ occurs in the smaller class 
$ \sr{\circ}{\cal P}_{[0,n]}^{B.2}$, that is, it satisfies the following conditional independence. 
\bea
P_i(d{a}_i| a^{i-1},b^{i-1})={\pi}_i^{2}(d{a}_i|b_{i-1}, b_{i-2})-a.a. \ (a^{i-1},b^{i-1}),\hso  i=0,1,\ldots,n.  \label{gioler41}
\eea
Moreover, any distribution from the class  $ \sr{\circ}{\cal P}_{[0,n]}^{B.2} $ induces a  channel output process $\{B_i: i=0,1, \ldots, n\}$, with  conditional probabilities which  are second-order Markov, that is,  
\bea
\Pi_i^P(db_i|b^{i-1}) =v_i^{\pi^2}(db_i|b_{i-1}, b_{i-2})-a.a. \hso b^{i-1}, \hso i=0, 1, \ldots, n, \label{UMC_RF_5}
\eea 
and the characterization of FTFI capacity is given by  the following expression. 
\begin{align}
 C_{A^n \rar B^n}^{FB, B.2} 
=&\sup_{ \{\pi_i^M(da_i| b_{i-1}, b_{i-2})\in {\cal M}({\mb A}_i), i=0, \ldots, n\} }\sum_{i=0}^{n}{\bf E}^{\pi^2}\left\{
\log\left(\frac{dQ_i(\cdot|B_{i-1},B_{i-2}, A_i)}{dv_i^{{\pi}^2}(\cdot|B_{i-1}, B_{i-2})}(B_i)\right)\right\}\label{gioler43}\\
\equiv &\sup_{ \{\pi_i^2(da_i|b_{i-1}, b_{i-2})\in {\cal M}({\mb A}_i), i=0, \ldots, n\} }\sum_{i=0}^{n}
I(A_i;B_i|B_{i-1}, B_{i-2})\label{gioler44} 
\end{align}
where
\begin{align}
v_i^{{\pi}^2}(d{b}_i|b_{i-1}, b_{i-2})=&\int_{{\mb A}_i}Q_i(d{b}_i|b_{i-1},b_{i-2},a_i) \otimes {\pi}_i^2(d{a}_i|b_{i-1}, b_{i-2}), \hso i=0,1,\ldots,n. \label{gioler46}
\end{align}
{\bf  Part B.}  Suppose the following two conditions hold.
\begin{align}
(a)& \hso {\gamma}_i(T^ia^n, T^ib^{n-1})=\gamma_i^{B.1}(a_i, b_{i-1}), \hso i=0, \ldots, n,    \label{UMC_rf20} \\
(b)& \hso C_{A^n \rar B^n}^{FB, B.1}(\kappa) \tri  \sup_{\big\{P_i(da_i | a^{i-1}, b^{i-1}) \ : i=0,\ldots,n\big\} \in {\cal P}_{[0,n]}(\kappa)   }
  \sum_{i=0}^{n}{\bf E}^{P}\Big\{\log\Big(\frac{dQ_i(\cdot| B_{i-1},B_{i-2}, A_i)}
  {d\Pi_i^{P}(\cdot| B^{i-1})}(B_i)\Big)\Big\}   \label{UMC_rf21}   \\
&= \inf_{s\geq 0} \sup_{ \{P_j(\cdot| a^{j-1}, b^{j-1}): j=0,\ldots, n\}   \in  {\cal P}_{[0,n]} } \Bigg\{ I(A^{n}
\rar {B}^{n}) - s \Big\{ {\bf E}^P\Big(c_{0,n}(A^n,B^{n-1})\Big)-\kappa(n+1) \Big\}\Bigg\}. \label{UMC_rf22}
\end{align}   
The  characterization of FTFI capacity is given by the following expression.
\begin{align}
C_{A^n \rar B^n}^{FB, B.2}(\kappa)=\sup_{\pi_i^2(da_i|b_{i-1}, b_{i-2})\in {\cal M}({\mb A}_i), i=0, \ldots, n: \frac{1}{n+1} {\bf E}^{\pi^2}\big\{ \sum_{i=0}^n \gamma_i^{B.1}(A_i, B_{i-1})\big\} \leq \kappa  }\sum_{i=0}^{n}
I(A_i;B_i|B_{i-1}, B_{i-2})\label{gioler44_TC} 
\end{align}
That is, in { Part A, B} the conditional distribution of the joint process $\{(A_i, B_i): i=0, 1, \ldots, n\}$ satisfies (\ref{SC_1}) and the channel output process $\{B_i: i=0,1, \ldots, n\}$ is a second-order Markov process, i.e., its conditional distribution  satisfies (\ref{SC_2}).

\end{lemma}
\begin{proof} { Part A. } Since the channel distribution of Class B, with $M=2$,    is a special case the channel  distribution (\ref{CD_C1}),  the statements of Theorem~\ref{thm-ISR} hold,   hence (\ref{UMC_rd1a})-(\ref{UMC_rdd2}) hold, and  the maximization  over  ${\cal P}_{[0,n]}$ occurs in the preliminary set $\overline{\cal P}_{[0,n]}^{A}$ defined by (\ref{rest_set_1}).  
By (\ref{UMC_rd1}), and since $\sr{\circ}{\cal P}_{[0,n]}^{B.2}\subset \overline{\cal P}_{[0,n]}^{A}$, then
\begin{align}
C_{A^n \rar B^n}^{FB, B.2} =&
\sup_{  \big\{\pi_i(da_i|b^{i-1}): i=0, \ldots, n\big\} }\sum_{i=0}^{n}
\int\log\left(\frac{dQ_i(\cdot| b_{i-1},b_{i-2},a_i)}{d\Pi^{{\pi}}_i(\cdot| b^{i-1})}(b_i)\right){\bf P}^{{\pi}}(d{b}_i,db^{i-1},da_i)   \label{ineq_1} \\
\geq & \sup_{  \big\{\pi_i(da_i|b^{i-1}): i=0, \ldots, n\big\}\in \sr{\circ}{\cal P}_{[0,n]}^{B.2} }   \sum_{i=0}^{n}
\int\log\left(\frac{dQ_i(\cdot| b_{i-1},b_{i-2}, a_i)}{d\Pi^{{\pi}}_i(\cdot| b^{i-1})}(b_i)\right){\bf P}^{{\pi}}(d{b}_i,db^{i-1},da_i) \label{ineq_2}   \\
\geq &    \sum_{i=0}^{n}
\int\log\left(\frac{dQ_i(\cdot| b_{i-1},b_{i-2}, a_i)}{d\nu^{{\pi^2}}_i(\cdot| b_{i-1}, b_{i-2})}(b_i)\right){\bf P}^{{\pi^2}}(d{b}_i,db_{i-1},b_{i-2}, da_i), \hso \forall \pi_i^2(da_i|b_{i-1}, b_{i-2}) \in  {\cal M}({\mb A}_i), i=0, \ldots,n   \label{ineq_3}
\end{align}
where $\{(\Pi_i^{\pi}(d{b}_i| b^{i-1}),v_i^{{\pi}^2}(d{b}_i| b_{i-1}, b_{i-2})):i=0,\ldots,n\}$ are given by (\ref{UMC_rdd2}), (\ref{gioler46}), and \\$\{{\bf P}^{\pi}(d{b}_i,d{b}^{i-1},d{a}_i)$, ${\bf P}^{{\pi}^2}(d{b}_i,d{b}_{i-1},b_{i-2}, d{a}_i):i=0,\ldots,n\}$ are induced by the channel and $\{{\pi}_i(d{a}_i| {b}^{i-1}),{\pi}_i^2(d{b}_i|{b}_{i-1}, b_{i-2}):i=0,\ldots,n\}$, respectively. Taking the supremum over $\{\pi_i^2(da_i|b_{i-1}, b_{i-2}): i=0, \ldots, n\}$, inequality  (\ref{ineq_3}) is retained and hence, the following lower bound is obtained.
\begin{align}
C_{A^n \rar B^n}^{FB, B.2} = (\ref{ineq_1}) \geq  &  
\sup_{ \{\pi_i^2(da_i|b_{i-1}, b_{i-2}): i=0, \ldots, n\}}\sum_{i=0}^{n}
\int\log\left(\frac{dQ_i(\cdot| b_{i-1},b_{i-2},a_i)}{dv^{{\pi}^2}_i(\cdot| b_{i-1},b_{i-2})}(b_i)\right){\bf P}^{{\pi}^2}(d{b}_i,db_{i-1},b_{i-2},da_i)  \label{ineq_5} \\
\equiv & \sup_{ \{\pi_i^2(da_i|b_{i-1}, b_{i-2}): i=0,\ldots, n\} }\sum_{i=0}^{n}
I(A_i;B_i|B_{i-1}, B_{i-2}).        \label{gioler47}
\end{align}
Next, the  variational equality of Theorem~\ref{thm-var} is applied to show the reverse inequality in  (\ref{ineq_5}) holds. Given a policy from  the set 
$ \overline{\cal P}_{[0,n]}^{A}$, and any arbitrary distribution $\{{V}_{i}(db_i|b^{i-1}): i=0, \ldots, n\}\in{\cal M}({\mb  B}_i): i=0, \ldots, n\}$, then 
\begin{align}
C_{A^n \rar B^n}^{FB, B.2} =&
\sup_{  \{\pi_i(da_i|b^{i-1}): i=0, \ldots, n\} }\sum_{i=0}^{n}
\int\log\left(\frac{dQ_i(\cdot| b_{i-1},b_{i-2},a_i)}{d\Pi^{{\pi}}_i(\cdot| b^{i-1})}(b_i)\right){\bf P}^{{\pi}}(d{b}_i,db^{i-1},da_i)   \label{ineq_10} \\
=&\sup_{\{\pi_i(da_i|b^{i-1}): i=0, \ldots, n\} } \hso \inf_{ \{ {V}_i(db_i|b^{i-1})\in{\cal M}( {\mb  B}_i): i=0, \ldots, n \} }    \sum_{i=0}^n \int \log\left(\frac{dQ_i(\cdot| b_{i-1}, b_{i-2},a_i)}{d{V}_i(\cdot|b^{i-1})}(b_i)\right){\bf P}^{\pi}(db_i, db^{i-1}, da_i) \label{gioler48aa}
 \end{align}
where $\{{\bf P}_i^{\pi}(b_i, b^{i-1}, d{a}_i): i=0, \ldots, n\}$ is defined by the channel distribution  and  $\{\pi_i(da_i|b^{i-1}): i=0,1, \ldots, n\} \in \overline{\cal P}_{[0,n]}^{A}$.
 Since  $\{{V}_{i}(db_i|b^{i-1})\in{\cal M}({\mb  B}_i): i=0, \ldots, n\}$ is arbitrary, then an upper bound for (\ref{gioler48aa}) is obtained as follows. Assume the arbitrary channel output conditional probability is the one satisfying the conditional independence 
\begin{align}
{V}_{i}(db_i|b^{i-1}) = \overline{V}_{i}(db_i|b_{i-1}, b_{i-2})-a.a.b^{i-1}, \hso  
i=0,1,\ldots,n. \label{gioler49bb_B}
\end{align}
Define the pay-off 
\begin{align}
\ell_i(a_i,s_i) \tri \int_{{\mb B}_i} \log\left(\frac{dQ_i(\cdot|s_i, a_i)}{d\overline{V}_i(\cdot|s_i)}(b_i)\right){Q}_i(db_i|s_i, a_i), \hso s_i \tri (b_{i-1}, b_{i-2}), \hso i=0, \ldots, n.
\end{align}
Then, by removing the infimum in  (\ref{gioler48aa}) over  $\{ {V}_i(db_i|b^{i-1})\in{\cal M}( {\mb  B}_i): i=0, \ldots, n\}$, and substituting (\ref{gioler49bb_B}), the following  upper bound is obtained.
\begin{align}
C_{A^n \rar B^n}^{FB, B.2} 
\leq& \sup_{ \{\pi_i(da_i|b^{i-1}): i=0, \ldots, n\} }   \sum_{i=0}^n \int \log\left(\frac{dQ_i(\cdot| b_{i-1},b_{i-2}, a_i)}{d\overline{V}_i(\cdot|b_{i-1}, b_{i-2})}(b_i)\right){\bf P}^{\pi}(db_i, db_{i-1},db_{i-2}, da_i), \hso \forall \overline{V}_i(db_i|b_{i-1}, b_{i-2}), i=0, \ldots, n 
 \label{gioler50a}\\
\sr{(\alpha)}{=} & \sup_{ \{\pi_i(da_i|b^{i-1}): i=0, \ldots, n\} }   \sum_{i=0}^n {\bf E}^{\pi} \Big\{\int_{{\mb B}_i} \log\left(\frac{dQ_i(\cdot| b_{i-1},b_{i-2}, a_i)}{d\overline{V}_i(\cdot|b_{i-1}, b_{i-2})}(b_i)\right){Q}_i(db_i| b_{i-1},b_{i-2}, a_i)\Big\}, \hso \forall \overline{V}_i(db_i|b_{i-1}, b_{i-2}), i=0, \ldots, n \label{gioler50aa} \\
\equiv & \sup_{ \{\pi_i(da_i|b^{i-1}): i=0, \ldots, n\} }   \sum_{i=0}^n {\bf E}^{\pi} \Big\{  \ell_i(A_i,S_i) \Big\}, \hso \forall \overline{V}_i(db_i|s_i),\hso  i=0, \ldots, n \label{gioler50aa_a}\\
\sr{(\beta)}{=} & \sup_{    \{\pi_i^2(da_i|b_{i-1}, b_{i-2}): i=0, \ldots, n\} } \sum_{i=0}^n \int\log\left(\frac{d{Q}_i(\cdot|b_{i-1}, b_{i-2}, a_i)}{d \overline{V}_i(\cdot|b_{i-1}, b_{i-2})}(b_i)\right){\bf P}^{{\pi}^2}(d{b}_i,d{b}_{i-1},db_{i-2}, d{a}_i), \hso \forall \overline{V}_i(db_i|b_{i-1}, b_{i-2}), i=0, \ldots, n  \label{gioler50}
\end{align}
where $(\alpha)$ is by definition, and   $(\beta)$ is obtained as follows. Since the pay-off in  (\ref{gioler50aa_a}), i.e, $\ell_i(\cdot, \cdot)$ is a function of  $(a_i,s_i)$, for $i=0, \ldots, n$, then $\{S_i: i=0, \ldots, n\}$ is the state process controlled by $\{A_i: i=0, \ldots, n\}$.  Moreover,    by virtue of  Bayes' theorem, and the channel definition, the  following identity holds. 
\bea
{\bf P}(ds_{i+1}|s^{i}, a^i)= {\bf P}(ds_{i+1}|s_i,a_i), \hso i=0, \ldots, n-1.\label{MT_Case_1_a_A}
\eea
In view of the  Markov structure of the controlled process $\{S_i: i=0, \ldots, n\}$, i.e., (\ref{MT_Case_1_a_A}),  then the expectation of the pay-off in (\ref{gioler50aa_a}) is given by 
\begin{align}
 \sum_{i=0}^n {\bf E}^{\pi} \Big\{  \ell(A_i,S_i) \Big\}=\int_{{\mathbb B}_{i-1}\times {\mathbb B}_{i-2} \times {\mathbb A}_i} \ell(a_i, s_i) {\bf P}(da_i|s_i) {\bf P}^\pi(ds_i), \hst \forall \overline{V}_i(db_i|s_i),\hso  i=0, \ldots, n. \label{gioler50aa_a_a}
\end{align}
Thus, $\big\{{\bf P}(ds_{i+1}|s_i,a_i), \hso i=0, \ldots, n-1\big\}$ is the controlled object and by the discussion on classical stochastic optimal control, i.e., Feature 1,  the 
supremum over $\{\pi_i(da_i|b^{i-1}): i=0, \ldots, n\} $ in (\ref{gioler50a}),  satisfies $\pi_i(da_i|b^{i-1})=\pi_i^{2}(da_i|b_{i-1}, b_{i-2})-a.a.b^{i-1}, i=0, \ldots, n$.\\
Alternatively, this is shown directly as follows. Note that  
\bea
{\bf P}^\pi(ds_i)= \int {\bf P}(ds_i|s_{i-1}, a_{i-1}) {\bf P}(da_{i-1}|s_{i-1}) {\bf P}^\pi (ds_{i-1})  \hso \Longrightarrow \hso {\bf P}^{\pi^2}(ds_i)= \int {\bf P}(ds_i|s_{i-1}, a_{i-1}) {\bf P}(da_{i-1}|s_{i-1}) {\bf P}^{\pi^2} (ds_{i-1})
\eea
that is, ${\bf P}^\pi(ds_i)\equiv {\bf P}^{\pi^2}(ds_i)$, depends on $\big\{{\bf P}(da_{j}|s_j)\equiv \pi_j^{2}(da_j|s_j): j=0, \ldots, i-1\}$,  and hence the right hand side in (\ref{gioler50aa_a_a}) depends on $\big\{{\bf P}(da_{j}|s_j)\equiv \pi_j^{2}(da_j|s_j): j=0, \ldots, i\}$. This implies,  the 
supremum over $\{\pi_i(da_i|b^{i-1}): i=0, \ldots, n\} $ in (\ref{gioler50a}),  satisfies $\pi_i(da_i|b^{i-1})=\pi_i^{2}(da_i|b_{i-1}, b_{i-2})-a.a.b^{i-1}, i=0, \ldots, n$, that is, the controlled object is second-order Markov, and consequently,  ${\bf P}^\pi(db_i, db_{i-1},db_{i-2}, da_i)= Q_i(db_i|b_{i-1},b_{i-2}, a_i) \otimes \pi_i^{2}(da_i|b_{i-1}, b_{i-2}) \otimes {\bf P}^{\pi^2}(db_{i-1}, db_{i-2})\equiv {\bf P}^{\pi^2}(db_i, db_{i-1},db_{i-2}, da_i), i=0, \ldots, n$.
Hence,  (\ref{gioler50}) is obtained.   \\
Since $\{ {V}_i(db_i|b^{i-1})\in{\cal M}( {\mb  B}_i): i=0, \ldots, n\}$ satisfying   (\ref{gioler49bb_B}), is an arbitrary distribution, let 
\begin{align}
{V}_{i}(db_i|b^{i-1})\tri &\int_{{\mb A}_i}Q_{i}
(d{b}_i|b_{i-1},b_{i-2}, a_i)\otimes {\bf P}^\pi(d{a}_i|b^{i-1})  \equiv {V}_{i}^\pi(db_i|b^{i-1})\nonumber \\
=&\overline{V}_i(db_i|b_{i-1}, b_{i-2}) \tri   \int_{{\mb A}_i}Q_{i}
(d{b}_i|b_{i-1},b_{i-2}, a_i)\otimes \pi_i^2(d{a}_i|b_{i-1}, b_{i-2})=\overline{V}_{i}^{\pi^2}(db_i|b_{i-1}, b_{i-2} )\\
\equiv& \nu_i^{\pi^2}(db_i|b_{i-1}, b_{i-2})-a.a. b^{i-1}, \ \ 
i=0,1,\ldots,n. \label{gioler49bb}
\end{align}
Then by substituting (\ref{gioler49bb}) into  (\ref{gioler50}), the following inequality is obtained. 
\begin{align}
C_{A^n \rar B^n}^{FB, B.2} 
\leq&  
\sup_{ \{\pi_i^2(da_i|b_{i-1}, b_{i-2}): i=0, \ldots, n\}}\sum_{i=0}^{n}
\int\log\left(\frac{dQ_i(\cdot| b_{i-1},b_{i-2}, a_i)}{dv^{{\pi}^2}_i(\cdot| b_{i-1}, b_{i-2})}(b_i)\right){\bf P}^{{\pi}^2}(d{b}_i,db_{i-1},db_{i-2},da_i)\\ 
\equiv & \sup_{ \{\pi_i^2(da_i|b_{i-1}, b_{i-2}): i=0,\ldots, n\} }\sum_{i=0}^{n}
I(A_i;B_i|B_{i-1}, B_{i-2})\label{equl_20}
\end{align}
Combining (\ref{gioler47})  and (\ref{equl_20}),   the supremum over     $\{\pi_i(da_i|b^{i-1}): i=0, \ldots, n\}$  in $C_{A^n \rar B^n}^{FB, B.2} $ occurs in the subset 
$ \sr{\circ}{\cal P}_{[0,n]}^{B.2}$, and hence  (\ref{gioler41})-(\ref{gioler46}) are a consequence of this fact.
{ Part B.} Using the definition of the transmission cost function  (\ref{UMC_rf20}) and (\ref{MT_Case_1_a_A}), then 
\begin{align}
 \sum_{i=0}^n {\bf E}^{{ P}} \Big\{ \gamma_i^{B.1}(A_i, B_{i-1}) \Big\} =\sum_{i=0}^n {\bf E}^{{\pi}} \Big\{ \gamma_i^{B.1}(A_i, B_{i-1}) \Big\}
 =\sum_{i=0}^{n} {\bf E}^{\pi^2}\Big\{\int_{  {\mathbb A}_i}   \gamma_i^{B.1}(a_i, B_{i-1}) \pi^2(a_i| B_{i-1})\Big\}
  \end{align}
  where the last equality is due to $S_i=(B_{i-1},B{i-2}), i=0, \ldots, n$ and the sample path pay-off depends only on $B_{i-1}$.
  The above expectation is a function of $\{\pi_i^2(da_i|b_{i-1}, b_{i-2}): i=0, \ldots, n\}$, hence 
by  { Part A}, the results are obtained. This completes the prove.
\end{proof}

The following remark clarifies certain aspects of the application of variational equality.\\

\begin{remark}(On the application of variational equality in Lemma~\ref{thm_BAM_2})\\
(a) The important point to be made regarding Lemma~\ref{thm_BAM_2} is that, for any channel of Class B with $M=2$ and transmission cost of class B with $K=1$,  the information structure  of channel input conditional distribution, which  maximizes directed information $I(A^n \rar B^n)$  is  ${\cal I}_i^P=\{b_{i-1}, b_{i-2}\}, i=0, 1, \ldots, n$, and  it is  determined by $\max\{K, M\}$.\\ 
(b) From Lemma~\ref{thm_BAM_2}, it follows that if the channel is replaced by $\{Q_i(db_i|b_{i-1}, a_i): i=0, \ldots, n\}$, the information structure  of channel input conditional distribution, which  maximizes directed information $I(A^n \rar B^n)$  is  ${\cal I}_i^P=\{b_{i-1}\}, i=0, 1, \ldots, n$, and the corresponding characterization of FTFI capacity is 
\bea
C_{A^n \rar B^n}^{FB, B.1} 
=
\sup_{ \big\{\pi_i^M(da_i|b_{i-1})\in {\cal M}({\mb A}_i), i=0, \ldots, n\big\} }\sum_{i=0}^{n}
I(A_i;B_i|B_{i-1}).
\eea
(c) By Lemma~\ref{thm_BAM_2}, if the channel is memoryless (i.e., $M=0$), and the transmission cost constraint is $\frac{1}{n+1} {\bf E}\big\{ \sum_{i=0}^n \gamma_i^{B.1}(A_i, B_{i-1})\big\} \leq \kappa$, then the  information structure  corresponding to the channel input conditional distribution, which  maximizes directed information $I(A^n \rar B^n)$,  is  ${\cal I}_i^P=\{b_{i-1}\}, i=0, 1, \ldots, n$, and the corresponding characterization of FTFI capacity is 
\bea
C_{A^n \rar B^n}^{FB, B.1}(\kappa) =
\sup_{\pi_i^1(da_i|b_{i-1})\in {\cal M}({\mb A}_i), i=0, \ldots,n: \frac{1}{n+1} {\bf E}^{\pi^1}\big\{ \sum_{i=0}^n \gamma_i^{B.1}(A_i, B_{i-1})\big\} \leq \kappa }\sum_{i=0}^{n} 
\int\log\Big(\frac{dQ_i(\cdot|a_i)}{dv^{{\pi}^1}_i(\cdot| b_{i-1})}(b_i)\Big){\bf P}^{{\pi}^1}(d{b}_i,db_{i-1}, da_i) 
\eea 
 where
\begin{align}
v_i^{{\pi}^1}(d{b}_i|b_{i-1})=&\int_{{\mb A}_i}Q_i(d{b}_i|a_i) \otimes {\pi}_i^1(d{a}_i|b_{i-1}), \hso i=0,1,\ldots,n. 
\end{align}
(d) Memoryless Channels. If the transmission cost constraint in (c) is replaced by $\frac{1}{n+1} {\bf E}\big\{ \sum_{i=0}^n \gamma_i^{B.0}(A_i)\big\} \leq \kappa$, since the channel is memoryless, then the derivation of Lemma~\ref{thm_BAM_2} can be repeated with (\ref{gioler49bb_B}) replaced by  $
{V}_{i}(db_i|b^{i-1}) = \overline{V}_{i}(db_i)-a.a.b^{i-1},  
i=0,1,\ldots,n$, to deduce  that in all equations in (c), the optimal channel input distribution  $\pi_i^1(d{a}_i|b_{i-1})$ is replaced by ${\pi}_i(da_i), i=0, \ldots, n$, and $C_{A^n \rar B^n}^{FB, B.1}(\kappa)=C_{A^n \rar B^n}^{FB, B.0}(\kappa)=\sup_{\pi(da_i): i=0, \ldots, n}\sum_{i=0}^n I(A_i; B_i)$, as expected.  That is, it is possible to derive the capacity achieving conditional independence property of memoryless channels with feedback, directly, without first showing via the converse to the coding theorem that feedback does not increase capacity (see \cite{cover-thomas2006}).  \\
(e) The derivation of Theorem~\ref{thm_BAM_2} is easily extended  to any channel of Class B and transmission cost function of Class B (i.e., with $M, K$ arbitrary); this is done next.
\end{remark}

\subsubsection{\bf Channels Class B and Transmission Costs  Class A or B}
\label{case_4}
Consider any channel distribution of Class B defined by (\ref{CD_C4}), i.e., given by  $\{Q_i(db_i|b_{i-M}^{i-1}, a_i): i=0, 1, \ldots, n\}$. \\
Since  the induced  joint distribution is  ${\bf P}^P(da^i, db^i)=\otimes_{j=0}^i P_j(da_j|a^{j-1}, b^{j-1})\otimes Q_j(db_j|b_{j-M}^{j-1}, a_j), i=0, \ldots, n$, the FTFI capacity is defined by 
\begin{align}
{C}_{A^n \rar B^n}^{FB,B.M} \tri & \sup_{\big\{P_i(da_i| a^{i-1},  b^{i-1})  :i=0,\ldots,n\big\} \in {\cal P}_{[0,n]} } \sum_{i=0}^n \int 
\log\Big(\frac{dQ_i(\cdot|b_{i-M}^{i-1}, a_i)}{d\Pi_i^{ P}(\cdot|b^{i-1})}(b_i)\Big) {\bf P}^P(db_{i-M}^i, da_i) \label{cor-ISR_25a_c_C4}
\end{align}
where 
 \begin{align}
\Pi_i^{ P} (db_i| b^{i-1})=&  \int_{{\mb A}^i} Q_i(db_i|b_{i-M}^{i-1}, a_i)\otimes P_i(da_i|a^{i-1}, b^{i-1})\otimes {\bf P}^{ P}(da^{i-1}|b^{i-1}), \hso  i=0, \ldots, n. \label{CIS_9a_new_C4}
\end{align}

The next theorems presents various  generalizations of Theorem~\ref{thm_BAM_2}. \\

\begin{theorem}
\label{cor-ISR_C4}
(Characterization of FTFI capacity of channel class B and transmission costs of class A or B)\\
\noi {\bf Part A.} Suppose the channel distribution is of Class B, that is, ${\bf P}_{B_i|B^{i-1}, A^i}(db_i|b^{i-1}, a^i)=Q_i(db_i| {\cal I}_{i}^{Q} )$, where  ${\cal I}_{i}^{Q}$  is  given by
\begin{align}
  {\cal I}_{i}^{Q}=\{b_{i-M}^{i-1}, a_i\}, \hso i=0, \ldots, n \label{cor-ISR_29_C4}
\end{align}
Then   the maximization in (\ref{cor-ISR_25a_c_C4})  over  ${\cal P}_{[0,n]}$ occurs in the subset 
\begin{align}
\sr{\circ}{\cal P}_{[0,n]}^{B.M} \tri     \{   P_i(da_i | a^{i-1}, b^{i-1})= \pi_i^M(da_i|b_{i-M}^{i-1})-a.a. (a^{i-1}, b^{i-1}): i=0, 1, \ldots, n\} \subset {\cal P}_{[0,n]} .
\end{align}
 and the characterization of the FTFI feedback  capacity is given by the following expression. 
\begin{align}
{C}_{A^n \rar B^n}^{FB,B.M} 
=& \sup_{\big\{\pi_i^M(da_i | b_{i-M}^{i-1}) \in {\cal M}({\mb A}_i) : i=0,\ldots,n\big\}} \sum_{i=0}^n {\bf E}^{ \pi^M}\left\{
\log\Big(\frac{dQ_i(\cdot|B_{i-M}^{i-1},A_i)}{v_i^{ \pi^M}(\cdot|B_{i-M}^{i-1})}(B_i)\Big)
\right\}  \label{cor-ISR_25a_a_c_C4}\\
 \equiv &\sup_{\big\{\pi_i^M(da_i |b_{i-M}^{i-1}) \in {\cal M}({\mb A}_i) : i=0,\ldots,n\big\}} \sum_{i=0}^n I(A_i; B_i|B_{i-M}^{i-1})
\end{align}
where 
\begin{align}
  v_i^{\pi^M}(db_i | b_{i-M}^{i-1}) =& \int_{  {\mb A}_i }   Q_i(db_i |b_{i-M}^{i-1}, a_i) \otimes   {\pi}_i^M(da_i | b_{i-M}^{i-1}), \hso i=0, \ldots, n, \label{cor-ISR_31_c_C4}\\
  {\bf P}^{\pi^M}(da_i,  d b_{i-M}^i)=& Q_i(db_i|b_{i-M}^{i-1}, a_i) \otimes \pi_i^M(da_i| b_{i-M}^{i-1}) \otimes {\bf P}^{\pi^M}( b_{i-M}^{i-1}),  \hso  i=0, \ldots, n. \label{cor-ISR_32_c_C4} 
 \end{align}
{\bf Part B.} Suppose the channel distribution is of Class B as in Part A, and the maximization in (\ref{cor-ISR_25a_c_C4}) is over  ${\cal P}_{0,n}(\kappa)$, defined with respect  to   transmission cost $\gamma_i(\cdot, \cdot)$, which is measurable with respect  to  ${\cal I}_{i}^{\gamma}$ given by
\begin{align}
 {\cal I}_{i}^{\gamma}=\{a_i,  b_{i-K}^{i-1}\}, \hso i=0, \ldots, n \label{cor-ISR_29_c_C4}
\end{align}
and the analogue of  Lemma~\ref{thm_BAM_2}, { Part B}, (b) holds.\\ 
The maximization in (\ref{cor-ISR_25a_c_C4}) over $\big\{P_i(da_i|a^{i-1},b^{i-1}), i=0, \ldots, n\big\} \in{\cal P}_{0,n}(\kappa)$  occurs in the subset 
\begin{align}
\sr{\circ}{\cal P}_{[0,n]}^{B.J}(\kappa) \tri     \Big\{   P_i(da_i | & a^{i-1}, b^{i-1})= \pi_i^J(da_i|b_{i-J}^{i-1})-a.a. (a^{i-1}, b^{i-1}),  i=0, 1, \ldots, n:  \nonumber \\
& \frac{1}{n+1} {\bf E}^{\pi^J}\Big( c_{0, n}(A^n, B^{n-1})\Big) \leq \kappa  \Big\} \subset {\cal P}_{[0,n]}(\kappa), \hst J \tri \max\{M, K\} \label{cor-ISR_29_cc_C4_b2}
\end{align}
and the characterization of FTFI capacity is given by the following expression.
\begin{align}
{C}_{A^n \rar B^n}^{FB,B.J}(\kappa) 
= \sup_{\big\{\pi_i^J(da_i | b_{i-J}^{i-1}) \in {\cal M}({\mb A}_i): i=0,\ldots,n   \big\}\in \sr{\circ}{\cal P}_{[0,n]}^{B.J}(\kappa)    } \sum_{i=0}^n {\bf E}^{ \pi^J}\left\{
\log\Big(\frac{dQ_i(\cdot|B_{i-M}^{i-1},A_i)}{d\nu_i^{{ \pi}^J}(\cdot|B_{i-J}^{i-1})}(B_i)\Big)
\right\}  \label{cor-ISR_B.2}
\end{align}
where 
\begin{align}
{\bf P}^{\pi^J}(db_{i-J}^i, da_i) =&Q_i(db_i|b_{i-M}^{i-1}, a_i)\otimes \pi_i^J(da_i|b_{i-J}^{i-1}) 
\otimes {\bf P}^{\pi^J}(db_{i-J}^{i-1}),  \hso i=0,1, \ldots, n,\label{cor-ISR_B.2_1} \\
\nu_i^{\pi^J}(db_i|b_{i-J}^{i-1}) =&\int_{{\mb A}_i} Q_i(db_i|b_{i-M}^{i-1}, a_i)\otimes \pi_i^J(da_i|b_{i-J}^{i-1}),  \hso i=0,1, \ldots, n.\label{cor-ISR_B.2_2}
\end{align}
{\bf Part C.}  Suppose the channel distribution is of Class B as in Part A, and the maximization in (\ref{cor-ISR_25a_c_C4}) is over ${\cal P}_{0,n}(\kappa)$, defined with respect to a  transmission cost of Class A, $\{\gamma_i^A(a_i, b^{i-1}): i=0, \ldots, n\}$, and the analogue of  Lemma~\ref{thm_BAM_2}, { Part B}, (b) holds.\\ 
The maximization in (\ref{cor-ISR_25a_c_C4}) over $\big\{P_i(da_i|a^{i-1},b^{i-1}), i=0, \ldots, n\big\} \in{\cal P}_{[0,n]}(\kappa)$  occurs in $\overline{\cal P}_{[0,n]}^A \bigcap{\cal P}_{[0,n]} $.  
\end{theorem}
\begin{proof} Part A. The derivation is based on the results obtained thus far, using Step 1 and Step 2 of the Two-Step procedure.    By Step 2 of the Two-Step Procedure, repeating the derivation of  Lemma~\ref{thm_BAM_2}, if necessary,  it can be shown that the optimal channel input distribution occurs in  $\sr{\circ}{\cal P}_{[0,n]}^{B.M}$. \\{ Part B.} The case with transmission cost is shown by applying Lagrange duality to define the unconstraint problem, and then noticing that  the upper bound resulting from the variational equality of directed information is achievable, provided the arbitrary distribution (analogue of (\ref{gioler49bb})) is chosen so that ${V}_i(db_i|b^{i-1})= {v}^{\pi^J}(db_i|b_{i-J}^{i-1})-a.a.b^{i-1}, i=0,1, \ldots, n, J\tri \max\{M, K\}$, establishing (\ref{cor-ISR_29_cc_C4_b2}).  
\\{ Part C.} Since the transmission cost is of Class  A, $\{\gamma_i^A(a_i, b^{i-1}): i=0, \ldots, n\}$, and the Channel distribution is of Class B, the statement of Theorem~\ref{thm-ISR}, Part C holds, hence the set of all channel input distributions, which maximize directed information $I(A^n\rar B^n)= \sum_{i=0}^n \int 
\log\Big(\frac{dQ_i(\cdot|b_{i-M}^{i-1}, a_i)}{d\Pi_i^{ P}(\cdot|b^{i-1})}(b_i)\Big) {\bf P}^P(db_{i-M}^i, da_i)$,  occur in the set $\overline{\cal P}_{[0,n]}^A$. Consequently, the channel output conditional  probabilities are given by 
\begin{align}
\Pi_i^{ P} (db_i| b^{i-1})=&  \int_{{\mb A}^i} Q_i(db_i|b_{i-M}^{i-1}, a_i)\otimes {\pi}(da_i|b^{i-1})\equiv \Pi_i^{\pi} (db_i| b^{i-1}), \hso  i=0, \ldots, n \label{CIS_9a_new_C4_new}
\end{align}
However, any attempt to apply the variational equality of directed information, as done in Lemma~\ref{thm_BAM_2}, to derive  upper bounds on the corresponding directed information,  which are achievable over arbitrary distributions, $\{{V}_{i}(db_i|b^{i-1})\in{\cal M}({\mb  B}_i): i=0, \ldots, n\}$, which satisfy  conditional independence  condition
\begin{align}
{V}_{i}(db_i|b^{i-1}) = \overline{V}_{i}(db_i|b_{i-L}^{i-1})-a.a.b^{i-1}, \hso \mbox{for any finite nonnegative $L$},   \hso  
i=0,1,\ldots,n \label{gioler49bb_BBB}
\end{align} 
will fail. This is  because the transmission cost of Class A, depends, for each $i$, on the entire past output symbols $\{b^{i-1}\}$, and hence the maximization step, using stochastic optimal control,  over channel input distributions from the set $\overline{\cal P}_{[0,n]}^A$, satisfying the average transmission cost constraint cannot occur is a smaller subset, i.e., recall Feature 1 of the discussion on classical stochastic optimal control. This completes the prove.
\end{proof}

\subsection{Implications on Dynamic Programming Recursion}
In this section, the implications of the information structures of the optimal channel input distributions, are discussed in the context of dynamic programming. 

{\bf Channels Class B and Transmission Costs Class B.}  Consider  a  channel distribution and transmission cost function, both  of Class B, given by ${\bf P}_{B_i|B^{i-1}, A^i}(db_i|b^{i-1}, a^i)=Q_i(db_i|b_{i-M}^{i-1},a_i)-a.a.(b^{i-1}, a^i),\gamma^{B.K}_{i}(a_i,b_{i-K}^{i-1}), i=0, \ldots, n$. \\
Since the output  process $\{B_i: i=0, \ldots, n\}$ is $J=\max\{M,K\}-$order Markov, i.e., (\ref{cor-ISR_B.2_2}), holds, 
and  the characterization of FTFI capacity is given by (\ref{cor-ISR_B.2}), the optimization over  $\sr{\circ}{\cal P}_{[0,n]}^{B.J }(\kappa )$ can be solved via  dynamic programming, as follows.  \\
Let $C_t^{B.J}: {\mb B}_{t-J}^{t-1} \longmapsto {\mb R}$ denote the cost-to-go corresponding to (\ref{cor-ISR_B.2}) from time ``$t$'' to the terminal time   ``$n$'' given the values of the output and input $B_{t-J}^{t-1}=b_{t-J}^{t-1}$, defined as follows.  \\
\begin{align}
C_t^{B.J}(b_{t-J}^{t-1}) =& \sup_{  \pi_i^J(da_i|b_{i-J}^{i-1}):\; i=t, t+1, \ldots, n}\hso {\bf E}^{\pi^{J}}\Big\{ \sum_{i=t}^n \Big[\int_{ {\mb B}_i }    \log\Big(\frac{dQ_i(\cdot| b_{i-M}^{i-1}, A_i)}{d\nu_i^{\pi^J}(\cdot| b_{i-J}^{i-1})}(b_i)\Big)   Q_i(db_i| b_{i-M}^{i-1}, A_i)- s \gamma_i^{B.K}(A_i, b_{i-K}^{i-1}) \Big]\Big| B_{t-J}^{t-1}=b_{t-J}^{t-1}   \Big\} \label{NCM-B.2-DP1_CC}
  \end{align}
  where $s \in [0, \infty)$ is the Lagrange and the term $(n+1)\kappa$ is not included.\\
Then the  cost-to-go satisfies the following dynamic programming  recursions. 
\begin{align}
C_n^{B.J}(b_{n-J}^{n-1}) =& \sup_{  \pi_n^J(da_n|b_{n-J}^{n-1})} \Big\{\int_{{\mb A}_n \times {\mb B}_n }    \log\Big(\frac{dQ_n(\cdot| b_{n-M}^{n-1}, a_n)}{d\nu_n^{\pi^J}(\cdot| b_{n-J}^{n-1})}(b_n)\Big)   Q_n(db_n| b_{n-M}^{n-1}, a_n)  \otimes \pi_n^J(da_n|b_{n-J}^{n-1})   \nonumber \\
  &- s \int_{{\mb A}_n} \gamma_n^{B.K}(a_n, b_{n-K}^{n-1}) \pi_n^J(da_n|b_{n-J}^{n-1})
  \Big\},  \label{NCM-B.2-DP2} \\
C_t^{B.J}(b_{t-J}^{t-1}) =& \sup_{  \pi_t^J(da_t|b_{t-J}^{t-1})} \Big\{\int_{{\mb A}_t \times {\mb B}_t }    \log\Big(\frac{dQ_t(\cdot| b_{t-M}^{t-1}, a_t)}{d\nu_t^{\pi^J}(\cdot| b_{t-J}^{t-1})}(b_t)\Big)   Q_t(db_t| b_{t-M}^{t-1}, a_t)  \otimes \pi_t^J(da_t|b_{t-J}^{t-1})  \nonumber \\
&  - s \int_{{\mb A}_t} \gamma_t^{B.K}(a_t, b_{t-K}^{t-1}) \pi_t^J(da_t|b_{t-J}^{t-1})+ \int_{{\mb A}_t \times {\mathbb B}_t }   C_{t+1}^{B.J}(b_{t+1-J}^t) Q_t(db_t| b_{t-M}^{t-1}, a_t)  \otimes \pi_t^J(da_t|b_{t-J}^{t-1})    \Big\}, \hso t=n-1,  \ldots, 0. \label{NCM-B.2-DP1}
  \end{align}
The characterization of the FTFI capacity is expressed via the $C_0^{B.J}(b_{-J}^{-1})$ and the fixed distribution $\mu_{B_{-J}^{-1}}(db_{-J}^{-1})$ by 
\begin{align}
C_{A^n \rar B^n}^{FB, B.J}(\kappa) =\inf_{s\geq 0}\Big\{\int_{{\mb B}_{-J}^{-1}} C_0^{B.J}(b_{-J}^{-1})\mu_{B_{-J}^{-1}}(db_{-J}^{-1})-(n+1)s \kappa\Big\}.
\end{align}
It is obvious from the above recursions that, that the information structure, $\{B_{t-J}^{t-1}: t=0,\ldots, n\}$, of the control object, namely, $\big\{\pi_t^J(da_t|b_{t-J}^{t-1}): i=0, \ldots, n\big\}$, induces transition probabilities of the controlled object, $\big\{\nu_i^{\pi^J}(db_t| b_{t-J}^{t-1}): t=0, \ldots, n\big\}$ which are $J-$order Markov, resulting in a significant reduction in computational complexity of the above dynamic programming recursions. This is one of the fundamental differences,  compared to other dynamic programming algorithms proposed in the literature, which do not investigate the impact of  information structures, on the characterization of FTFI capacity, and by extension of feedback capacity. 

{\bf Special Case-Unit Memory Channel Output (UMCO) $M=K=1$.} Since in this case, $J=1$, the corresponding dynamic programming recursions are degenerate versions of (\ref{NCM-B.2-DP2}), (\ref{NCM-B.2-DP1}), obtained by setting   $K=M=1, J=1$, i.e., 
\begin{align}
Q_t(db_t|b_{t-M}^{i-1}, a_i) \longmapsto & Q_t(db_t|b_{i-1}, a_i), \hso \gamma_t^{B.K}(a_t, b_{t-K}^{t-1} \longmapsto \gamma_t^{B.1}(a_t, b_{t-1}), \nonumber \\
& \pi_t^{J}(da_t|b_{t-J}^{t-1})\longmapsto \pi_t^{1}(da_t|b_{t-1}), \hso C_t^{B.J}(b_{t-J}^{t-1}) \longmapsto C_t^{B.1}(b_{t-1}), \hso t=n, \ldots, 0.
\end{align}

 This degenerate dynamic programming recursion is the simplest, because the joint process $\big\{(A_i, B_i): i=0, \ldots, n\big\}$ is jointly Markov (first-order), and the channel input conditional distribution is $\big\{\pi_i^1(da|b_{i-1}): i=0, \ldots, n\big\}$. At each time $t$, the dynamic programming recursion involves 3-letters, $\{b_t, a_t, b_{t-1}\}$, where $b_{t-1}$ is fixed, for $t=n, n-1, \ldots, 0$. 

It is noted that,  for the case of finite alphabet spaces $\{({\mb A}_i, {\mb B}_i): i=0, \ldots, n\}$, the UMCO without transmission cost constraints is analyzed extensively by Chen and Berger in \cite{chen-berger2005} (and it is discussed by Berger in \cite{berger_shannon_lecture}), under the assumption the optimal channel input conditional distribution satisfies conditional independence $P(da_i|a^{i-1}, b^{i-1})=\pi_i^1(da|b_{i-1}), i=0, \ldots, n$, which then implies $\big\{(A_i, B_i): i=0, \ldots, n\big\}$ is jointly Markov, and hence the corresponding characrerization of FTFI capacity is given by $C_{A^n \rar B^n}^{FB,B.1}=\sup_{  \pi_i(da_i|b_{i-1}): i=0, \ldots, n}\sum_{i=0}^n I(A_i;B_i|B_{i-1})$.  To the best of the authors knowledge, the current paper, provides, for the first time, a derivation of the fundamental assumptions, upon which the  results derived in \cite{chen-berger2005}, are based on.

The main point to be made regarding this section, is  that the information structure of the optimal channel input distribution maximizing directed information, can be obtained for many different classes of channels with memory, and many different classes of transmission cost functions, and that the corresponding conditional independence properties of optimal channel input distributions and characterizations of the FTFI capacity  are generalizations of (\ref{CI_DMC}) and the two-letter capacity formulae of Shannon, corresponding to   memoryless  channels.\\
These structural properties of optimal channel input conditional distributions simplify the computation of the corresponding  FTFI capacity characterization, and its per unit time limiting versions, the characterization of feedback capacity. \\
 
\begin{remark}(Generalizations to channels with memory on past channel inputs)\\
The reader may verify that the methodology developed this paper, to identify the information structures of optimal channel input distributions satisfying conditional independence, is also  applicable  to   general channel distributions and transmission cost functions of the form,  
\begin{align}
\big\{ {\bf P}_{B_i|B_{i-M}^{i-1}, A_{i-L}^i}(db_i|b_{i-M}^{i-1}, A_{i-L}^i): i=0, \ldots, n\big\}, \hso \big\{\gamma_{i}(A_{i-N}^i, b_{i-K}^{i-1}): i=0, \ldots, n\big\}
\end{align}  
($\{N, L\}$ are nonnegative integers)  which   depend on past source symbols. However, such generalizations of the structural properties of optimal channel input conditional distributions,   are beyond the scope of this paper. 
\end{remark}

\section{Application Example: Gaussian Linear Channel Model}\label{example}
The objective of the application example discussed below  is to illustrate the role of the information structures of optimal channel distributions in deriving  closed form expressions for feedback capacity, capacity without feedback, and to illustrate  hidden aspects of  feedback. 


Consider the time-invariant version of a Gaussian-Linear Channel Model (G-LCM) of class B with transmission cost  of Class B,  defined by 
\begin{align}
&B_i   = C\; B_{i-1} +D \; A_i + V_{i},\hso B_{-1}=b_{-1},  \hso i= 0, \ldots, n, \label{LCM-A.1_a_TI} \\   
&\frac{1}{n+1} \sum_{i=0}^n {\bf E} \Big\{ \langle A_i, R A_i \rangle + \langle B_{i-1}, Q B_{i-1} \rangle \Big\}\leq \kappa, \hso   R \in {\mb S}_{++}^{q\times q}, \; Q \in {\mb S}_{+}^{p \times p} \label{LCM-A.1_aa_TI}\\
&{\bf P}_{V_i|V^{i-1}, A^i}(dv_i|v^{i-1}, a^{i})={\bf P}_{V_i}(dv_i)-a.a.(v^{i-1}, a^i), \hso V_i \sim N(0, K_{V_i}),\hso  K_{V_i}=K_V \in S_{++}^{p\times p}, \hso i= 0, \ldots, n
\end{align} 
where $\langle\cdot, \cdot \rangle$ denotes inner product of vectors, ${\mb S}_{++}^{q\times q}$ denotes the set of symmetric positive definite $q$ by $q$ matrices and  ${\mb S}_{+}^{q \times q}$ the set of positive semi-definite matrices. The initial state $b_{-1}$ is known to the encoder and decoder. \\
From Theorem~\ref{cor-ISR_C4}, the optimal channel input distribution maximizing directed information satisfies conditional independence  $P_i(da_i|a^{i-1}, b^{i-1})=\pi_i(da_i|b_{i-1})-a.a.(a^{i-1}, b^{i-1}), i=0, \ldots, n$.  Moreover,  it can be easily shown, i.e., using  the maximum entropy properties of Gaussian distributions, or by solving the corresponding dynamic programming recursion of the FTFI capacity, that the optimal distribution satisfying the average cost constraint is Gaussian, i.e., $\pi_i(da_i|b_{i-1})\equiv \pi_i^g(da_i|b_{i-1}), i=0, \ldots, n$, which then implies the joint process is also Gaussian, i.e, $(A_i, B_i, V_i)\equiv (A_i^g, B_i^g, V_i), i=0, \ldots, n$, provided of course that the RV $B_{-1}$ is also Gaussian.\\
Any such optimal channel input conditional distribution can be realized via an orthogonal decomposition as follows.   
\begin{align}
&A_i^g \tri   U_{i}^g + Z_i^g, \hso U_{i}^g = g_i^{B.1}(B_{i-1}^g)\equiv \Gamma_{i,i-1} B_{i-1}^g, \hso i=0, \ldots, n,  \label{DP_UMCO_C1_TI}\\
& Z_i^g  \:\:  \mbox{is  independent of}\:\:  \Big(A^{g, i-1}, B^{g,i-1}\Big), \; Z^{g,i} \hso \mbox{is independent of} \hso V^i, i=0, \ldots, n, \label{new_11}\\
& \Big\{Z_i^g \sim N(0, K_{Z_i}) : i=0,1, \ldots, n\Big\} \: \: \mbox{is an independent Gaussian process} \label{new_12}
\end{align}
Moreover, substituting (\ref{DP_UMCO_C1_TI}) into   (\ref{LCM-A.1_a_TI}) the channel output process is given by 
\bea
B_i^g= C B_{i-1}^g+ D U_{i}^g + D Z_i^g + V_i, \hso i=0, \ldots, n. \label{DP_UMCO_C2_TI}
\eea
The corresponding characterization of the FTFI capacity is the following. 
\begin{align}
{C}_{A^n \rar B^n}^{FB,B.1} (\kappa) &
 = \sup_{\big\{(g_i^{B.1}(\cdot), K_{Z_i}),   i=0,\ldots,n   \big\} \in  {\cal E}_{[0,n]}^{B.1}(\kappa)   }   \sum_{i=0}^n H(B_i^g|B_{i-1}^g) - H(V^n)   \\ 
{\cal E}_{[0,n]}^{B.1}(\kappa)   &\tri \Big\{(g_i^{B.1}(b_{i-1}), K_{Z_i}), i=0, \ldots,n:  \frac{1}{n+1}{\bf E}^{g^{B.1}} \Big( \sum_{i=0}^n     \Big[\langle A_i^g, R A_i^g \rangle + \langle B_{i-1}^g, Q B_{i-1}^g \rangle \Big] \Big) \leq \kappa \Big\}.  \label{fea_IH}
\end{align}
Suppose the pair $(g_i^{B.1}(\cdot), K_{Z_i})\equiv (g^{B.1}(\cdot), K_{Z}), i=0, \ldots,n$, i.e.,  is restricted  to time-invariant,   and consider 
the per unit time limiting version of the characterization of FTFI capacity, defined by $C_{A^\infty \rar B^\infty}^{FB} (\kappa)\tri \lim_{n \longrightarrow \infty} \frac{1}{n+1}\tri C_{A^n \rar B^n}^{FB,B.1}(\kappa)$.  Then one method to obtain $C_{A^\infty \rar B^\infty}^{FB, B.1} (\kappa)$ is  via dynamic programming as follows \cite{hernandezlerma-lasserre1996,kumar-varayia1986}. \\
Under appropriate conditions expressed in terms of the matrices $\{C, D, R, Q\}$ \cite{kumar-varayia1986},   there  exists a pair $\Big(J^{B.1,*},C^{B.1}(b)\Big), J^{B.1,*} \in {\mathbb R}$, $C^{B.1}: {\mathbb R}^p \longmapsto {\mathbb R}$, which  satisfies the following dynamic programming equation corresponding to  ${C}_{A^\infty \rar B^\infty}^{FB,B.1} (\kappa)$.   
\begin{align}
J^{B.1,*} + C^{B.1}(b) = &\sup_{ (u,K_{Z})\in {\mathbb R}^q\times {\mb S}_{+}^{q\times q}  }  \Bigg\{ \frac{1}{2} \log \frac{ | D K_{Z} D^T+K_{V}|}{|K_{V}|} - tr\Big(s R K_{Z}\Big)  + s \kappa   - s \Big[ \langle u, R u\rangle + \langle b, Q b \rangle  \Big] \nonumber \\
&  + {\bf E}^{g^{B.1}}\Big\{  C^{B.1}(B_{0}^g)       \Big| B_{-1}^g=b\Big\}  \Bigg\} \label{DP_IH}
\end{align}
where    $s\equiv s(\kappa) \geq 0$ is the Lagrange multiplier associated with the average transmission cost constraint. \\
It can be verified that the solution to the dynamic programming is  given by
\begin{align}
C^{B.1}(b)=&-\langle b, P b\rangle, \label{DP-UMCO_C13_IH_aa} \\
J^{B.1,*}= &  \sup_{ K_{Z}\in  {\mb S}_+^{q\times q}}\Big\{  \frac{1}{2} \log \frac{ | D K_{Z} D^T+K_{V}|}{|K_{V}|}+s\kappa- tr\Big(s\; R K_{Z}\Big)- tr \Big(P\Big[D K_{Z}D^T+ K_{V}\Big]\Big)\Big\}  \label{DP-UMCO_C13_IH}  
\end{align}
and that the  optimal stationary policy $g^{B.1,*}(\cdot)$  is given by
\begin{align}
&g^{B.1,*}(b)= \Gamma^*b, \\
&\Gamma^*= - \Big(D^T PD +s R\Big)^{-1}D^T P C, \label{opt_gam}\\
&P=C^T P C+s Q-C^T P D\Big(D^T P D+s R\Big)^{-1} \Big(C^T P D\Big)^T. \label{DP-UMCO_C12_IH_new}\\
&spec\Big(C+D \Gamma^*\Big) \subset {\mathbb D}_o. \label{st_1}
\end{align}
where $spec(A) \subset {\mathbb C}$ denotes the spectrum of a matrix $A \in {\mathbb R}^{q \times q}$, i.e.,  the set of all its eigenvalues, and ${\mathbb D}_o \tri \big\{c \in {\mathbb C}: |c| <1\big\}$ denotes  the open unit disc of the space of complex number ${\mathbb C}$. Note that (\ref{DP-UMCO_C12_IH_new}) is the well-known Riccati equation of Gaussian Linear Quadratic stochastic optimal control problems, and $\Gamma^*\equiv \Gamma^*(P)$ corresponds to the positive semi-definite solution $P\succeq 0$ of the Riccati equation \cite{kumar-varayia1986}, satisfying (\ref{st_1}), to   ensure the eigenvalues of the closed loop channel output recursion  (\ref{DP_UMCO_C2_TI}), i.e., corresponding to  $U_i^{g,*}=\Gamma^* B_{i-1}^g, i=0, \ldots, $, are within the open unit disc ${\mathbb D}_o$.  \\
The optimal covariance $K_{Z}^*$ is determined  from the optimization problem \eqref{DP-UMCO_C13_IH} and the Lagrange multiplier, for a given $\kappa$, i.e.,  $s\equiv s(\kappa)$ is found from the average constraint. The feedback capacity is given by  
\bea
{C}_{A^\infty \rar B^\infty}^{FB, B.1} (\kappa)= J^{B.1,*}\equiv J^{B.1,*}(\kappa), \hso \kappa \in [\kappa_{min}, \infty) \subset [0, \infty).   \label{opt_s}
  \eea    
The analysis of the Multiple Input Multiple Output       G-LCM is done in \cite{chakouloy2016}, and requires extensive investigation of properties of solutions to Riccati equations. The complete solution for the scalar G-LCM is presented below, to illustrate additional features, which are not given in \cite{chakouloy2016}. \\ 
{\bf Scalar Case, $p=q=1, D=1$.}  By solving \eqref{DP-UMCO_C12_IH_new}, the  positive semi-definite solution of the Riccati equations is given by  
\begin{align}
P=\frac{s\, \left(Q - R + C^2\, R + F\right)}{2}, \hso  \hso
F =\sqrt{\left(R[C-1]^2+Q\right)\, \left(R[C+1]^2 +Q \right)}\label{p2} 
\end{align}
 The value of $\Gamma^*\equiv \Gamma^*(P)$ is obtained by substituting the positive semi-definite   solution of the Riccati equation in \eqref{opt_gam}, to obtain
\bea
\Gamma^*=-\frac{C\, \left(Q - R + C^2\, R + F\right)}{Q + R + C^2\, R + F}, \hso |C+\Gamma^*| <1. \label{gam_star}
\eea
This is valid irrespectively of whether $C$ is stable, i.e., $|C|<1$ or unstable, i.e., $|C|\geq 1$, and includes, as we show shortly, the special case $Q=C=0$, i.e., corresponding to the memoryless channel.  \\ 
The optimal covariance $K_{Z}^*$, is obtained by solving the optimization problem \eqref{DP-UMCO_C13_IH}, which  gives
\bea
K_{Z}^*=\frac{1}{s\, \left(Q + R + C^2\, R + F\right)} - K_{V}\geq 0 \label{opt_kz}
\eea
while   the Lagrange multiplier, $s$, is found from the average constraint or by performing the infimum over $s\geq 0$ of $J^{B.1, *}$ evaluated at $(P, K_Z^*)$ given  by 
\eqref{DP-UMCO_C13_IH}, to obtain 
\bea
s\equiv s(\kappa)=\frac{1}{2(\kappa +  K_{V}R)}.\label{opt_s2}
\eea
The minimum power $\kappa$ required so that  the optimal covariance is non-negative, i.e.,  $K_{Z}^*\geq 0$ is found by substituting \eqref{opt_s2} in \eqref{opt_kz} and it is given by 
\bea
\kappa_{min}=\frac{K_{V}\, \left(Q - R + C^2\, R + F\right)}{2} \geq 0.
\eea  
Finally by substituting \eqref{opt_kz} and \eqref{opt_s2} in \eqref{DP-UMCO_C13_IH}, the following expression of the feedback capacity is obtained.

\bea
{C}_{A^\infty \rar B^\infty}^{FB, B.1}  (\kappa)= \left\{ \begin{array}{llll}  0 & \mbox{if}  & \kappa \in [0,\kappa_{min})   \\
 \frac{1}{2}\mathrm{log}\!\left(\frac{2(\kappa +  K_{V}\, R)}{K_{V}\, \left(Q + R + C^2\, R + F\right)}\right)
 & \mbox{if}   &  \kappa \in [\kappa_{min}, \infty) . 
\end{array} \right.  \label{dual_CAP}
\eea
It can be verified that the value of $\kappa_{min}$ depends on the values of $Q=0$, $Q>0$ and $|C|<1$, $|C|\geq 1$. For $Q=C=0$ the feedback capacity ${C}_{A^\infty \rar B^\infty}^{FB, B.1}  (\kappa)$ degenerates to that of memoryless channels, as expected. \\
Next, specific special cases are analyzed to gain additional insight on the dependence of capacity on $|C|<1$ and $|C|\geq 1$.\\
{\it (a) Suppose $Q=R=1$.}  Then 
\bea
F=\sqrt{C^4+4}, \hso s=\frac{1}{2(\kappa +  K_{V})}, \hso K_{Z}^*=\frac{2\, \kappa-K_{V}\,  \sqrt{C^4 + 4}  - C^2\, K_{V}}{\sqrt{C^4 + 4} + C^2 + 2}.
\eea
The feedback capacity is given by
\bea
{C}_{A^\infty \rar B^\infty}^{FB, B.1}  (\kappa)= \left\{ \begin{array}{llll}  0 & \mbox{if}  & \kappa \in [0,\kappa_{min})   \\
 \frac{1}{2}\mathrm{log}\!\left(\frac{2(\kappa +  K_{V})}{K_{V}\, \left(\sqrt{C^4 + 4} + C^2 + 2\right)}\right)
 & \mbox{if}   &  \kappa \in [\kappa_{min}, \infty) 
\end{array} \right.  \label{dual_CAP1}
\eea
where $
\kappa_{min}=\frac{K_{V}\, \left(\sqrt{C^4 + 4} + C^2\right)}{2}.$ Clearly, if $C=0$ then $\kappa_{min}=K_V$ and this is attributed to the fact that, the power transfer of the channel output process is reflected in the average power constraint, i.e., $Q=1$.  
\begin{figure}
  \centering
    \includegraphics[scale=0.7 ]{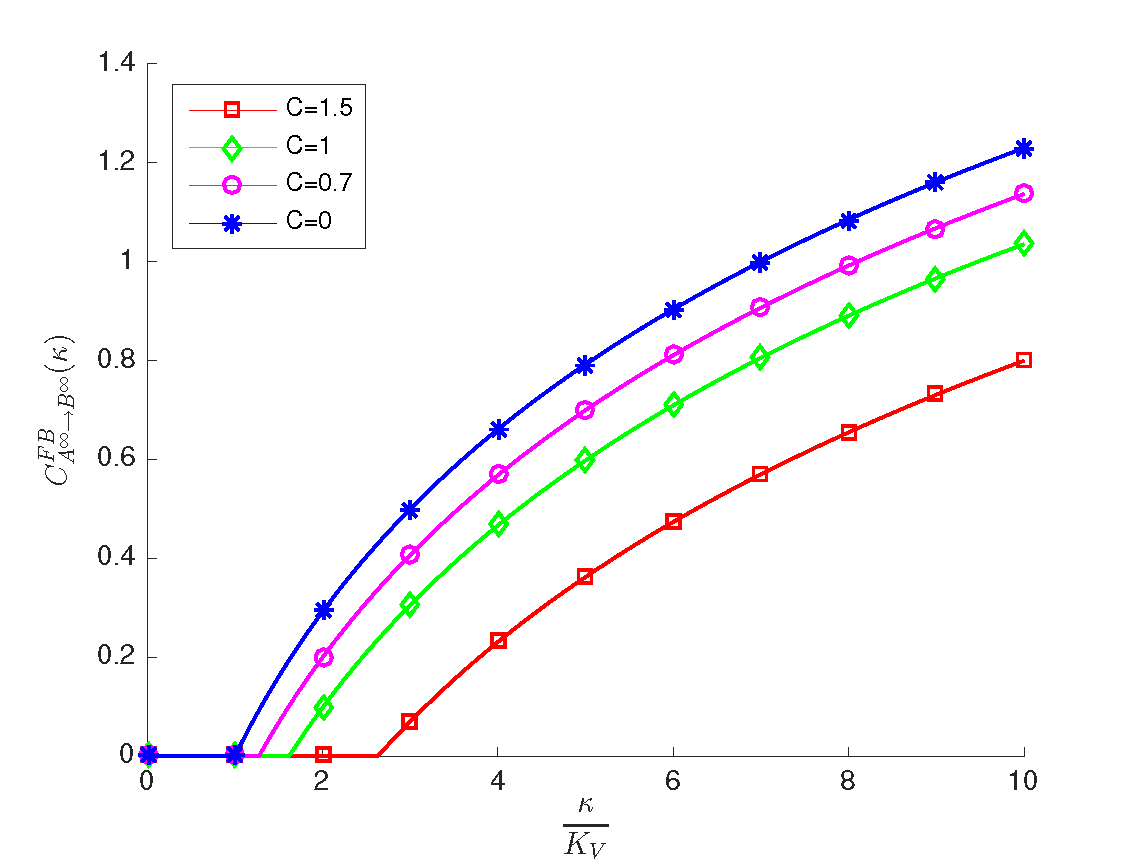}
      \caption{Feedback capacity for the scalar Linear-Gaussian Channel Model ($Q=1, R=1, D=1$).}      \label{graph_cap}
\end{figure}

The feedback capacity for  $D=Q=R=1$ is illustrated in Fig.\ref{graph_cap}, for stable and unstable  values of the parameter $C$.  It illustrates that there is a minimum value $\kappa_{min}>0$, because the transmission cost function includes the power of the channel output process, and because of this, part of the power $\kappa$ is transfer to the channel output.   \\
\begin{figure}
  \centering
    \includegraphics[scale=0.7 ]{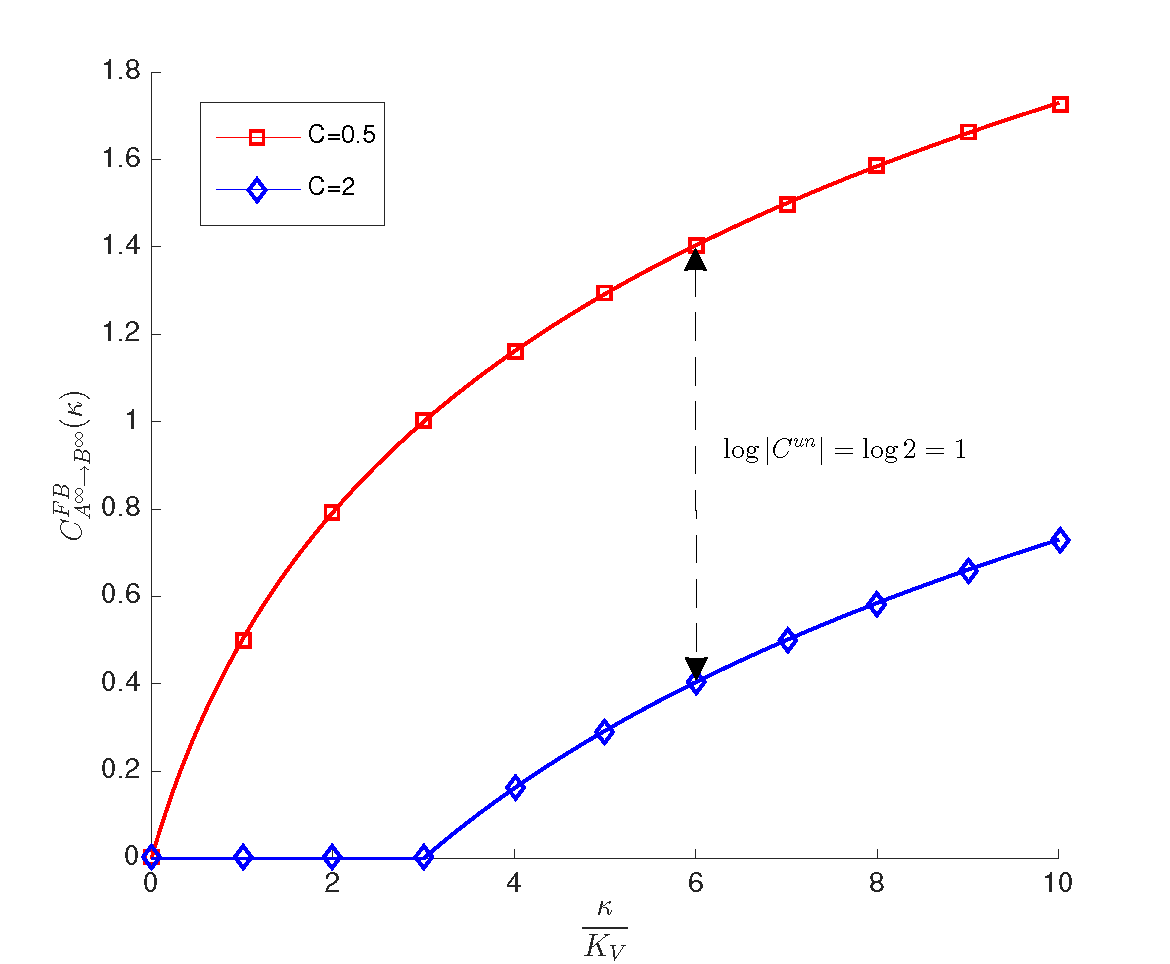}
      \caption{Feedback capacity for the scalar Linear-Gaussian Channel Model ($Q=0, R=1, D=1$), where $C^{un}$ denotes the value of $C$ of the unstable channel (C=2).}      \label{comp}
\end{figure}
{\it (b) Suppose $D=R=1, Q=0$.} Then the cost constraint is independent of past  channel output symbols, and $F=C^2-1$ which yields 
\bea
P= \left\{ \begin{array}{llll}  0 & \mbox{if}  & C|< 1   \\
 C^2-1
 & \mbox{if}   &  |C|\geq 1. 
\end{array} \right.  \label{P_sc}
\eea
 The optimal strategy which achieves feedback capacity ${C}_{A^\infty \rar B^\infty}^{FB, B.1}(\kappa)$ is given by
\bea
\Big(\Gamma^*, K_Z^*\Big)= \left\{ \begin{array}{llll} (0, \kappa), & \kappa \in [0, \infty) & \mbox{if}\hso  |C|< 1 \\ \\
\Big(-\frac{C^2-1}{C},\frac{\kappa + K_V (1-C^2)}{C^2 }\Big), & \kappa \in [\kappa_{min}, \infty), \hso \kappa_{min} \tri   {(C^2-1)K_V} & \mbox{if} \hso |C|\geq 1 \\ \\
\Big(-\frac{C^2-1}{C},0\Big), & \kappa \in [0,\kappa_{min}], \hso  & \mbox{if} \hso |C|\geq 1.
 \end{array} \right. \label{dual_CAP_CI}
\eea
Let ${C}_{A^\infty \rar B^\infty}^{FB,Stable}  (\kappa)$ denote the feedback capacity if the channel is stable, i.e.,  $|C|< 1 $ and ${C}_{A^\infty \rar B^\infty}^{FB,Unstable}(\kappa)$ denote the feedback capacity if the channel is unstable, i.e., $|C|\geq 1$. Then,  the corresponding feedback capacity is given by the following expressions.\\

For $|C|< 1 $:
\begin{align}
{C}_{A^\infty \rar B^\infty}^{FB, B.1}  (\kappa)\tri {C}_{A^\infty \rar B^\infty}^{FB,Stable}  (\kappa)=\frac{1}{2} \log\left(1+ \frac{\kappa }{K_{V}}\right), \hso \kappa \in [0, \infty). \label{dual_CAP2_st}
\end{align}
For $|C|\geq 1 $:
\begin{align}
{C}_{A^\infty \rar B^\infty}^{FB, B.1}  (\kappa)\tri {C}_{A^\infty \rar B^\infty}^{FB,Unstable}(\kappa)= \left\{ \begin{array}{lll}  
\frac{1}{2} \log\left(1+ \frac{\kappa }{K_{V}}\right) -\log{\mid}C{\mid}& \mbox{if}  &  \kappa \in [\kappa_{min}, \infty) \\ \\
 0 & \mbox{if}   &  \kappa \in [0, \kappa_{min}]. 
\end{array} \right.  \label{dual_CAP2_un}
\end{align}
 Then it is clear from \eqref{dual_CAP2_st} and \eqref{dual_CAP2_un}, that 
\bea
{C}_{A^\infty \rar B^\infty}^{FB,Unstable}  (\kappa)={C}_{A^\infty \rar B^\infty}^{FB,Stable}  (\kappa)-\log{\mid}C{\mid}, \hso  \hso \kappa \in [\kappa_{min}, \infty).
\eea
Therefore, the rate loss due to the instability of the channel is given by
\bea
\mbox{Rate Loss of Unstable Channels} \tri {C}_{A^\infty \rar B^\infty}^{FB,Stable}  (\kappa)-{C}_{A^\infty \rar B^\infty}^{FB,Unstable}(\kappa)= \left\{ \begin{array}{llll} \frac{1}{2} \log\left(1+ \frac{\kappa }{K_{V}}\right), & \kappa \in [0,\kappa_{min}]   \\ \\
\log{\mid}C{\mid}, & \kappa \in [\kappa_{min}, \infty).
 \end{array} \right. \label{inst}
\eea
The feedback capacity of a stable channel ($C=0.5$) and an unstable channel ($C=2$), is depicted in Fig. \ref{comp}. The dotted arrow denotes the rate loss due to the instability of the channel with parameter $C=2$, which is equal to $\log{2}=1$ bit. It is worth noting that for unstable channels, the feedback capacity is zero, unless the power $\kappa$  exceeds the critical level $\kappa_{min}$. This the minimum power required to stabilize the channel. Beyond this threshold all the remaining power ($\kappa - \kappa_{min}$) is allocated to information transfer.\\ On the other hand, if the channel is stable, i.e., $|C|<1$, since $Q=0$, there is no  emphasis on power transfer of the channel output process, and feedback capacity degenerates to the capacity of memoryless additive Gaussian noise channel, i.e., feedback does not increase capacity compared to that of memoryless channels, as verified from  (\ref{dual_CAP2_st}). This is, however, fundamentally different from the case $|C|<1$ and $Q>0$ discussed in (a).

\section{Conclusion}
Stochastic optimal control theory and a variational equality of directed information are applied, to develop a methodology to identify the  information structures  of optimal channel input conditional distributions, which maximize directed information, for certain classes of channel conditional distributions and transmission cost constraints. The subsets of the maximizing distributions are characterized by conditional independence.

One of the main theorems of this paper states that, for any channel conditional distribution with  finite memory on past channel outputs, subject to  any average transmission cost constraint corresponding to a specific transmission cost function, the information structure of the optimal channel input conditional distribution, which maximizes directed information, is determined by the maximum of the memory of the channel distribution and the functional dependence of the transmission cost function on past channel outputs. 
This theorem provides, for the first time, a direct analogy, in terms of the conditional independence properties of maximizing distributions,  between the characterization of feedback capacity of channels with memory, and Shannon's two-letter characterization of capacity of memoryless channels.  

Whether a similar method, based on stochastic optimal control theory and variational equalities of mutual and directed information,   can be developed for extremum problems of capacity of channels with memory and without feedback, and for general extremum problems of network information theory,  remains, however to be seen.

\bibliographystyle{IEEEtran}
\bibliography{Bibliography_capacity}

\end{document}